\documentclass[a4paper,11pt]{amsart}
\usepackage[a4paper,margin=25mm]{geometry}

\usepackage{amsmath,amsthm,amssymb,verbatim}
\usepackage[T1]{fontenc}
\usepackage{enumerate}
\usepackage{hyperref}
\usepackage{lmodern}
\usepackage[utf8x]{inputenc}
\usepackage{graphicx}

\newcommand{\letters}{\ensuremath{F}}

\newcommand{\lettersup}{\letters_1^{\textnormal{up}}}
\newcommand{\lettersdown}{\letters_1^{\textnormal{down}}}
\newcommand{\letterschoice}{\letters_1^{\textnormal{choice}}}
\newcommand{\algofont}[1]{\textnormal{\selectfont\sffamily#1}}
\newcommand{\algmain}{\algofont{TtoG}}

\newcommand{\alggreedypairs}{\algofont{Greedy2Chains}}

\newcommand{\algpop}{\algofont{Pop}}
\newcommand{\algremblocks}{\algofont{RemCrChains}}
\newcommand{\algremchild}{\algofont{GenPop}}
\newcommand{\RePair}{\algofont{RePair}}
\newcommand{\Mref}[1]{(\hyperref[M#1]{M#1})}
\newcommand{\Mrefall}{{\Mref{1}--\Mref{3}}}

\newcommand{\paircompression}{unary pair compression}
\newcommand{\Paircompression}{Unary pair compression}
\newcommand{\parentchild}{parent-leaf pair}
\newcommand{\algtreechaincomp}{\algofont{TreeChainComp}}

\newcommand{\algtreechildcomp}{\algofont{TreeLeafComp}}
\newcommand{\algtreepaircomp}{\algofont{TreeUnaryComp}}

\newcommand{\grammarsize}{g}
\newcommand{\propersize}{\grammarsize r}
\newcommand{\grammarsizezero}{g_0}
\newcommand{\grammarsizeone}{g_1}
\newcommand{\grammarsizetilde}{\widetilde{g_0}}

\newcommand{\twodots}{\mathinner{\ldotp\ldotp}}
\DeclareMathOperator{\eval}{val}
\DeclareMathOperator{\sk}{sk}
\newcommand{\weight}{\mathsf{w}}

\newcommand{\algradix}{\algofont{RadixSort}}

\newcommand{\mycount}[1]{\ensuremath{\textnormal{\textsl{count}}^{\textnormal{#1}}}}
\newcommand{\mynew}{\textnormal{\textsl{new}}}
\newcommand{\leftl}[1][(a)]{\ensuremath{\textnormal{\textsl{left}}#1}}

\newcommand{\rightl}[1][(a)]{\ensuremath{\textnormal{\textsl{right}}#1}}
\newcommand{\makeset}[2]{\ensuremath{ \{ #1 \: | \: #2 \} }}

\newcommand{\flag}{\textnormal{\textsl{flag}}}
\newcommand{\next}{\textnormal{\textsl{next}}}

\newcommand{\mytree}{\ensuremath{T}}
\newcommand{\treeci}{\ensuremath{T'}}

\providecommand{\Ocomp}{\mathcal{O}}
\providecommand{\size}{\ensuremath{\Ocomp(|\mytree|)}}
\providecommand{\sizeprim}{\ensuremath{\Ocomp(|\mytree'|)}}

\newtheorem{theorem}{Theorem}
\newtheorem{lemma}{Lemma}
\newtheorem{corollary}{Corollary}
\theoremstyle{definition}

\theoremstyle{remark}

\newtheorem{example}{Example}
\newtheorem{clm}{Claim}
\newcommand{\alphabet}{\ensuremath{\mathbb F}}
\newcommand{\contexts}{\ensuremath{\mathbb Y}}
\newcommand{\grammar}{\ensuremath{\mathbb G}}
\DeclareMathOperator{\rank}{rank}

\usepackage[ruled]{algorithm}
\usepackage[noend]{algpseudocode}

\newcommand{\GRref}[1]{(\hyperref[Gr #1]{GR#1})}
\newcommand{\SKref}[1]{(\hyperref[sk #1]{SK#1})}

\newcommand{\GRrefall}{\GRref{1}--\GRref{5}}

\newcommand{\CPref}[1]{(\hyperref[cr #1]{CR#1})}
\newcommand{\FCref}[1]{(\hyperref[fc #1]{FC#1})}
\newcommand{\HGref}[1]{(\hyperref[lg #1]{HG#1})}

\begin{document}

\title{Approximation of smallest linear tree grammar}

\author[A.\ Je\.z]{Artur Je\.z}
\thanks{This work was supported under National Science Centre, Poland project number 2014/15/B/ST6/00615.}
\address{Institute of Computer Science, University of Wroc{\l}aw \\
ul.\ Joliot-Curie~15, 50-383 Wroc{\l}aw, Poland\\
\texttt{aje@cs.uni.wroc.pl}}

\author{Markus Lohrey}
\address{
University of Siegen, Department of Electrical Engineering and Computer Science, 57068 Siegen, Germany\\
\texttt{lohrey@eti.uni-siegen.de}}

\begin{abstract}
A simple linear-time algorithm
for constructing a linear 
context-free tree grammar of size $\Ocomp(r\grammarsize + r \grammarsize \log (n/r \grammarsize))$
for a given input tree \mytree{} of size $n$ is presented, where $\grammarsize$ is the size of 
a minimal linear context-free tree grammar for \mytree, and $r$ is the maximal rank
of symbols in \mytree{} (which is a constant in many applications). This is the first
example of a grammar-based tree compression algorithm with a good, i.e.\ logarithmic in terms of the size of the input tree, approximation
ratio. The analysis of the algorithm uses an extension of the recompression technique
from strings to trees.
\end{abstract}

	\maketitle

\section{Introduction}
	\label{sec:intro}

\noindent	
{\em Grammar-based compression} has emerged to an active field in string compression during the last decade.
The idea is to represent a given string $s$ by a small context-free grammar that generates only $s$;
such a grammar is also called a {\em straight-line program}, briefly SLP. For instance, the word $(ab)^{1024}$ can be represented
by the SLP with the productions $A_0 \to ab$ and $A_i \to A_{i-1} A_{i-1}$ for $1 \leq i \leq 10$ ($A_{10}$ is the start 
symbol). The size of this grammar is much smaller than the size (length) of the string $(ab)^{1024}$. In general,
an SLP of size $n$ (the size of an SLP is usually defined as the total length of all right-hand sides of productions)
can produce a~string of length $2^{\Omega(n)}$. Hence, an SLP can be seen as the succinct representation
of the generated word. The principle task of grammar-based string compression is to construct, from a~given input
string $s$, a small SLP that generates $s$. Unfortunately, finding a minimal (with respect to size) SLP for a given input
string is not achievable in polynomial time, unless {\bf P} = {\bf NP}~\cite{SLPapproxNPhard}
(recently the same result was shown also in case of a~constant-size alphabet~\cite{SLPaproxNPhardnew}).
Therefore, one can concentrate either on heuristic grammar-based compressors~\cite{KiefferY96,RePair,Sequitur},
or compressors whose output SLP is guaranteed to be not much larger than a size-minimal SLP for the input string~\cite{SLPaprox2,grammar,simplegrammar,SLPaprox,SLPaproxSakamoto}.
In this paper we are interested mostly in the latter approach.
Formally, in~\cite{SLPaprox2} the approximation ratio for a grammar-based compressor $\mathcal{G}$ is defined as the 
function $\alpha_{\mathcal{G}}$ with
$$
\alpha_{\mathcal{G}}(n) = \max  \frac{\text{size of the SLP produced by }\mathcal{G} \text{ with input } x}{\text{size of a minimal SLP for } x},
$$
where the maximum is taken over all strings of length $n$ (over an arbitrary alphabet).
The above statement means that unless {\bf P} = {\bf NP} there is no polynomial time grammar-based compressor
with the approximation ratio $1$.  Using approximation lower bounds for computing vertex covers, it is shown
in~\cite{SLPaprox2} that unless {\bf P} = {\bf NP} there is no polynomial time grammar-based compressor, whose 
approximation ratio is less than the constant 8569/8568. 

Apart from this complexity theoretic bound, the authors of~\cite{SLPaprox2}
prove lower and upper bounds on the approximation ratios of well-known grammar-based string compressors
({\sf LZ78}, {\sf BISECTION}, {\sf SEQUENTIAL}, {\sf RePair}, etc.). The currently best known approximation ratio of a polynomial time 
grammar-based string compressor is of the form $\Ocomp(\log (n/g))$, where $g$ is the size of a smallest
SLP for the input string. Actually, there are several compressors achieving this approximation ratio~\cite{SLPaprox2,grammar,simplegrammar,SLPaprox,SLPaproxSakamoto}
and each of them works in linear time (a property that a reasonable compressor should have). 

At this point, the reader might ask, what makes grammar-based compression so attractive. There are actually
several arguments in favour of grammar-based compression:
\begin{itemize}
\item The output of a grammar-based compressor 
is a clean and simple object, which may simplify the analysis
of a compressor or the analysis of algorithms that work on compressed data; see~\cite{Lohreysurvey} for a survey.
\item  There are grammar-based compressors which achieve very good compression ratios. For example {\sf RePair}~\cite{RePair} 
performs very well in practice and was for instance used for the compression of web graphs~\cite{ClaudeN10}.  
\item  The idea of grammar-based string compression can be generalised to other data types as long as suitable 
grammar formalisms are known for them.
See for instance the recent work on grammar-based graph compression~\cite{DBLP:conf/icde/ManethP16}.
\end{itemize}
The last point is the most important one for this work. In\cite{BuLoMa07}, grammar-based compression was generalised
from strings to trees.\footnote{A tree in this paper is always a rooted ordered tree over a ranked alphabet, i.e., every
node is labelled with a symbol and the rank of this symbol is equal to the number of children of the node.}
For this, context-free tree grammars were used. Context free tree grammars that produce only a single tree
are also known as straight-line context-free tree grammars (SLCF tree grammars).
Several papers deal with algorithmic problems on trees that are succinctly represented by SLCF tree grammars~%
\cite{onecontextvariablecompressed,GasconGS11,LoMa06,LohreyMS12,Schmidt-Schauss12rta,Schmidt-SchaussSA11}.
In~\cite{LohreyMM13}, {\sf RePair}  was generalised 
from strings to trees, and the resulting algorithm {\sf TreeRePair} achieves excellent results on real XML  trees. 
Other grammar-based tree compressors were developed in~\cite{LoMaNoe13,LohreyTreenlogn},
but none of these compressors has a good approximation ratio.
For instance, in~\cite{LohreyMM13} a series of trees is constructed, where the $n$-th tree $t_n$ has size $\Theta(n)$, 
there exists an SLCF tree grammar for $t_n$ of size $\Ocomp(\log n)$, but the grammar produced by {\sf TreeRePair} 
for $t_n$ has size $\Omega(n)$ (and similar examples can be constructed for the compressors in~\cite{LoMaNoe13,BuLoMa07}).

In this paper, we give the first example of a grammar-based tree compressor \algmain{} (for ``tree to grammar'') with an approximation ratio of 
$\Ocomp(\log (n/\grammarsize))$ assuming
the maximal rank $r$ of symbols is bounded and where $\grammarsize$ denotes the size of the smallest grammar generating the given tree;
otherwise the approximation ratio becomes $\Ocomp(r  + r \log (n/\grammarsize r))$.
Our algorithm \algmain{} is based on the work~\cite{grammar} of the first author, where another  grammar-based string compressor with an 
approximation ratio of  $\Ocomp(\log (n/\grammarsize))$ is presented (here $\grammarsize$ denotes the size of the smallest grammar for the input string).
The remarkable fact about this latter compressor is that in contrast
to~\cite{SLPaprox2,simplegrammar,SLPaprox,SLPaproxSakamoto} it does not use the {\sf LZ77} factorization of a string (which makes the compressors
from~\cite{SLPaprox2,simplegrammar,SLPaprox,SLPaproxSakamoto}
not suitable for a generalization to trees, since {\sf LZ77} ignores the tree structure and no good analogue of {\sf LZ77} for trees is known),
but is based on the {\em recompression technique}. This  technique was introduced in~\cite{fullyNFA} and
 successfully applied for a variety of 
algorithmic problems for SLP-compressed strings~\cite{fullyNFA,FCPM}  and word equations~\cite{wordeqgroups,onevarlinear,wordequations}. 
The basic idea is to compress a string using two operations: 
\begin{itemize}
\item block compressions: replace every maximal substring of the form $a^\ell$ for a letter $a$
by a new symbol $a_\ell$;
\item pair compression: for a given partition $\Sigma_\ell \uplus \Sigma_r$ replace every substring
$ab \in \Sigma_\ell \Sigma_r$ by a new symbol $c$.
\end{itemize}
It can be shown that the composition of block compression followed by pair compression (for a suitably chosen partition of the 
input letters) reduces the length of the string by a constant factor. Hence, the iteration of block compression
followed by pair compression yields a string of length one after a logarithmic number of phases.
By reversing a single compression step, one obtains a grammar rule for the introduced letter
and thus reversing all such steps yields an SLP for the initial string.
The term ``recompression'' refers to the fact,
that for a given SLP \grammar{},  block compression and pair compression can be simulated on  \grammar.
More precisely, one can compute from \grammar{} a new SLP $\grammar'$, which is not much larger than 
\grammar{} such that $\grammar'$ produces the result of block compression (respectively, pair compression)
applied to the string produced by \grammar. In~\cite{grammar}, the recompression technique is used to 
bound the approximation ratio of the above compression algorithm based on block and pair compression.

In this work we generalise the recompression technique from strings to trees. 
The operations of block compression and pair compression can be directly applied to chains of unary nodes
(nodes having only a single child)
in a tree. But clearly, these two operations alone cannot reduce the size of the initial tree by a constant
factor. Hence we need a third compression operation that we call {\em leaf compression}. It merges all children
of a node that are leaves into the node. The new label of the node determines the old label, the sequence of labels
of the children that are leaves, and their positions in the sequence of all children of the node.
Then, one can show that a~single phase, consisting of block compression (that we call chain compression), 
followed by pair compression (that we call \paircompression), followed by leaf compression reduces
the size of the initial tree by a constant factor. As for strings, we obtain an SLCF tree grammar for the input tree by reversing
the sequence of compression operations. The recompression approach again yields an approximation ratio
of $\Ocomp(\log (n/\grammarsize))$ (assuming that the maximal rank of symbols is a constant) 
for our compression algorithm \algmain{}, but the analysis is technically more subtle.

\begin{theorem}
\label{thm: main}
The algorithm \algmain{} runs in linear time, and for a tree $T$ of size $n$, 
it returns an SLCF tree grammar of size $\Ocomp(\grammarsize r + \grammarsize r \log(n/\grammarsize r))$,
where $\grammarsize$ is the size of a smallest SLCF grammar for $T$ and $r$ is the maximal rank of a symbol in $T$.
\end{theorem}
Note that in some specific cases it could happen that $n < \grammarsize r$ and so the term $\log(n/\grammarsize r)$ is in fact negative,
we follow the usual practice of bounding the logarithm from below by $0$, i.e.\ in such a case we assign $0$ as the value of the logarithm.

\medskip

\paragraph{\bf Related work on grammar-based tree compression}
We already mentioned  that grammar-based tree compressors were developed in~\cite{BuLoMa07,LohreyMM13,LoMaNoe13},
but none of these compressors has a good approximation ratio.
Another grammar-based tree compressors was presented in~\cite{Akutsu10}. It is based on the {\sf BISECTION} algorithm
for strings and has an approximation ratio of $\Ocomp(n^{5/6})$.
But this algorithm uses a different form of grammars (elementary ordered tree grammars) and it is not clear
whether the results from~\cite{Akutsu10} can be extended to SLCF tree grammars, or
whether the good algorithmic results for SLCF-compressed trees 
\cite{GasconGS11,LoMa06,LohreyMS12,Schmidt-Schauss12rta,Schmidt-SchaussSA11} can be extended to 
elementary ordered tree grammars.
Let us also mention the work from~\cite{toptrees} where trees are compressed by so called
top dags. These are another hierarchical representation of trees. Upper bounds on the size of the minimal top
dag are derived in~\cite{toptrees} and compared with the size of the minimal dag (directed
acyclic graph). More precisely, it is shown in%
~\cite{toptrees,DBLP:conf/wea/Hubschle-Schneider15} that the size of the minimal top dag is at most by a factor of $\Ocomp(\log n)$
larger than the size of the minimal dag.  Since 
dags can be seen as a special case of SLCF tree grammars, our main result is stronger.
In~\cite{LohreyTreenlogn},  the worst case  size of the output grammar of 
grammar-based tree compressors was investigated
and an algorithm that always returns an SLCF tree grammar of size
$\Ocomp(\frac{n}{\log_\sigma n})$ was given,
where $\sigma$ is the size of the input alphabet.
In fact this algorithm can be implemented in linear time or in logarithmic space~\cite{LohreyTreenlognarxiv}.
Note that (up to constant factors) the upper bound $\Ocomp(\frac{n}{\log_\sigma n})$ matches
the information-theoretic lower bound.
Slightly weaker results were obtained for the already mentioned top dags:
it was shown that top dags have size at most
$\Ocomp(\frac{n}{\log_\sigma n} \log \log n)$~\cite{DBLP:conf/wea/Hubschle-Schneider15}.
Finally, the performance of grammar-based string compression of trees that are encoded by
preorder traversal string was compared with grammar-based tree compresson~\cite{treetoSLP}:
the smallest string SLP for the preorder traversal of a tree can be exponentially smaller
than the smallest  SLCF tree grammar for the same tree.
But on a downside there are queries that can be efficiently (in {\bf P}) computed, when trees are represented
by SLCF tree grammars,  but become {\bf PSPACE}-complete, when  trees are represented by string SLPs.

\medskip

\paragraph{\bf Other applications of the technique: context unification}
The recompression method can be applied to word equations and it is natural to hope that its generalization to trees
also applies to appropriate generalizations of word equations.
Indeed, the tree recompression approach is used in~\cite{contextunification} to show
that the context unification problem can be solved in {\bf PSPACE}. It was a long 
standing open problem whether context unification is decidable~\cite{RTAproblem90}.

\medskip

\paragraph{\bf Parallel tree contraction}
Our compression algorithm is similar to algorithms for parallel tree evaluation~\cite{treecontraction1,treecontraction2}.
Here, the problem is to evaluate an algebraic expression of size $n$ in time $\Ocomp(\log n)$ on a PRAM. Using parallel tree contraction, this
can be achieved on a PRAM with $\Ocomp(n/\log n)$ many processors. The rake operation in parallel tree contraction
is the same as our leaf compression operation, whereas the {\em compress operations} contracts chains of unary nodes
and hence corresponds to block compression and pair compression. On the other hand, the specific 
features of block compression and pair compression that yield the approximation ratio of $\Ocomp(\log (n/g))$
have no counterpart in parallel tree contraction.

\medskip

\paragraph{\bf Computational model}
To achieve the linear running time we need some assumption on the computational model and form of the input.
We assume that numbers of $\Ocomp(\log n)$ bits
(where $n$ is the size of the input tree) 
can be manipulated in time  $\Ocomp(1)$ 
and that the labels of the input tree come from an interval $[1,\twodots,n^c]$, where $c$ is some constant.
Those assumption are needed so that we can employ \algradix,
which sorts $m$ many $k$-ary numbers of length $\ell$ in time $\Ocomp(\ell m+ \ell k)$, see e.g.~\cite[Section~8.3]{CoLeiRivStein09}.
In fact, we need a slightly more powerful version of \algradix{} that sorts lexicographically $m$ sequences of digits from $[1,\twodots, k]$
of lengths $\ell_1, \ell_2, \ldots, \ell_m$ in time $\Ocomp(k + \sum_{i=1}^m \ell_i)$. This is a standard generalisation of \algradix~\cite[Theorem 3.2]{DBLP:books/aw/AhoHU74}.
If for any reason the labels do not belong to an interval  $[1,\twodots,n^c]$, we can sort them in time $\Ocomp(n \log n)$ and replace them
with numbers from $\{1,2,\ldots, n\}$.

\section{Preliminaries}

\subsection{Trees}
Let us fix for every $i \geq 0$ a countably infinite set $\alphabet_i$ of {\em letters} (or {\em symbols}) of rank $i$,
where $\alphabet_i \cap \alphabet_j = \emptyset$ for $i \neq j$, and let $\alphabet = \bigcup_{i \geq 0} \alphabet_i$.
Symbols in $\alphabet_0$ are called \emph{constants}, while symbols in $\alphabet_1$ are called \emph{unary letters}.
We also write $\rank(a)=i$ if $a \in \alphabet_i$. A \emph{ranked alphabet} is a finite subset of $\alphabet$.
Let $F$ be a ranked alphabet. We also write $F_i$ for $F \cap \alphabet_i$ and $F_{\geq i}$ for $\bigcup_{j \geq i} F_i$.
An \emph{$F$-labelled tree} is a rooted, ordered tree whose nodes are labelled with elements from $F$,
satisfying the condition that if a node $v$ is labelled with $a$ then it has exactly $\rank(a)$ children, which are 
linearly ordered (by the usual left-to-right order).
We denote by $\mathcal T(F)$ the set of $F$-labelled trees.
In the following we simply speak about trees when the ranked alphabet is clear from the context or unimportant.
When useful, we identify an $F$-labelled tree with a term over $F$ in the usual way.
The size of the tree $t$ is its number of nodes and is denoted by $|t|$.
We assume that a tree is given using a pointer representation, i.e., each node has a list of its children (ordered from left to right)
and each node (except for the root) has a pointer to its parent node.

Fix a countable set $\contexts$ with
$\contexts \cap \alphabet = \emptyset$ of \emph{(formal) parameters}, which are usually denoted by
$y, y_1, y_2, \ldots$.
For the purposes of building trees with parameters, we treat all parameters as constants,
and so $F$-labelled trees with parameters from $Y\subseteq \contexts$ (where $Y$ is finite)
are simply $(F \cup Y)$-labelled trees, where the rank of every $y \in Y$ is $0$.
However, to stress the special role of parameters we write $\mathcal T(F,Y)$ for the set of $F$-labelled trees with parameters from $Y$.
We identify $\mathcal T(F)$ with $\mathcal T(F,\emptyset)$.
In the following we talk about \emph{trees with parameters} (or even trees) 
when the ranked alphabet and parameter set is clear from the context or unimportant.
The idea of parameters is best understood when we represent trees as terms:
For instance $f(y_1,a,y_2,y_1)$ with parameters $y_1$ and $y_2$ can be seen as a term with variables $y_1$, $y_2$ and we can instantiate
those variables later on.
A {\em pattern} (or {\em linear tree}) is a tree $t \in \mathcal T(F,Y)$, that contains for every $y \in Y$ at most one $y$-labelled node. Clearly,
a tree without parameters is a pattern.
All trees in this paper are patterns, and we do not mention this assumption explicitly in the following.

\begin{figure}
	\centering
		\includegraphics{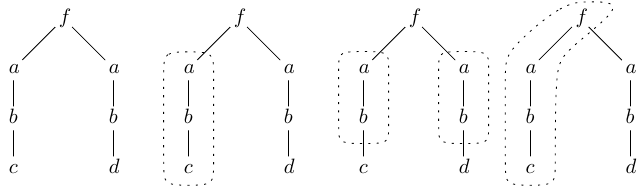}
	\caption{The tree $f(a(b(c)), a(b(d)))$, its subpattern $a(b(c))$, subpatterns $a(b(y))$, and subpattern $f(a(b(c)),y)$
	in which $a$ is the first child of $f$.}
	\label{fig:tree_pattern}
\end{figure}

When we talk of a {\em subtree} $u$ of a tree $t$, we always mean a full subtree in the sense that for every node of $u$ all children of that node in $t$
belong to $u$ as well. In contrast, a {\em subpattern} $v$ of $t$ is obtained from a subtree $u$ of $t$ by removing some of the subtrees of $u$. If we replace
these subtrees by pairwise different parameters, then we obtain a pattern $p(y_1, \ldots, y_n)$ and we say that \textit{(i)} the subpattern $v$ is an {\em occurrence}
of the pattern $p(y_1,\ldots,y_n)$ in $t$ and \textit{(ii)} $p(y_1,\ldots,y_n)$ is the pattern corresponding to the subpattern $v$ (this pattern is unique up to renaming
of parameters).
This later terminology applies also to subtrees, since a subtree is a subpattern as well.
A {\em context} $c(y)$ is a pattern with exactly one parameter, and occurrences of a context $c(y)$ in a tree are called {\em subcontexts}.
To make this notions clear, consider for instance the tree $f(a(b(c)),a(b(d)))$ with $f \in \alphabet_2$, $a,b \in \alphabet_1$ and $c,d \in \alphabet_0$.
It contains one occurrence of the pattern (in fact, tree)  $a(b(c))$, two occurrences of the pattern (in fact, context) $a(b(y))$ and one of the pattern 
 (in fact, context) $f(a(b(c)),y)$, see Figure~\ref{fig:tree_pattern}.

A {\em chain pattern} is a context of the form $a_1(a_2(\ldots(a_k(y))\ldots))$ with $a_1, a_2, \ldots, a_k \in\alphabet_1$.
A \emph{chain} in a tree $t$ is an occurrence of a chain pattern in $t$.
A chain $s$ in $t$ is \emph{maximal} if there is no chain $s'$ in $t$ with $s \subsetneq s'$.
A $2$-chain is a chain consisting of only two nodes (which, most of the time, are labelled
with different letters). 
For $a \in \alphabet_1$, an $a$-maximal chain is a chain such that (i) all nodes are labelled with $a$ and 
(ii) there is no chain $s'$ in $t$ such that $s \subsetneq s'$ and all nodes of $s'$ are labelled with $a$ too.
Note that an $a$-maximal chain is not necessarily a maximal chain. 
Consider for instance the tree $b(a(a(a(c))))$. The unique occurrence of the chain pattern $a(a(a(y)))$ is an
$a$-maximal chain, but is not maximal. The only maximal chain is the unique occurrence of the chain pattern
$b(a(a(a(y))))$,  see Figure~\ref{fig:tree_chain}.

\begin{figure}
	\centering
		\includegraphics{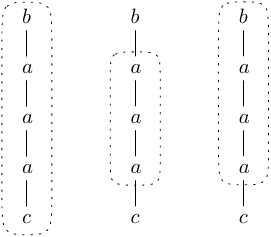}
	\caption{The tree $b(a(a(a(c))))$, its unique $a$-maximal chain (an occurrence of $a(a(a(y)))$)
	and its unique maximal chain (an occurrence of $b(a(a(a(y))))$).}
	\label{fig:tree_chain}
\end{figure}

We write $a_1a_2\cdots a_k$ for the chain pattern $a_1(a_2(\ldots(a_k(y))\ldots))$ 
and treat it as a string (even though this `string' still needs an argument on its right to form a proper term).
In particular, we write $a^\ell$  for the chain pattern consisting of $\ell$ many $a$-labelled nodes
and we write $vw$ (for chain patterns $v$ and $w$) for what should be $v(w(y))$.

\subsection{SLCF tree grammars}

For the further considerations, fix a countable infinite set $\mathbb{N}_i$ of symbols of rank $i$
with $\mathbb{N}_i \cap \mathbb{N}_j = \emptyset$ for $i \neq j$. Let $\mathbb{N} = \bigcup_{i \geq 0} \mathbb{N}_i$.
Furthermore, assume that $\alphabet \cap \mathbb{N} = \emptyset$.
Hence, every finite subset 
$N \subseteq \mathbb{N}$ is a ranked alphabet.
A \emph{linear context-free tree grammar},
\emph{linear CF tree grammar} for short, \footnote{There exist also non-linear CF tree grammars, which we do not need for our purpose.} is a tuple 
$\grammar =(N,F,P,S)$ such that the following 
conditions hold:
\begin{enumerate}[(1)]
	\item $N \subseteq \mathbb{N}$ is a finite set of \emph{nonterminals}.
	\item $F \subseteq \alphabet$ is a finite set of \emph{terminals}.
	\item $P$ (the set of \emph{productions}) is a finite set of
          pairs $(A, t)$ (for which we write $A \to t$), where $A \in N$ and $t \in \mathcal{T}(F \cup N,  \{y_1,\ldots,y_{\mathsf{rank}(A)}\})$
          is a pattern, which contains exactly one $y_i$-labelled node for each $1 \leq i \leq \rank(A)$.
          \item $S\in N$ is the \emph{start nonterminal}, which is of rank $0$.
\end{enumerate}
To stress the dependency of $A$ on its parameters we sometimes write
$A(y_1,\ldots,y_{\rank(A)}) \to t$ instead of $A \to t$.
Without loss of generality we assume that every nonterminal $B\in N\setminus\{S\}$ 
occurs in the right-hand side $t$ of some production $(A\to t)\in P$
(a much stronger fact is shown in~\cite[Theorem 5]{LohreyMS12}).

\begin{figure}
	\centering
		\includegraphics{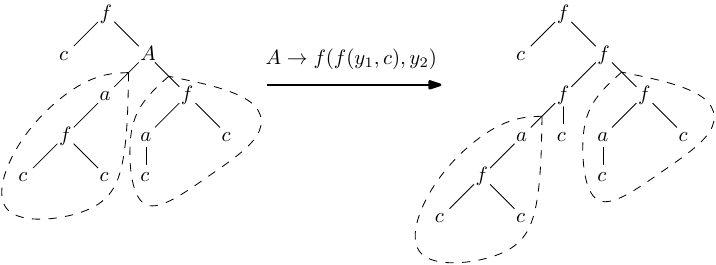}
	\caption{The tree $f(c,A(a(f(c,c)),f(a(c),c)))$ and the result of applying the rule $A(y_1,y_2) \rightarrow f(f(y_1,c),y_2)$. The subtrees that are substituted for parameters are within dashed blobs.}
	\label{fig:rewriting}
\end{figure}

A linear CF tree grammar  $\grammar$ is \emph{$k$-bounded} (for a natural number  $k$) if $\mathsf{rank}(A)\leq k$ for every $A\in N$.
Moreover, $\grammar$ is \emph{monadic} if it is $1$-bounded.
The derivation relation $\Rightarrow_\grammar$ on
$\mathcal T(F \cup N,Y)$ is defined as follows:
$s \Rightarrow_\grammar s'$ if and only if there is a production $(A(y_1,\ldots,y_\ell) \rightarrow t) \in P$
such that
$s'$ is obtained from $s$ by replacing some subtree $A(t_1,\ldots,t_\ell)$
of $s$ 
by $t$ with each $y_i$ replaced by $t_i$.
Intuitively, we replace  an $A$-labelled node by the pattern $t(y_1\ldots,y_{\rank(A)})$ and thereby identify the $j$-th child of $A$ with the unique $y_j$-labelled
node of the pattern, see Figure~\ref{fig:rewriting}.
Then $L(\grammar)$ is the set of all trees from $\mathcal{T}(F)$ 
(so $F$-labelled without parameters) that can be derived from $S$ (in arbitrarily many steps).

A \emph{straight-line context-free tree grammar} (or \emph{SLCF grammar} for short) 
is a linear CF tree grammar $\grammar = (N,F,P,S)$, where
\begin{itemize}
	\item for every $A \in N$ there is \emph{exactly one} production $(A \rightarrow t)\in P$
	with left-hand side $A$,
	\item if $(A \to t) \in P$ and $B$ occurs in $t$ then $B < A$, where $<$ is a linear order
	on $N$, and 
	\item $S$ is the maximal nonterminal with respect to $<$.
\end{itemize}
By the first two conditions, every $A \in N$ derives exactly one tree from $\mathcal T(F,\{y_1,\ldots,y_{\mathsf{rank}(A)}\})$. We denote this tree by $\eval(A)$ (like \emph{valuation}).
Moreover, we define $\eval(\grammar) = \eval(S)$, which is a tree from $\mathcal T(F)$.
In fact,  every tree from $\mathcal{T}(F \cup N,Y)$ derives
a unique tree from $\mathcal{T}(F,Y)$, where $Y$ is an arbitrary finite set of parameters.
For an SLCF grammar $\grammar = (N,F,P,S)$ we can assume without loss of generality that for every production $(A \to t) \in P$ the parameters
$y_1,\ldots,y_{\rank(A)}$ occur in $t$  in the order $y_1, y_2, \ldots, y_{\rank(A)}$ from left to right.
This can be ensured by a simple bottom-up rearranging procedure, see~\cite[proof of Theorem 5]{LohreyMS12}.
In the rest of the paper, when we speak of grammars, we always mean SLCF grammars.

\subsection{Grammar size}
When defining the \emph{size} $|\grammar|$ of the SLCF grammar \grammar,
one possibility is $|\grammar|=\sum_{(A\to t)\in P}|t|$, i.e., the sum of all sizes
of all right-hand sides.
However, consider for instance the rule $A(y_1,\ldots,y_\ell) \to f(y_1,\ldots,y_{i-1},a,y_i,\ldots,y_\ell)$.
It is in fact enough to describe the right-hand side as $(f,(i,a))$, as we have $a$ as the $i$-th child of $f$.
On the remaining positions we just list the parameters, whose order is known to us (see the remark in the previous
paragraph).
In general, each right-hand side of $\grammar$ can be specified by listing for each node its children that are \emph{not} parameters together
with their positions in the list of all children. These positions are numbers between $1$ and $r$  
(it is easy to show that our algorithm \algmain{} creates only nonterminals of rank at most $r$,
see Lemma~\ref{lem: small rank}, and hence
every node in a right-hand side has at most $r$ children)
and therefore fit into $\Ocomp(1)$ machine words.
For this reason we define the size $|\grammar|$ as the total number of non-parameter nodes
in all right-hand sides.
Note that such an approach is well-established;
see for instance~\cite{BuLoMa07}.

Should the reader prefer to define the size of a grammar as the total number of all nodes (including
parameters) in all right-hand sides, then
the approximation ratio of our algorithm \algmain{} has to be multiplied with the additional factor $r$.

\subsection{Notational conventions} \label{sec:conventions}
Our compression algorithm \algmain{} takes the input tree and applies to it local compression operations,
each such operation decreases the size of the tree.
With \mytree{} we always denote the current tree stored by \algmain, whereas
$n$ denotes the size of the initial input tree.
The algorithm \algmain{} relabels the nodes of the tree with fresh letters. With $F$ we always
denote the set of letters occurring in the current tree \mytree. 
By $r$ we denote the maximal rank of the letters occurring in the initial input tree.
The ranks of the fresh letters do not exceed $r$.

\begin{figure}
	\centering
		\includegraphics{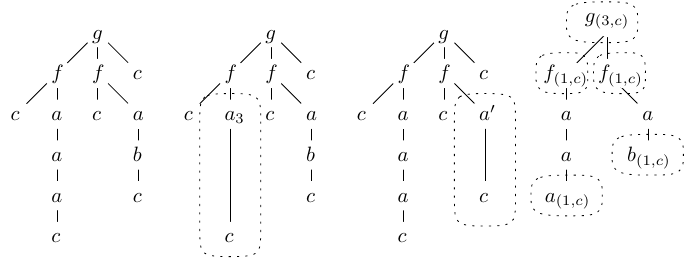}
	\caption{A tree and the results of compression operations: $a$-chain compression, $ab$-compression, and (parallel) leaf compression.}
	\label{fig:operations}
\end{figure}

\subsection{Compression operations}\label{sec: compression operations}
Our compression algorithm \algmain{} is based on three local replacement rules applied to trees:
\begin{enumerate}[(a)]
\item \label{a compression}$a$-maximal chain compression (for a unary symbol $a$), 
\item \label{b compression}\paircompression{},
\item \label{c compression}and leaf compression.
\end{enumerate}
Operations \eqref{a compression} and \eqref{b compression} apply only to unary letters and are direct translations of the operations used in the recompression-based algorithm for
constructing a grammar for a given string~\cite{grammar}.
To be more precise, \eqref{a compression} and \eqref{b compression} affect only chains, return chains as well, and when a chain is treated as a string
the results of \eqref{a compression} and \eqref{b compression}, respectively, correspond to the results of the corresponding operations on strings.
On the other hand, the last operations \eqref{c compression} is  new and designed specifically to deal with trees.
Let us inspect these operations:

\medskip

\paragraph{$a$-maximal chain compression} For a unary letter $a$ replace every $a$-maximal chain consisting of $\ell > 1$ nodes with a fresh unary letter $a_\ell$ (for all $\ell > 1$).

\medskip
\paragraph{$(a,b)$-pair compression} For two unary letters $a \neq b$ replace every occurrence of $ab$ by a single node labelled with a fresh unary letter $c$ (which identifies the pair $(a,b)$).

\medskip
\paragraph{$(f, i_1,a_1,\ldots,i_\ell,a_\ell)$-leaf compression:} 
	For $f \in F_{\geq 1}$, $\ell \geq 1$, $a_1, \ldots, a_\ell \in F_0$ and 
	$0 < i_1 < i_2 < \cdots < i_\ell \leq \rank(f) =: m$ replace every occurrence of $f(t_1,\ldots,t_m)$,
	where $t_{i_j} = a_j$ for $1 \leq j \leq \ell$ and $t_i$ is a non-constant for $i \not\in \{i_1, \ldots, i_\ell\}$, 
	with $f'(t_1,\ldots,t_{i_1-1},t_{i_1+1},\ldots,t_{i_\ell-1},t_{i_\ell+1},\ldots,t_m)$,
	where $f'$ is a fresh letter of rank $\rank(f) - \ell$ (which identifies $(f, i_1,a_1,\ldots,i_\ell,a_\ell)$).

\medskip
\noindent
Note that each of these operations decreases the size of the current tree. Also
note that for each of these compression operations one has to specify some arguments:
for chain compression the unary letter $a$,
for \paircompression{} the unary letters $a$ and $b$, and for leaf compression the letter $f$ (of rank at least 1) as well as the list
of positions $i_1 < i_2 < \dots < i_\ell$ and the constants $a_1$, \ldots, $a_{\ell}$.

Despite its rather cumbersome definition, the idea behind leaf compression is easy:
For a fixed occurrence of $f$ in a tree we `absorb' all leaf-children of $f$ that are constants (and do the same
for all other occurrences of $f$ that have the same set of leaf-children on the same positions).

Every application of one of our compression operations can be seen as the `backtracking' of a production of the grammar that we construct: 
When we replace $a^\ell$ by $a_\ell$,
we in fact introduce the new nonterminal $a_\ell(y)$ with the production 
\begin{equation} \label{production-a_l}
a_\ell(y) \to a^\ell(y).
\end{equation}
When we replace all  occurrences of the chain $ab$ by $c$, the new production is
\begin{equation} \label{production-c}
c(y) \to a(b(y)).
\end{equation}
Finally, for a $(f, i_1,a_1\ldots,i_\ell,a_\ell)$-leaf compression the production is
\begin{equation} \label{production-f'}
f'(y_1, \ldots, y_{\rank(f)-\ell}) \to f(t_1,\ldots,t_{\rank(f)}), 
\end{equation}
where $t_{i_j} = a_j$ for $1 \leq j \leq \ell$ and 
every $t_i$ with $i \not\in \{i_1, \ldots, i_\ell\}$ is a parameter
(and the left-to-right order of the parameters in the right-hand side is $y_1, \ldots, y_{\rank(f)-\ell}$).

Observe that all productions introduced in~\eqref{production-a_l}--\eqref{production-f'}
are for nonterminals of rank at most $r$.
\begin{lemma}
	\label{lem: small rank}
The rank of nonterminals defined by \algmain{} is at most $r$.
\end{lemma}
During the analysis of the approximation ratio of \algmain{}  we also
consider the nonterminals of a smallest grammar generating the given input tree.
To avoid confusion between these nonterminals and the nonterminals of the grammar produced
by \algmain{}, we insist on calling the fresh symbols introduced by \algmain{} ($a_\ell$, $c$, and $f'$ in \eqref{production-a_l}--\eqref{production-f'})
\emph{letters} and add  them to the set $F$ of current letters, so that $F$ always denotes the set of letters in the current
tree.  In particular, whenever we talk about nonterminals, productions, etc.\ we mean the ones of the smallest grammar
we consider.

Still, the above rules \eqref{production-a_l}, \eqref{production-c}, and \eqref{production-f'}
form the grammar returned by our algorithm \algmain{} and we need to estimate their size.
In order to not mix the notation, we call the size of the rule for a new letter $a$ the \emph{representation cost} for $a$
and say that $a$ \emph{represents} the subpattern it replaces in \mytree.
For instance, the representation cost of $a_\ell$ in \eqref{production-a_l} is $\ell$,
the representation cost of $c$ in \eqref{production-c} is $2$, 
and the representation cost of $f'$ in \eqref{production-f'} is $\ell + 1$.
A crucial part of the analysis of \algmain{}
is the reduction of the representation cost for $a_\ell$.
Note that instead of representing $a^\ell(y)$ directly via the rule \eqref{production-a_l},
we can introduce new unary letters representing some shorter chains in $a^\ell$
and build a longer chains using the smaller ones as building blocks. For instance, the rule  $a_8(y) \to a^8(y)$ 
can be replaced by the rules
$a_8(y) \to a_4(a_4(y))$, $a_4(y) \to a_2(a_2(y))$ and $a_2(y) \to a(a(y))$. This yields a
total representation cost of $6$ instead of $8$.
Our algorithm employs a particular strategy for representing $a$-maximal chains.
Slightly abusing the notation we say that the sum of the sizes of the right-hand
sides of the generated subgrammar is the representation cost for $a_\ell$ (for our strategy).

\subsection{Parallel compression}
The important property of the compression operations is that we can perform many of them in parallel:
Since different $a$-maximal chains and $b$-maximal chains do not overlap (regardless of whether $a = b$ or not)
we can perform  $a$-maximal chain compression for all $a \in \letters_1$ in parallel (assuming that the new letters do not belong to $\letters_1$).
This justifies the following compression procedure for compression of all $a$-maximal chains (for all $a \in \letters_1$) in a tree $t$:
\begin{algorithm}[H]
	\caption{\algtreechaincomp$(\letters_1,t)$: Compression of chains of letters from $\letters_1$ in a tree $t$}
	\label{alg:treechaincomp}
	\begin{algorithmic}[1]
    \For{$a \in \letters_1$} \Comment{chain compression}
			\For{$\ell \gets 1 \twodots |t|$}
				\State replace every $a$-maximal chain of size $\ell$ by a fresh letter $a_\ell$ \Comment{$a_\ell \notin \letters_1$}
			\EndFor
    \EndFor
	\end{algorithmic}
\end{algorithm}
We refer to the procedure \algtreechaincomp{} simply as \emph{chain compression}.
The running time of an appropriate implementation is considered in the next section
and the corresponding representation cost is addressed in Section~\ref{sec: size}.

A similar observation applies to leaf compressions: we can perform several different leaf compressions
as long as we do not try to compress the letters introduced by these leaf compressions.
\begin{algorithm}[H]
	\caption{\algtreechildcomp$(\letters_{\geq 1},\letters_0,t)$: leaf compression for parent nodes in $\letters_{\geq 1}$,
	and leaf-children in $\letters_0$ for a tree $t$}
	\label{alg:treechildcomp}
	\begin{algorithmic}[1]
		\For{$f \in \letters_{\geq 1}, 0 < i_1 <i_2 < \cdots < i_\ell \leq \rank(f) =: m, (a_1, a_2, \ldots, a_\ell) \in \letters_0^\ell$}
						\State replace each subtree $f(t_1,\ldots,t_m)$ s.t.~$t_{i_j} = a_j$ for $1 \leq j \leq \ell$
						and $t_i \notin F_0$ for
						\par $\! \! \! i \not\in \{i_1, \ldots, i_\ell\}$ 
						by
						$f'(t_1,\ldots,t_{i_1-1},t_{i_1+1},\ldots,t_{i_\ell-1},t_{i_\ell+1},\ldots,t_m)$
						\Comment{$f' \notin \letters_{\geq 1} \cup \letters_0$}
		\EndFor
	 \end{algorithmic}
\end{algorithm}
We refer to the procedure \algtreechildcomp{} as \emph{leaf compression}.
An efficient implementation is given in the next section, while the analysis of the number of introduced letters
is done in Section~\ref{sec: size}.

The situation is more subtle for \paircompression: observe that in a chain $abc$ we can compress $ab$ or $bc$ but we cannot do both
in parallel (and the outcome depends on the order of the operations).
However, as in the case of string compression~\cite{grammar},
parallel $(a,b)$-pair compressions are possible when we take $a$ and $b$ from disjoint subalphabets
$\lettersup$ and $\lettersdown$, respectively. In this case we can tell for each unary letter whether it should be the parent node
or the child node in the compression step and the result does not depend on the order of the considered $2$-chains,
as long as the new letters do not belong to $\lettersup \cup \lettersdown$.
\begin{algorithm}[H]
	\caption{\algtreepaircomp$(\lettersup,\lettersdown,t)$: $(\lettersup,\lettersdown)$-compression for a tree $t$}
	\label{alg:treepaircomp}
	\begin{algorithmic}[1]
		\For{$a \in \lettersup$ and $b \in \lettersdown$}
			\State replace each occurrence of  $ab$ with a fresh unary letter $c$ \Comment{$c \notin \lettersup \cup \lettersdown$}
		\EndFor
	 \end{algorithmic}
\end{algorithm}
The procedure \algtreepaircomp{} is called $(\lettersup,\lettersdown)$-compression in the following.
Again, its efficient implementation is given in the next section and the analysis of the number of introduced letters
is done in Section~\ref{sec: size}.

\section{Algorithm}\label{sec:algorithm}
In a  single phase of the algorithm \algmain{}, 
chain compression, $(\lettersup,\lettersdown)$-compression
and leaf compression are executed in this order (for an appropriate choice of the partition $\lettersup, \lettersdown$).
The intuition behind this approach is as follows: If the tree $t$ in question does not have any unary letters,
then leaf compression on its own reduces the size of $t$ by at least half, as it effectively reduces all constant nodes, i.e., leaves of the tree, and more than half of the nodes are leaves.
On the other end of the spectrum is the situation in which all nodes (except for the unique leaf)
are labelled with unary letters. In this case our instance is in fact a string. Chain compression
and \paircompression{} correspond to the operations of block compression and pair compression, respectively, from the earlier work on string
compression~\cite{grammar}, where it is shown that block compression followed by pair compression reduces the size 
of the string by a constant factor
(for an appropriate choice of the partition $\lettersup, \lettersdown$ of the letters occurring in the string).
The in-between cases are a mix of those two extreme scenarios and it can be shown that for them
the size of the instance drops by a constant factor in one phase as well.

Recall from Section~\ref{sec:conventions} that
$\mytree$ always denotes the current tree kept by \algmain{} and that  $F$ is the set of letters occurring in $\mytree$.
Moreover, $n$ denotes  the size of the input tree. 
\begin{algorithm}[H]
	\caption{\algmain: Creating an SLCF grammar for the input tree \mytree}
	\label{alg:main}
	\begin{algorithmic}[1]
	\While{$|\mytree|>1$} \label{alg:mainloop}
		\State $\letters_1 \gets $ list of unary letters in \mytree{} \label{listing letters}
		\State $\mytree \gets \algtreechaincomp(\letters_1,\mytree)$ \Comment{time \size}
		\State $\letters_1 \gets $ list of unary letters in \mytree{} \label{listing pairs}
		\State compute partition $\letters_1 = \lettersup \uplus \lettersdown$ using the algorithm from Lemma~\ref{lem:finding partition}\label{partition letters}  \Comment{time \size}
		\State $\mytree \gets \algtreepaircomp(\lettersup,\lettersdown,\mytree)$ \Comment{time \size}
		\State $\letters_0 \gets $ list of constants in \mytree, $\letters_{\geq 1} \gets $ list of other letters in \mytree{}
		\State $\mytree \gets  \algtreechildcomp(\letters_{\geq 1},\letters_0,\mytree)$  \Comment{time \size}
	\EndWhile
	\State \Return constructed grammar
	 \end{algorithmic}
\end{algorithm}
A single iteration of the main loop of \algmain{} is called a \emph{phase}.
In the rest of this section we show how to implement \algmain{}  in linear time (a polynomial implementation is straightforward),
while in Section~\ref{sec: size} we analyse the approximation ratio of \algmain{}.

Since the compression operations use \algradix{} for grouping,
it is important that right before such a compression the letters in \mytree{} form an interval of numbers.
As no letters are replaced in the listing of letters preceding such a compression,
it is enough to guarantee that after each compression, as a post-processing,
letters are replaced so that they form an interval of numbers.
Such a post-processing takes linear time.

\begin{lemma}[{cf.~\cite[Lemma~1]{grammar}}]
\label{lem:letters are consecutive}
After each compression operation performed by \algmain{} we can rename in time \sizeprim{} the letters used in \mytree{}
so that they form an interval of numbers, where $\mytree'$ denotes the tree before the compression step.
Furthermore, in the preprocessing step we can, in linear time, ensure the same property for the input tree.
\end{lemma}
\begin{proof}
Recall that we assume that the input alphabet consists of letters that can be identified with elements from an interval $\{1, \ldots, n^c\}$
for a constant $c$, see the discussion in the introduction.
Treating them as $n$-ary numbers of length $c$, we 
we can sort them using \algradix{}
in $\Ocomp(cn)$ time, i.e., in linear time. Then we can renumber the letters to $1, 2, \ldots, n'$ for some $n' \leq n$.
This preprocessing is done once at the beginning.

Fix the compression step and suppose that before the listing preceding this compression the letters formed an interval $[m,\twodots, m + k]$.
Each new letter, introduced in place of a compressed subpattern (i.e., a chain $a^\ell$, a chain $ab$
or a node $f$ together with some leaf-children) is assigned a consecutive value,
and so after the compression the letters occurring in  \mytree{} are within an interval $[m,\twodots, m + k']$
for some $k' > k$, note also that $k'-k \leq \size \leq \sizeprim$,
as each new letter labels a node in \mytree. It is now left to re-number the letters from $[m,\twodots, m + k']$,
so that the ones occurring in \mytree{} indeed form an interval.
For each symbol in the interval $[m,\twodots, m + k']$ we set a flag to $0$. Moreover, we set a variable $\next$ to $m+k'+1$. 
Then we traverse $T$ (in an arbitrary way).
Whenever we spot a letter $a \in [m,\twodots, m + k']$ with $\flag[a] = 0$, we set
$\flag[a] := 1$; $\mynew[a] := \next$, and $\next := \next+1$.
Moreover, we replace the label of the current node (which is $a$) by $\mynew[a]$.
When we spot a symbol $a \in [m,\twodots, m + k']$ with $\flag[a] = 1$, then
we replace the label of the current node (which is $a$) by $\mynew[a]$.
Clearly the running time is $\Ocomp(\size + \sizeprim)$ and after the algorithm the symbols form a subinterval of $[m+k'+1,\twodots,m+2k'+1]$.
\end{proof}
The reader might ask, why we do not assume in Lemma~\ref{lem:letters are consecutive} that the letters used in \mytree{}
form an initial interval of numbers (starting with $1$). The above proof can be easily modified so that it ensures this property.
But then, we would assign new names to letters, which makes it difficult to produce the final output grammar at the end.

\subsection{Chain compression} \label{sec chain comp}

The efficient implementation of \algtreechaincomp$(\letters_1,\mytree)$ is very simple:
We traverse $\mytree$. For an $a$-maximal chain of size $1 < \ell \leq |\mytree|$
we create a record $(a,\ell,p)$, where 
$p$ is the pointer to the top-most node in this chain.
We then sort these records lexicographically using \algradix{} (ignoring the last component and viewing
$(a,\ell)$ as a number of length $2$).
There are at most $|\mytree|$ records and we assume that $F$ can be identified with an interval,
see Lemma~\ref{lem:letters are consecutive}.  
Hence, \algradix{}  needs time $\size$ to sort the records.
Now, for a fixed unary letter $a$, the consecutive tuples with the first component $a$
correspond to all $a$-maximal chains, ordered by size.
It is easy to replace them in time \size{}  with new letters.

\begin{lemma}
\label{lem: chain comp time}
\algtreechaincomp$(\letters_1,\mytree)$ can be implemented in \size{} time.
\end{lemma}

Note that so far we did not care about the representation cost for the new letters that replace $a$-maximal chains.
We use a particular scheme to represent $a_{\ell_1}, a_{\ell_2}, \ldots, a_{\ell_k}$, which has a representation cost
of $\Ocomp(k + \sum_{i=1}^k \log(\ell_{i} - \ell_{i-1}))$, where we take $\ell _0 = 0$ for convenience.
This is an easy, but important improvement over $\Ocomp(k + \sum_{i=1}^k \log \ell_i)$
obtained using the binary expansion of the numbers $\ell_1, \ell_2, \ldots, \ell_k$.

\begin{lemma}[{cf.~\cite[Lemma~2]{grammar}}]
\label{lem: cost of powers}
Given a list $\ell_1 <  \ell_2 < \dots < \ell_k$ we can represent the letters $a_{\ell_1}, a_{\ell_2}, \ldots, a_{\ell_k}$
that replace the chain patterns $a^{\ell_1}, a^{\ell_2}, \ldots, a^{\ell_k}$ with a total cost
of $\Ocomp(k + \sum_{i=1}^k \log(\ell_{i} - \ell_{i-1}))$, where $\ell_0 = 0$.
\end{lemma}
\begin{proof}
The proof is identical, up to change of names, to the proof of Lemma~2 in~\cite{grammar},
still we supply it for completeness.

Firstly observe that without loss of generality we may assume that
the list $\ell_1, \ell_2, \dots, \ell_k$ is given in a sorted way,
as it can be easily obtained form the sorted list of occurrences of $a$-maximal chains.
For simplicity define $\ell_0 = 0$ and let $\ell = \max_{i=1}^k (\ell_{i} - \ell_{i-1})$.

In the following, we define rules for certain new unary letters $a_m$, each of them derives $a^m$
(in other words, $a_m$ represents $a^m$).
For each $1 \leq i \leq \lfloor\log \ell\rfloor$ introduce a new letter $a_{2^i}$ with the rule $a_{2^i}(y_1) \to a_{2^{i-1}}(a_{2^{i-1}}(y_1))$,
where $a_1$ simply denotes $a$. Clearly $a_{2^i}$ represents $a^{2^i}$
and the representation cost summed over all $1 \leq i \leq \lfloor\log \ell\rfloor$ is $2 \lfloor\log \ell\rfloor$.

Now introduce new unary letters $a_{\ell_i - \ell_{i-1}}$ for each $1 \leq i \leq k$,
which represent $a^{\ell_i - \ell_{i-1}}$.
These letters are represented using the binary expansions of the numbers $\ell_i - \ell_{i-1}$, i.e., by concatenation of $\lfloor \log(\ell_i - \ell_{i-1}) \rfloor+1$
many letters from $a_1, a_2, \ldots, a_{2^{\lfloor\log \ell\rfloor}}$.
This introduces an additional representation cost of  $\sum_{i=1}^k (1 + \lfloor\log(\ell_{i} - \ell_{i-1})\rfloor) \leq
k + \sum_{i=1}^k \log(\ell_{i} - \ell_{i-1})$.

Finally, each $a_{\ell_i}$ is  represented as $a_{\ell_{i}}(y_1) \to a_{\ell_{i} -\ell_{i-1}} (a_{\ell_{i-1}}(y_1))$,
which adds $2k$ to the representation cost.
Summing all contributions yields the promised value $\Ocomp(k + \sum_{i=1}^k \log(\ell_{i} - \ell_{i-1}))$.
\end{proof}

In the following we also use a simple property of chain compression:
Since no two $a$-maximal chains can be next to each other,
there are no $b$-maximal chains (for any unary letter $b$) of length greater than $1$ in \mytree{} after chain compression.

\begin{lemma}[cf.~{\cite[Lemma~3]{grammar}}]
\label{lem:consecutive letters are different}
In line~\ref{listing pairs} of algorithm \algmain{} there is no node in \mytree{} such that this node and its child are labelled with the same unary letter.
\end{lemma}
\begin{proof}
The proof is straightforward: suppose for the sake of contradiction that there is a node $u$ that is labelled with the unary letter $a$ and $u$'s unique
child $v$ is labelled with $a$, too.
There are two cases:

\medskip
\noindent
{\em Case 1.} The letter $a$ was present in \mytree{} in line~\ref{listing letters}:
But then $a$ was listed in $\letters_1$ in line~\ref{listing letters} and $u$ and $v$
are part of an $a$-maximal chain that was replaced by a single node during \algtreechaincomp$(\letters_1,\mytree)$.

\medskip
\noindent
{\em Case 2.} The letter  $a$ was introduced during \algtreechaincomp$(\letters_1,\mytree)$:
Assume that $a$ represents $b^\ell$.
Hence $u$ and $v$ both replaced $b$-maximal chains.  But this is not possible
since the definition of a $b$-maximal chain implies that two $b$-maximal chains are not adjacent.
\end{proof}

\subsection{\Paircompression}

The operation of  \paircompression{} is implemented similarly as chain compression.
As already noticed, since 2-chains can overlap, compressing all 2-chains at the same time is not possible.
Still, we can find a subset of non-overlapping chain patterns of length 2
 in \mytree{} such that a (roughly) constant fraction
of unary letters in \mytree{} is covered by occurrences of these chain patterns.
This subset is defined by a \emph{partition} of the letters from $\letters_1$ occurring in \mytree{} into subsets $\lettersup$ and $\lettersdown$.
Then we replace all $2$-chains, whose first (respectively, second) node is labelled with a letter from  $\lettersup$ (respectively, $\lettersdown$).
Our first task is to show that indeed such a partition exists and that it can be found in time \size{}.

\begin{lemma}
\label{lem:finding partition}
Assume that (i) \mytree{} does not contain an occurrence of a chain pattern $aa$ for some $a \in F_1$ and (ii)
that the symbols in $\mytree$  form an interval of numbers.
Then, in time \size{} one can find
a partition $\letters_1 = \lettersup \uplus \lettersdown$  such that
the number of occurrences of chain patterns from $\lettersup\lettersdown$ in \mytree{} is at least $(n_1 - c + 2)/4$,
where $n_1$ is the number of nodes in \mytree{} with a unary label
and $c$ is the number of maximal chains in \mytree.
In the same running time we can provide for each $ab \in \lettersup\lettersdown$ occurring in \mytree{}
a lists of pointers to all occurrences of $ab$ in \mytree.
\end{lemma}

\begin{proof}
For a choice of $\lettersup$ and $ \lettersdown$ we say that occurrences of $ab \in \lettersup\lettersdown$
are \emph{covered} by the partition $\letters_1 = \lettersup \uplus \lettersdown$.
We extend this notion also to words: a partition covers also occurrences of a chain pattern $ab$ in a word (or set of words).

The following claim is a slighter stronger version of~\cite[Lemma~4]{grammar},
the proof is essentially the same, still, for completeness, we provide it below:

\begin{clm}[{\cite[Lemma~4]{grammar}}] \label{clm:finding partition}
\em For words $w_1$, $w_2$, \ldots, $w_c$ that do not contain a factor $aa$ for any symbol $a$ and whose alphabet  can be identified
with an interval of numbers of size $m$, one can
in time $\Ocomp(\sum_{i=1}^c |w_i| + m)$ partition the letters occurring in $w_1$, $w_2$, \ldots, $w_c$ into sets $\lettersup$ and $\lettersdown$ such that
the number of occurrences of chain patterns from $\lettersup\lettersdown$ in $w_1$, $w_2$, \ldots, $w_c$ is at least $(\sum_{i=1}^c (|w_i|-1))/4$.
In the same running time we can provide for each $ab \in \lettersup\lettersdown$ occurring in $w_1$, $w_2$, \ldots, $w_c$
a lists of pointers to all occurrences of $ab$ in $w_1$, $w_2$, \ldots, $w_c$.
\end{clm}
It is easy to derive the statement of the lemma from this claim:
Consider all maximal chains in \mytree, and let us treat the corresponding chain patterns as strings $w_1, w_2, \ldots, w_c$.
The sum of their lengths is $n_1 \leq |\mytree|$.
By the assumption from the lemma no two consecutive letters in strings $w_1$, $w_2$, \ldots, $w_c$ are identical.
Moreover, the alphabet of $w_1$, $w_2$, \ldots, $w_c$ is within an interval of size $\Ocomp(|T|)$.
By Claim~\ref{clm:finding partition} one can compute in time $\Ocomp(\sum_{i=1}^c |w_i| + |\mytree|) \leq \size$
a partition $\lettersup \uplus \lettersdown$ of $\letters_1$ such that $\frac{\sum_{i=1}^c (|w_i|-1)}{4}$
many 2-chains from $w_1$, $w_2$, \ldots, $w_c$ are covered by this partition,
and hence the same applies to \mytree.
Moreover, by Claim~\ref{clm:finding partition} one can also compute in time 
$\Ocomp(\sum_{i=1}^c |w_i|) \leq \size{}$ for every $ab \in  \lettersup\lettersdown$ occurring in $w_1$, $w_2$, \ldots, $w_c$
a lists of pointers to all occurrences of $ab$ in $w_1$, $w_2$, \ldots, $w_c$. It is straightforward to compute from this list
a lists of pointers to all occurrences of $ab$ in \mytree.

Let us now provide a proof of Claim~\ref{clm:finding partition}:

\medskip
\noindent
{\em Proof of Claim~\ref{clm:finding partition}.}
Observe that finding a partition reduces to the (well-studied and well-known) problem of finding a cut in a directed and weighted graph:
For the reduction, for each letter $a$ we create a node $a$ in a graph and make the weight of the edge $(a, b)$
the number of occurrences of $ab$ in $w_1$, $w_2$, \ldots, $w_c$. 
A \emph{directed cut} in this graph is a partition $V_1 \uplus V_2$ of the vertices, and the weight of this cut is the sum of all weights of edges
in $V_1 \times V_2$.
It is easy to see that a directed cut of weight $k$ corresponds to a partition of the letters covering exactly $k$ occurrences of chain patterns 
(and vice-versa).
The rest of the the proof gives the standard construction~\cite[Section~6.3]{mitzenmacher2005probability}
in the terminology used in the paper (the running time analysis is not covered in standard sources).

The existence of a partition covering at least one fourth of the occurrences
can be shown by a~simple probabilistic argument:
Divide $\letters_1$ into $\lettersup$ and $\lettersdown$ randomly, where each letter
goes to each of the parts with probability $1/2$. 
Fix an occurrence of $ab$, then $a \in \lettersup$ and $b \in \lettersdown$ with probability $1/4$.
There are $\sum_{i=1}^c (|w_i|-1)$ such 2-chains in $w_1$, $w_2$, \ldots, $w_c$, so the expected number of occurrences of patterns from 
$\lettersup\lettersdown$ in $w_1$, $w_2$, \ldots, $w_c$ is $(\sum_{i=1}^c (|w_i|-1))/4$. Hence, there exists a partition that covers at least $(\sum_{i=1}^c (|w_i|-1))/4$
many occurrences of 2-chains.
Observe, that the expected number of occurrences of patterns from $\lettersup \lettersdown \cup \lettersdown \lettersup$
is $(\sum_{i=1}^c (|w_i|-1))/2$.

The deterministic construction of  a partition covering at least $(\sum_{i=1}^c (|w_i|-1))/4$ occurrences
follows by a simple derandomisation, using the conditional expectation approach.
It is easier to first find a~partition $\lettersup \uplus \lettersdown$ such that at least $(\sum_{i=1}^c (|w_i|-1))/2$ many occurrences of 
2-chains  in $w_1$, $w_2$, \ldots, $w_c$ are covered
by $\lettersup\lettersdown \cup \lettersdown\lettersup$. We then choose $\lettersup\lettersdown$ or $\lettersdown\lettersup$,
depending on which of them covers more occurrences.

Suppose that we have already assigned some letters to $\lettersup$ and $\lettersdown$
and we have to decide where the next letter $a$ is assigned to.
If it is assigned to $\lettersup$, then all occurrences of patterns from $a\lettersup \cup \lettersup a$ are not going
to be covered, while occurrences of patterns from $a \lettersdown \cup \lettersdown a$ are. 
A similar observation holds if $a$ is assigned to $\lettersdown$.
The algorithm \alggreedypairs{} makes a greedy choice, maximising the number of covered 2-chains in each step.
As there are only two options, the choice covers at least half of all occurrences of 2-chains that contain the letter $a$
and a letter from $\lettersup \uplus \lettersdown$. 
Finally, as each occurrence of a pattern $ab$  from $w_1$, $w_2$, \ldots, $w_c$ is considered exactly once
(namely when the second letter of $a$ and $b$ is considered in the main loop),
this procedure guarantees that at least half of all  2-chains in $w_1$, $w_2$, \ldots, $w_c$ are covered.

In order to make the selection efficient,
the algorithm \alggreedypairs{} below keeps  for every letter $a$  counters
$\mycount{up}[a]$ and $\mycount{down}[a]$, 
storing the number of occurrences of patterns from
$a\lettersup \cup \lettersup a$ and $a\lettersdown \cup \lettersdown a$, respectively,  in $w_1$, $w_2$, \ldots, $w_c$.
These counters are updated as soon as a letter is assigned to $\lettersup$ or $\lettersdown$.

\begin{algorithm}[H]
  \caption{\alggreedypairs{} \label{greedy pairs}}
  \begin{algorithmic}[1]
  \State $\letters_1 \gets $ set of letters used in $w_1$, $w_2$, \ldots, $w_c$
	\State $\lettersup \gets \lettersdown \gets \emptyset$ \Comment{organised as a bit vector}
	\For{$a \in \letters_1$}
		\State $\mycount{up}[a] \gets \mycount{down}[a] \gets 0$ \Comment{initialisation}
	\EndFor
	\For{$a \in \letters_1$} 
		\If{$\mycount{down}[a] \geq \mycount{up}[a]$} \Comment{choose the one that guarantees larger cover}
			\State $\textnormal{\textsl{choice}} \gets \textnormal{up}$
				\Else
			\State $\textnormal{\textsl{choice}} \gets \textnormal{down}$
		\EndIf
			\For{$b \in \letters_1$ and all occurrences of $ab$ or $ba$ in $w_1$, $w_2$, \ldots, $w_c$} \State $\mycount{choice}[b] \gets \mycount{choice}[b] + 1$ \label{counter update}
			\EndFor
			\State $\letterschoice \gets \letterschoice \cup \{a\}$
	\EndFor
	\If{\# occurrences of patterns from $\lettersdown\lettersup$ in $w_1$, $w_2$, \ldots, $w_c$ > \# occurrences of patterns from $\lettersup\lettersdown$ in $w_1$, $w_2$, \ldots, $w_c$}
		\State switch $\lettersdown$ and $\lettersup$ \label{actual partition}
	\EndIf
	\State \Return $(\lettersup, \lettersdown)$
  \end{algorithmic}
\end{algorithm}

By the argument given above, when $\letters_1$ is partitioned into $\lettersup$ and $\lettersdown$ by \alggreedypairs,
at least half of all $2$-chains in $w_1$, $w_2$, \ldots, $w_c$ are occurrences of patterns from
$\lettersup\lettersdown \cup \lettersdown \lettersup$.
Then one of the choices $(\lettersup,\lettersdown)$ or $(\lettersdown,\lettersup)$ covers at least one fourth of 
all $2$-chains in $w_1$, $w_2$, \ldots, $w_c$.

It is left to give an efficient variant of \alggreedypairs.
The non-obvious operations are  the updating of $\mycount{choice}[b]$ in line~\ref{counter update}
and the choice of the actual partition in line~\ref{actual partition}.
All other operation clearly take at most time $\Ocomp(\sum_{i=1}^c |w_i|)$.
The latter is simple: since we organise $\lettersup$ and $\lettersdown$ as bit vectors,
we can read each $w_1$, $w_2$, \ldots, $w_c$ from left to right (in any order) and calculate the number of occurrences of patterns from $\lettersup\lettersdown$ as well as 
those from $\lettersdown\lettersup$ in time $\Ocomp(\sum_{i=1}^c |w_i|)$  (when we read a~pattern $ab$ we check in $\Ocomp(1)$ time whether $ab \in \lettersup \lettersdown$ or $ab \in \lettersdown \lettersup$).
Afterwards we choose the partition that covers more $2$-chains in $w_1$, $w_2$, \ldots, $w_c$.

To implement $\mycount{up}$ and $\mycount{down}$,
for each letter $a$ in $w_1$, $w_2$, \ldots, $w_c$ we store a \emph{right list}
$\rightl = \makeset{b}{ab \text{ occurs in } w}$, represented as a list.
Furthermore, the element $b$ on the right list points
to a list of all occurrences of the pattern $ab$ in $w_1$, $w_2$, \ldots, $w_c$.
There is a similar \emph{left list} $\leftl = \makeset{b}{ba \text{ occurs in } w}$.
We comment on how to create the left lists and right lists in linear time later.

Given \rightl[] and \leftl[], performing the update in line~\ref{counter update} is easy:
We go through $\rightl$ (respectively, $\leftl$) and increment $\mycount{up}[b]$ for every occurrence of $ab$ (respectively, $ba$).
Note that in this way each of the lists $\rightl$ ($\leftl$) is read once during \alggreedypairs,
the same applies also to pointers from those lists.
Therefore, all updates of $\mycount{up}$ and $\mycount{down}$ only need time $\Ocomp(\sum_{i=1}^c |w_i|)$,
as the total number of pointers on those lists is \size.

It remains to show how to initially create \rightl{} (\leftl{} is created similarly).
We read $w_1$, $w_2$, \ldots, $w_c$. When reading a  pattern $ab$ we create a record
$(a,b,p)$, where $p$ is a pointer to this occurrence.
We then sort these records lexicographically using \algradix, ignoring the last component.
There are $\sum_{i=1}^c |w_i|$ records and the alphabet is an interval of size $m$,
so \algradix{} needs time $\Ocomp(\sum_{i=1}^c |w_i| + m)$. 
Now, for a fixed letter $a$, the consecutive tuples with the first component $a$
can be turned into \rightl: for $b \in \rightl$ we want to store a list $I$ of pointers to occurrences of $ab$.
On a sorted list of records the entries $(a,b,p)$ for $p \in I$ form an interval of consecutive records.
This shows the first statement from Claim~\ref{clm:finding partition}.

In order to show the second statement from Claim~\ref{clm:finding partition}, i.e.,
in order to get for each $ab \in \lettersup\lettersdown$
the lists of pointers to occurrences of $ab$ in $w_1$, $w_2$, \ldots, $w_c$,
it is enough to read \rightl[] and filter the patterns $ab$ such that $a\in \lettersup$ and $b \in \lettersdown$;
the filtering can be done in $\Ocomp(1)$ per occurrence as $\lettersup$ and $\lettersdown$ are represented as bitvectors.
The total needed time is $\Ocomp(\sum_{i=1}^c |w_i|)$.
This concludes the proof of Claim~\ref{clm:finding partition}
and thus also the proof of Lemma~\ref{lem:finding partition}.
\end{proof}

When for each pattern $ab \in \lettersup\lettersdown$ the list of its occurrences in \mytree{} is provided,
the replacement of these  occurrences is done by going through the list and replacing each of the occurrences, which is done in linear time.
Note that since $\lettersup$ and $\lettersdown$ are disjoint, the considered occurrences cannot overlap and the order of the replacements
is unimportant.

\begin{lemma}
\label{lem: pair comp time}
\algtreepaircomp$(\lettersup,\lettersdown,\mytree)$ can be implemented in \size{} time.
\end{lemma}

\subsection{Leaf compression}
Leaf compression is done in a way similar to chain compression and $(\lettersup,\lettersdown)$-compression:
We traverse \mytree. Whenever we reach a node $v$ labelled with a symbol $f \in F_{\geq 1}$, we scan the list of its children.
Assume that this list is $v_1, v_2, \ldots, v_m$.
When no $v_i$ is a leaf, we do nothing.
Otherwise, let $1 \leq i_1 < i_2 <  \cdots < i_\ell \leq m$ be a list of those positions such that $v_{i_k}$ is a leaf,
say labelled with a constant $a_k$, for all $1 \leq k \leq \ell$.
We create a record $(f,i_1,a_1,i_2,a_2,\ldots, i_\ell,a_\ell,p)$,
where $p$ is a pointer to node $v$, and continue with the traversing of \mytree.
Observe that the total number of elements in the created tuples is at most $2|\mytree|$.
Furthermore each position index is at most $r \leq |\mytree|$ and by Lemma~\ref{lem:letters are consecutive}
also each letter is a number from an interval of size at most $|\mytree|$.
Hence \algradix{} sorts those tuples (ignoring the pointer coordinate) in time \size{} 
(we use the \algradix{} version for lists of varying length).
After the sorting the tuples corresponding to nodes with the same label and the same constant-labelled children (at the same positions)
are consecutive on the returned list, so we can easily perform the replacement.
Given a tuple $(f,i_1,a_1,i_2,a_2,\ldots, i_\ell,a_\ell,p)$ we use the last component (i.e.\ pointer) in the created records to localize the node, replace the label $f$ with the fresh label $f'$
and remove the children at positions $i_1, i_2, \ldots, i_\ell$
(note that in the meantime some other children might become leaves, we do not remove them, though).
Clearly all of this takes time \size.

\begin{lemma}
\label{lem: leaf comp time}
\algtreechildcomp$(\letters_{\geq 1}, \letters_0, \mytree)$ can be implemented in \size{} time.
\end{lemma}

\subsection{Size and running time}
It remains to estimate the total running time of our algorithm \algmain, summed over all phases.
As each subprocedure in a phase has running time \size{} and there are constant number of them in a phase,
it is enough to show that $|T|$ is reduced by a constant factor per phase (then the sum of the running
times over all phases is a geometric sum).

\begin{lemma}
\label{lem:number of phases}
In each phase, $|\mytree|$ is reduced by a constant factor.
\end{lemma}
\begin{proof}
For $i \geq 0$ let $n_i$, $n_i'$, $n_i''$ and $n_i'''$ be the number of nodes labelled with a letter of rank $i$ 
in $\mytree$ at the beginning of the phase, after  chain compression, \paircompression, and leaf compression, respectively.
Let $n_{\geq 2} = \sum_{i \geq 2} n_i$ and define $n'_{\geq 2}$, $n''_{\geq 2}$, and $n'''_{\geq 2}$ similarly.
We have
\begin{equation}
\label{eq: leaves and inner}
n_0 \geq n_{\geq 2} + 1 \enspace .
\end{equation}
To see this, note that there are $n_0 + n_1 + n_{\geq 2} - 1$ nodes that are children (`$-1$' is for the root). On the other hand,
a node of arity $i$ is a parent node for $i$ children. So the number of children is at least $2 n_{\geq 2} + n_1$.
Comparing those two values yields~\eqref{eq: leaves and inner}.

We next show that
$$
n_0''' + n_1''' + n_{\geq 2}''' \leq \frac{3}{4}(n_0 + n_1 + n_{\geq 2}) \enspace ,
$$
which shows the claim of the lemma.
Let $c$ denote the number of maximal chains in \mytree{} at the beginning of the phase,
this number does not change during chain compression and \paircompression.
Observe that
\begin{equation}
	\label{eq: bounding number of chains}
	c \leq n_{\geq 2} + n_0 \enspace.
\end{equation}
Indeed, consider a maximal chain. Then the node below the chain has a label from $F_{\geq 2} \cup F_0$.
Summing this up over all chains, we get~\eqref{eq: bounding number of chains}.

Clearly after chain compression we have $n_0' = n_0$, $n_1' \leq n_1$ and $n_{\geq 2}' = n_{\geq 2}$.
Furthermore, the number of maximal chains does not change.
During  \paircompression, by Lemma~\ref{lem:finding partition},
we choose a partition such that at least $\frac{n_1' - c + 2}{4}$ many $2$-chains are compressed
(note that the assumption of Lemma~\ref{lem:finding partition} that no parent node and its child are labelled with the same
unary letter is satisfied by Lemma~\ref{lem:consecutive letters are different}),
so the size of the tree is reduced by at least $\frac{n_1' - c + 2}{4}$.
Hence, the size of the tree after  \paircompression{} is at most
\begin{align}
\notag
n_0'' + n_1'' + n_{\geq 2}''
	&\leq \notag
n_0' + n_1' + n_{\geq 2}' - \frac{n_1' - c + 2}{4}\\
	&= \notag
n_0' + \frac{3n_1'}{4} + n_{\geq 2}' + \frac{c}{4} - \frac{1}{2} \\
	&\leq 	\label{eq: after pair compression}
n_0 + \frac{3n_1}{4} + n_{\geq 2} + \frac{c}{4} - \frac{1}{2} \enspace.
\end{align}
Lastly, during leaf compression the size is reduced by $n_0'' = n_0$.
Hence the size of \mytree{} after all three compression steps is
\begin{align*}
n_0''' + n_1''' + n_{\geq 2}'''
	&= n_0'' + n_1'' + n_{\geq 2}'' - n_0 &\text{leaf compression}\\
	&\leq
n_0 + \frac{3n_1}{4} + n_{\geq 2} + \frac{c}{4} - \frac{1}{2} - n_0  &\text{from~\eqref{eq: after pair compression}}\\
	&=
\frac{3n_1}{4} + n_{\geq 2} + \frac{c}{4} -\frac{1}{2} &\text{simplification}\\
	&\leq 
\frac{3n_1}{4} + n_{\geq 2} + \underbrace{\frac{n_{\geq 2}}{4} + \frac{n_0}{4}}_{\geq c/4} - \frac{1}{2} &\text{from~\eqref{eq: bounding number of chains}}\\
	&<
\frac{n_0}{4} + \frac{3n_1}{4} + \frac{5n_{\geq 2}}{4}&\text{simplification}\\
	&<
\frac{3}{4}\Big(n_0 + n_1 + n_{\geq 2}\Big) , &\text{from~\eqref{eq: leaves and inner}} 
\end{align*}
as claimed.
\end{proof}

\begin{theorem}
\label{thm: time}
\algmain{} runs in linear time.
\end{theorem}
\begin{proof}
By Lemma~\ref{lem:letters are consecutive} there is a linear preprocessing.
By Lemmata~\ref{lem:letters are consecutive},~\ref{lem: chain comp time},~\ref{lem: pair comp time}, and~\ref{lem: leaf comp time},
each phase takes \size{} time and by Lemma~\ref{lem:number of phases},
$|\mytree|$ drops by a constant factor in each phase.
As the initial size of $\mytree$ is $n$, the total running time is $\Ocomp(n)$.
\end{proof}

\section{Size of the grammar: recompression}
\label{sec: size}

To bound the cost of representing the letters introduced during the construction of the SLCF grammar,
we start with a smallest SLCF grammar $\grammar_{\textnormal{opt}}$ generating the input tree  \mytree{} (note that $\grammar_{\textnormal{opt}}$
is not necessarily unique)
and show that we can transform it into an SLCF grammar \grammar{} (also generating \mytree) of a special normal form, called handle grammar.
This form is described in detail in Section~\ref{subsec: normal form}.
The grammar \grammar{} is of size
$\Ocomp(r|\grammar_{\textnormal{opt}}|)$, where $r$ is the maximal rank of symbols in $F$ (the set of letters occurring in $\grammar_{\textnormal{opt}}$).
The transformation is based on known results on normal forms for SLCF grammars~\cite{LohreyMS12},
see Section~\ref{subsec: normal form}.

To bound the size of \grammar, 
we assign \emph{credits} to \grammar:
each occurrence of a letter in a right-hand side of \grammar{} has two units of credit.
If such a letter is removed from \grammar{} for any reason,
its credit is \emph{released}
and if a new letter is inserted into
some right-hand side of a rule,
then we \emph{issue} its credit.

During the run of \algmain{} we modify \grammar{}, preserving its special handle form, so that it generates \mytree{}
(i.e., the current tree kept by \algmain) after each of the compression steps of \algmain.
In essence, if a compression is performed on \mytree{}
then we also apply it on \grammar{} and modify \grammar{}
so that it generates the tree \mytree{} after the compression step.
Then the cost of representing the letters introduced by \algmain{}
is paid by credits released during the compression of letters in \algmain.
Therefore, instead of computing the total representation cost of the new letters, it suffices
to calculate the total amount of issued credit, which is much easier than calculating
the actual representation cost.
Note that this is entirely a mental experiment for the purpose of the analysis,
as \grammar{} is not stored or even known by \algmain.
We just perform some changes on it depending on the actions of \algmain.

The analysis outlined above is not enough to bound the representation
cost for chain compression, we need specialised tools for that.
They are described in Section~\ref{subsec: chain cost}.

In this section we show a slightly weaker bound, the full proof of the bound from Theorem~\ref{thm: main}
is presented in Section~\ref{sec: improved}.

\subsection{Normal form}
\label{subsec: normal form}
As explained above, in our mental experiment we modify the grammar \grammar{}
and perform the compression operations on it.
To make the analysis simpler, we want to have a special form
in which the compression operation will not interact too much
between different parts of the grammar.
This idea is formalised using \emph{handles}:
We say that a pattern $t(y)$ is a \emph{handle} if it is of the form
$$
f(w_1(\gamma_1),w_2(\gamma_2),\ldots,w_{i-1}(\gamma_{i-1}),y,w_{i+1}(\gamma_{i+1}),\ldots,w_\ell(\gamma_\ell)) ,
$$
where $\rank(f) = \ell$, every $\gamma_j$ is either a constant symbol or a nonterminal of rank $0$, every $w_j$ is a chain pattern, and $y$ is a parameter,
see Figure~\ref{fig:handle}.
Note that $a(y)$ for a unary letter $a$ is a handle.
Since handles have one parameter only, for handles $h_1, h_2, \ldots, h_\ell$
we write $h_1h_2\cdots h_\ell$ for the tree $h_1(h_2(\ldots (h_\ell(y))))$
and treat it as a string, similarly to chains patterns.

We say that an SLCF grammar $\grammar = (N,F,P,S)$ is a \emph{handle grammar} (or simply ``\grammar{} is handle'') if
the following conditions hold:
\begin{enumerate}[({HG}1)]
	\item \label{lg 1} $N \subseteq \mathbb{N}_0 \cup \mathbb{N}_1$
	\item \label{lg 2} For $A \in N \cap \mathbb{N}_1$ the unique rule for $A$ is of the form
	$$
	A \to u B v C w
		\quad \text{or} \quad
	A \to u B v
		\quad \text{or} \quad
	A \to u  ,
	$$
	where $u$, $v$, and $w$ are (perhaps empty) sequences of handles and $B,C \in N_1$.
	We call $B$ the \emph{first} and $C$ the \emph{second} nonterminal in the rule for $A$, see Figure~\ref{fig:rules}.
	\item \label{lg 3} For $A \in N \cap \mathbb{N}_0$ the rule for $A$ is of the (similar) form
	$$
	A \to u B v C
		\quad \text{or} \quad
	A  \to u B v c
		\quad \text{or} \quad
	A  \to u C
		\quad \text{or} \quad
	A \to u c ,
	$$
	where $u$ and $v$ are (perhaps empty) sequences of handles,
	$c$ is a constant, $B \in N_1$, and $C \in N_0$, see Figure~\ref{fig:rules}.
	Again we speak of the \emph{first} and \emph{second} nonterminal in the rule for $A$.
\end{enumerate}

\begin{figure}
	\centering
	\includegraphics{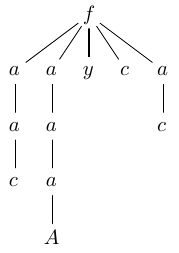}
	\caption{A handle, where $A$ has rank $0$, $c$ is a constant and $a$ is a unary letter.}
	\label{fig:handle}
\end{figure}

Note that the representation of the rules for nonterminals from $\mathbb{N}_0$ is not unique. Take for 
instance the rule $A \to f(B,C)$, which can be written as $A \to h(C)$ for the handle $h(y) = f(B,y)$ 
or as $A \to h'(B)$ for the handle $h' = f(y,C)$.
On the other hand, for nonterminals from $\mathbb{N}_1$  the representation of the rules is unique,
since there is a unique occurrence of the parameter $y$ in the right-hand side.

\begin{figure}
	\centering
	\includegraphics{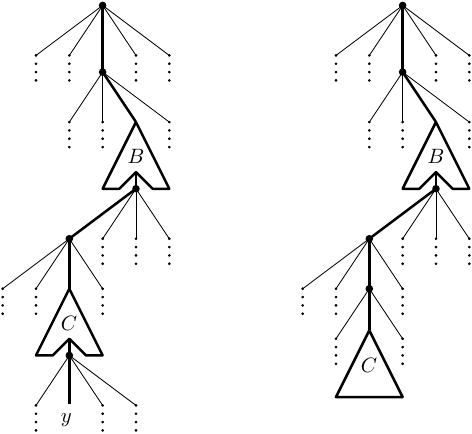}
	\caption{Two possible shapes of right-hand sides in a handle grammar 
	 for a nonterminal of rank $1$ (left) and rank $0$ (right), respectively. The dots symbolise the chains of unary letters
		ended by a nonterminal of rank $0$ or a~constant.}
	\label{fig:rules}
\end{figure}

What is left is to show how to transform an arbitrary SLCF grammar $\grammar'$ into an equivalent handle grammar \grammar.
There is a known construction that transforms an SLCF grammar $\grammar'$ into an equivalent monadic SLCF grammar \grammar~\cite[Theorem~10]{LohreyMS12}
(i.e.\ every nonterminal
of $\grammar'$ has rank $0$ or $1$).
While the original paper~\cite{LohreyMS12} contains only weaker statement,
in fact this construction returns a handle grammar
which has $\Ocomp(|\grammar'|)$ many occurrences of nonterminals of arity 1 in the rules
and $\Ocomp(r |\grammar'|)$ occurrences of nonterminals of arity 0 and letters.
This stronger result is 
repeated in the appendix for completeness.

\begin{lemma}[cf.~{\cite[Theorem~10]{LohreyMS12}}]
\label{lem: monadic grammar}
From a given SLCF grammar $\grammar'$ of size $\grammarsize = |\grammar'|$
one can construct an equivalent handle grammar \grammar{} of size $\Ocomp(r \grammarsize)$ 
with only $\Ocomp(\grammarsize)$ many occurrences of nonterminals of arity 1 in the rules
(and $\Ocomp(r \grammarsize)$ occurrences of nonterminals of arity 0).
\end{lemma}
The construction and proof of~\cite[Theorem~10]{LohreyMS12} yield the claim,
though the actual statement in~\cite{LohreyMS12} is a bit weaker.
For completeness, the proof of this stronger statement is given in the appendix.

\begin{figure}
	\centering
		\includegraphics{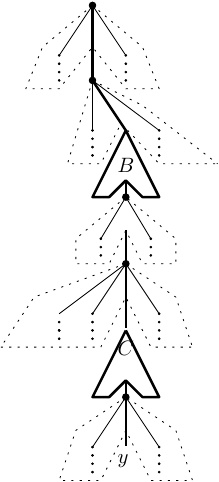}
	\caption{The tree obtained from a rule for nonterminal of rank $1$. Each handle and nonterminal represents a context, each such context is enclosed either by a dotted or fat line, respectively.}
	\label{fig:concatenating}
\end{figure}

When considering  handle grammars it is useful to have some intuition about the trees they derive.
Recall that a \emph{context} is a pattern $t(y)\in \mathcal T(\alphabet,\{y\})$ with a unique occurrence
of the only parameter $y$.
Observe that each nonterminal $A \in \mathbb{N}_1$ derives a unique context $\eval(A)$,
the same applies to a handle $f$ and so we write $\eval(f)$ as well.
Furthermore, we can `concatenate' contexts, so we write them in string notation.
Also, when we attach a tree from $\mathcal{T}(\alphabet)$ to a context, we obtain another tree from $\mathcal{T}(\alphabet)$.
Thus, when we consider a rule $A \to h_1 \cdots h_i B h_{i+1} \cdots h_j C h_{j+1} \cdots h_k$ in a handle grammar (where 
$h_1, \ldots, h_k$ are handles and $A$, $B$, and $C$ are nonterminals of rank 1)
then 
$$\eval(A) = \eval(h_1) \cdots \eval(h_i) \eval(B) \eval(h_{i+1}) \cdots \eval(h_j) \eval(C) \eval(h_{j+1}) \cdots \eval(h_k),$$
i.e., we concatenate the contexts derived by the handles and nonterminals,
see Figure~\ref{fig:concatenating}. 
Similar considerations apply  to other rules of handle grammars as well, also the ones for
nonterminals of rank $0$.

\subsection{Intuition and invariants}
For a given input tree \mytree{} we start (as a mental experiment) with a smallest SLCF grammar generating \mytree{}.
Let $\grammarsize$ be the size of this grammar. We first 
transform it to a handle grammar \grammar{} of size $\Ocomp(\grammarsize r)$ using Lemma~\ref{lem: monadic grammar}.
The number of nonterminals of rank $0$ (resp., $1$) in \grammar{} is bounded by $\Ocomp(\grammarsize r)$ (resp., $\Ocomp(\grammarsize)$).

In the following, by \mytree{} we denote the current tree stored by \algmain.
For analysing the size of the grammar produced by \algmain{} applied to \mytree, we use the accounting method,
see e.g.~\cite[Section 17.2]{CoLeiRivStein09}.
With each occurrence of a letter from $\alphabet$ in \grammar{}'s rules we associate two units of \emph{credit}
(no credit is assigned to occurrences of nonterminals in rules).
During the run of \algmain{} we appropriately modify \grammar,
so that $\eval(\grammar) = \mytree$ (recall that $\mytree$ always denotes the current tree stored by \algmain).
In other words, we perform the compression steps of \algmain{} also on \grammar.
We thereby always maintain the invariant that every occurrence of a letter from $\alphabet$ in \grammar's rules
has two units of credit. In order to do this, we have to {\em issue} (or pay) some new credits during the modifications,
and we have to do a precise bookkeeping on the amount of issued credit.
On the other hand, if we do a compression step in \grammar, then we remove some
occurrences of letters. The credit associated with these occurrences is then {\em released}
and can be used to pay for the representation cost of the new letters introduced by the compression
step. For \paircompression{} and leaf compression, the released credit indeed suffices to pay
the representation cost for the fresh letters, but for chain compression the released credit does
not suffice. Here we need some extra amount that will be estimated separately later on
in Section~\ref{subsec: chain cost}. At the end, we can bound the size of the grammar produced
by \algmain{} by the sum of the initial credit assigned to \grammar,
which is at most $\Ocomp(r \grammarsize)$ by Lemma~\ref{lem: monadic grammar},
plus the total amount of issued credit plus the extra cost estimated in Section~\ref{subsec: chain cost}.

An important difference between our algorithm and the string compression algorithm from~\cite{grammar}, which we generalise,
is that we add new nonterminals to \grammar{} during its modification. To simplify notation, we denote 
with $m$ always the number of nonterminals of the current grammar \grammar, and
we denote its nonterminals with $A_1, \ldots, A_m$. We assume that $i < j$ if
$A_i$ occurs in the right-hand side of $A_j$, and that $A_m$ is the start nonterminal.
With $\alpha_i$ we always denote the current right-hand side of $A_i$. In other words,
the productions of \grammar{} are $A_i \to \alpha_i$ for $1 \leq i \leq m$.

Again note that the modification of \grammar{} is not really carried out by \algmain,
but is only done for the purpose of analysing \algmain.

Suppose a compression step, for simplicity say an $(a,b)$-pair compression,
is applied to \mytree. We should also reflect it in \grammar.
The simplest solution would be to perform the same compression on each of the rules of \grammar,
hoping that in this way all occurrences of $ab$ in $\eval(\grammar)$ are replaced by $c$.
However, this is not always the case. For instance, the 2-chain $ab$ may occur
`between' a nonterminal and a unary letter. This intuition is made precise in Section~\ref{subsec: pair compression}.
To deal with this problem, we modify the grammar, so that the problem disappears.
Similar problems occur also when we want to replace an $a$-maximal chain
or perform leaf compression. The
solutions to those problems are similar and are given in Section~\ref{subsec:block compression}
and Section~\ref{subsec:child compression}, respectively.

To ensure that \grammar{} stays handle and to estimate the amount of issued credit,
we show that the grammar preserves the following invariants, where
$n_0 = \Ocomp(\grammarsize r)$ (respectively, $n_1 = \Ocomp(\grammarsize)$) is the  initial number of occurrences of nonterminals from $\mathbb{N}_0$ (respectively, $\mathbb{N}_1$) in
\grammar{} while $\grammarsizezero$ and $\grammarsizeone$ are those values at some particular moment.
Similarly, $\grammarsizetilde$ is the number of occurrences of nonterminals from $\widetilde{N_0}$.

\begin{enumerate}[({GR}1)]
	\item \label{Gr 1} \grammar{} is handle.
	\item \label{Gr 2} \grammar{} has nonterminals $N_0 \cup N_1 \cup \widetilde{N_0}$, where $\widetilde{N_0} \cup N_0 \subseteq \mathbb{N}_0$,
	$|N_0| \leq n_0$ and $N_1 \subseteq \mathbb{N}_1$, $|N_1| \leq n_1$.
	\item \label{Gr 3} The number $g_0$ of occurrences nonterminals from $N_0$ in \grammar{} never increases  (and is initially $n_0$), and
	the number $g_1$ of occurrences of nonterminals from $N_1$ also never increases (and is initially $n_1$).
	\item \label{Gr 4} The number $\grammarsizetilde$ of occurrences of nonterminals from $\widetilde{N_0}$ in \grammar{}
	is at most $n_1(r-1)$.
	\item \label{Gr 5} The rules for $A_i \in \widetilde{N_0}$ are of the form $A_i \to w A_j$ or $A_i \to w c$,
	where $w$ is a string of unary symbols, $A_j \in N_0 \cup \widetilde{N_0}$, and $c$ is a constant.
\end{enumerate}
Intuitively, $N_0$ and $N_1$ are subsets of the initial nonterminals of rank $0$ and $1$, respectively,
while $\widetilde{N_0}$ are the nonterminals introduced by \algmain, which are all of rank $0$.

Clearly, \GRrefall{} hold for the initial handle grammar \grammar{} obtained by Lemma~\ref{lem: monadic grammar}.

\subsection{($\lettersup,\lettersdown)$-compression}
\label{subsec: pair compression}
We begin with some necessary definitions that help to classify $2$-chains.
For a non-empty tree or context $t$ its \emph{first} letter is the letter that labels the root of $t$.
For a context $t(y)$ which is not a parameter its \emph{last} letter is the label of the node above the one labelled with $y$.
For instance, the last letter of the context $a(b(y))$ is $b$ and the last letter of the context $f(a(c),y)$ is $f$, which is also
the first letter.

A chain pattern $ab$ has a \emph{crossing occurrence} in a nonterminal $A_i$ if one of the following holds:
\begin{enumerate}[({CR}1)]
	\item \label{cr 1} $aA_j$ is a subpattern of $\alpha_i$ and the first letter of $\eval(A_j)$ is $b$
	\item \label{cr 2} $A_j(b)$ is a subpattern of $\alpha_i$ and the last letter of $\eval(A_j)$ is $a$
	\item \label{cr 3} $A_j(A_k)$ is a subpattern of $\alpha_i$, the last letter of $\eval(A_j)$ is $a$
	and the first letter of $\eval(A_k)$ is $b$.
\end{enumerate}
A chain pattern $ab$ is \emph{crossing} if it has a crossing occurrence in any nonterminal and \emph{non-crossing} otherwise.
Unless explicitly written, we use this notion only in case $a \neq b$.

When every chain pattern $ab \in \lettersup\lettersdown$ is noncrossing,
simulating $(\lettersup,\lettersdown)$-com\-pression on \grammar{} is easy:
It is enough to apply $\algtreepaircomp$ (Algorithm~\ref{alg:treepaircomp}) to each 
right-hand side of \grammar{}.
We denote the resulting grammar with  $\algtreepaircomp(\lettersup,\lettersdown,\grammar)$.

In order to distinguish between the nonterminals, grammar, etc.\
before and after the application of \algtreepaircomp{} (or, in general, any procedure) 
we use `primed' symbols, i.e.,  $A_i'$, $\grammar'$, \treeci{}
for the nonterminals, grammar and tree, respectively, after the compression step and `unprimed' symbols 
(i.e., $A_i$, \grammar{}, \mytree{}) for the ones before.

\begin{lemma}
\label{lem:noncrossing compression}
Let $\grammar$ be a handle grammar
and $\grammar' = \algtreepaircomp(\lettersup,\lettersdown,\grammar)$.
Then the following hold:
\begin{itemize}
\item If \grammar{} satisfies \GRrefall{} then $\grammar'$ satisfies~\GRrefall{} as well. 
\item If there is no crossing chain pattern from $\lettersup\lettersdown$ in \grammar, then
$$\eval(\grammar') = \algtreepaircomp(\lettersup,\lettersdown,\eval(\grammar)).$$
\item The grammar $\grammar'$ has the same number of occurrences of nonterminals of each rank as $\grammar$.
\item The credit for new letters in $\grammar'$ and the cost of representing these new letters are paid by the released credit.
\end{itemize}
\end{lemma}
 
\begin{proof}
Clearly, $\eval(\grammar')$ can be obtained from $\eval(\grammar)$ by compressing some occurrences of patterns from 
$\lettersup\lettersdown$.  Hence, to show that $\eval(\grammar') = \algtreepaircomp(\lettersup,\lettersdown,\eval(\grammar))$, it suffices to show that
$\eval(\grammar')$ does not contain occurrences of patterns from $\lettersup\lettersdown$. 
By induction on $i$ we show that for every $1 \leq i \leq m$, 
$\eval(A_i')$ does not contain occurrences of patterns from $\lettersup\lettersdown$. 
To get a contradiction, consider an occurrence of $ab \in \lettersup\lettersdown$ in $\eval(A'_i)$.
If it is generated by an explicit occurrence of $ab$ in the right-hand side of $A_i'$
then it was present already in the rule for $A_i$, since we do not introduce new occurrences of the letters
from $\grammar$. So, the occurrence of $ab$ is replaced by a new letter in $\grammar'$.
If the occurrence is contained within the subtree generated by some $A'_j$ ($j < i$), then the occurrence is compressed
by the inductive assumption.
The remaining case is that there exists a  crossing occurrence of $ab$ in the rule for $A'_i$.
However note that if $a$ is the first (or $b$ is the last) letter of $\eval(A'_j)$,
then it was also the first (respectively, last) letter of $\eval(A_j)$ in the input instance,
as we do not introduce new occurrences of the old letters.
Hence, the occurrence of $ab$ was crossing already in the input grammar \grammar, which is not possible by the assumption of the lemma.

Each  occurrence of $ab \in \lettersup\lettersdown$ has 4 units of credit (two for each symbol), which are released in
the compression step. Two of the released units are used to pay for the credit of the new occurrence of the symbol $c$
(which replaces the occurrence of $ab$), while the other two units are used to pay for the representation cost of  $c$
(if we replace more than one occurrence of $ab$ in $\grammar$, some credit is wasted).

Let us finally argue that the invariants \GRrefall{} are preserved:
Replacing an occurrence of $ab$ with a single unary letter $c$ cannot make a handle grammar
a non-handle one, so~\GRref{1} is preserved.
Similarly, \GRref{5} is preserved.
The set of nonterminals and the number of occurrences of the nonterminals is unaffected, so also \GRref{2}--\GRref{4} are preserved.
\qedhere
\end{proof}

By Lemma~\ref{lem:noncrossing compression} it is left to assure that indeed all occurrences of chain patterns from
$\lettersup\lettersdown$ are noncrossing.
What can go wrong? Consider for instance the grammar with the rules $A_1(y) \to a(y)$ and  $A_2 \to A_1(b(c))$.
The pattern $ab$ has a crossing occurrence.
To deal with crossing occurrences we change the grammar.
In our example, we replace $A_1$ by $a$ in the right-hand side of $A_2$,
leaving only $A_2 \to ab(c)$, which does not contain a crossing occurrence of $ab$.

Suppose that some $ab \in \lettersup\lettersdown$ is crossing because of~\CPref{1}.
Let $aA_i$ be a subpattern of some right-hand side and let $\eval(A_i) = bt'$.
Then it is enough to modify the rule for $A_i$
so that $\eval(A_i) = t'$ and replace each occurrence of $A_i$ in a right-hand side by $bA_i$.
We call this action \emph{popping-up $b$ from $A_i$}.
The similar operation of \emph{popping down} a letter $a$ from $A_i \in N \cap \mathbb{N}_1$ is symmetrically defined
(note that both pop operations apply only to unary letters). See Figure~\ref{fig:uncrpair} for an example.
A similar operation of popping letters in the context of tree grammars is used also in~\cite{DBLP:conf/icde/BottcherHJM16}.

The lemma below shows that popping up and popping down removes all crossing occurrences of $ab$.
Note that the operations of popping up and popping down can be performed for several letters in parallel:
The procedure $\algpop(\lettersup,\lettersdown,\grammar)$ below  `uncrosses' all occurrences of patterns
from the set $\lettersup\lettersdown$, assuming that $\lettersup$ and $\lettersdown$ are disjoint
subsets of $F_1$ (and we apply it only in the cases in which they are disjoint).

Recall that for a handle grammar, right-hand sides can be viewed as sequences of nonterminals
and handles. Hence, we can speak of the first (respectively, last) symbol of a right-hand side.

\begin{algorithm}[H]
  \caption{$\algpop(\lettersup,\lettersdown,\grammar)$: Popping letters from $\lettersup$ and $\lettersdown$
  \label{pop code full}}
  \begin{algorithmic}[1]
  \For{$i \gets 1 \twodots m-1$}
		\If{the first symbol of $\alpha_i$ is $b \in \lettersdown$} \label{still to the right}  \Comment{popping up $b$}
			\If{$\alpha_i = b$}
				\State replace $A_i$ in all right-hand sides of  $\grammar$ by $b$ \label{remove-b}
			\Else
				\State remove this leading $b$ from $\alpha_i$
				\State replace $A_i$ in all right-hand sides of $\grammar$ by $bA_i$
			\EndIf
		\EndIf
		\If{$A_i\in N_1$ and  the last symbol of $\alpha_i$ is $a \in \lettersup$} \label{still to the left} \Comment{popping down $a$}
			\If{$\alpha_i = a$}
				\State replace $A_i$ in all right-hand sides of $\grammar$ by $a$ \label{remove-a}
			\Else
				\State remove this final $a$ from $\alpha_i$
				\State replace $A_i$ in all right-hand sides of  $\grammar$  by $A_ia$
			\EndIf
		\EndIf
	\EndFor
  \end{algorithmic}
\end{algorithm}

\begin{figure}
	\centering
		\includegraphics{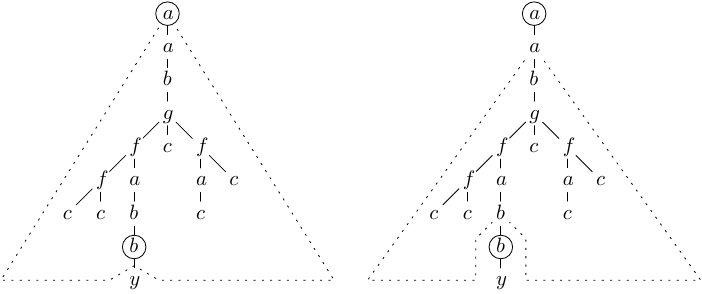}
	\caption{A tree before and after $\algpop(\{a\},\{b\}, \cdot)$. Affected nodes are encircled, the dotted lines enclose the patterns generated by a nonterminal.}
	\label{fig:uncrpair}
\end{figure}

\begin{lemma}
\label{lem:uncrossing pairs}
Let \grammar{} be a handle grammar and $\grammar' = \algpop(\lettersup,\lettersdown,\grammar)$,
where $\lettersup \cap \lettersdown = \emptyset$.
Then, the following hold:
\begin{itemize}
\item $\eval(A_m') = \eval(A_m)$ and hence $\eval(\grammar) = \eval(\grammar')$.
\item All chain patterns from $\lettersup\lettersdown$ are non-crossing in 
$\grammar'$.
\item If \grammar{} satisfies~\GRrefall{}, then so does $\grammar'$.
\item The grammar $\grammar'$ has at most the same number of occurrences of nonterminals of each rank as $\grammar$.
\item During the computation of $\grammar'$ from \grammar{}, at most two letters are popped  for each occurrence of 
a nonterminal of rank $1$ and at most one letter is popped for each occurrence of 
a nonterminal of rank $0$.
In particular, if \grammar{} satisfies \GRrefall{} then at most $4 \grammarsizeone + 2\grammarsizezero + 2\grammarsizetilde$
units of credit are issued during the computation of $\grammar'$.
\end{itemize}
\end{lemma}

\begin{proof}
Observe first that whenever we pop up $b$ from some $A_i$, then we replace each of $A_i$'s occurrences in \grammar{} with $bA_i$
(or with $b$, when $\eval(A_i) = b$), and similarly for the popping down operation, thus the value of $\eval(A_j)$ is not changed for $j \neq i$.
Hence, in the end we have $\eval(A_m') = \eval(A_m) = \mytree$ (note that $A_m$ does not pop letters).

Secondly, we show that if the first letter of $\eval(A_i')$ (where $i < m$) is $b' \in \lettersdown$ then we popped-up a letter from $A_i$
(which by the code is some $b \in \lettersdown$);
a similar claim holds by symmetry for the last letter of $\eval(A_i)$.
So, suppose that the claim is not true and consider the nonterminal $A_i$ with the smallest $i$ such that
the first letter of $\eval(A_i')$ is $b' \in \lettersdown$
but we did not pop up a letter from $A_i$.
Consider the first symbol of $\alpha_i$ when \algpop{} considered $A_i$
in line~\ref{still to the right}. Note, that as \algpop{} did not pop up a letter from $A_i$, the first letter of $\eval(A_i)$ and
$\eval(A_i')$ is the same and hence it is $b' \in \lettersdown$.
So $\alpha_i$ cannot begin with a letter as then it is $b' \in \lettersdown$ which should have been popped-up.
Hence, the first symbol of $\alpha_i$ is some nonterminal $A_j$ for $j < i$.
But then the first letter of $\eval(A_j')$ is $b' \in \lettersdown$ and so by the inductive assumption \algpop{} popped-up a~letter from $A_j$.
Hence, $\alpha_i$ begins with a letter when $A_i$ is considered in line~\ref{still to the right}. We obtained a  contradiction.

Suppose now that after \algpop{} there is a crossing pattern $ab \in \lettersup\lettersdown$.
This is due to one of the bad situations \CPref{1}--\CPref{3}.
We consider only \CPref{1}; the other cases are dealt  in a similar fashion. Hence, 
assume that $aA'_i$ is a subpattern in a right-hand side of $\grammar'$ and the first letter of $\eval(A'_i)$ is $b$.
Note that as $a \notin \lettersdown$ is labelling the parent node of an occurrence of $A'_i$ in $\grammar'$,
$A_i$ did not pop up a letter.
But the first letter of $\eval(A_i')$ is $b \in \lettersdown$. So, $A_i$ should have popped up a~letter by our earlier claim, which is a~contradiction.

Note that \algpop{} introduces at most two new letters for each occurrence of a~nonterminal of rank 1,
so from $N_1$ (one letter popped up and one popped down),
and at most one new letter for each occurrence of a nonterminal of rank 0, so from  $N_0 \cup \widetilde{N}_0$ (as nonterminals of rank $0$ cannot pop down a letter).
As each letter has two units of credit, the estimation on the number of issued credit follows from \GRrefall.

Concerning the preservation of the invariants, note that $\algpop$ does not introduce new nonterminals or new 
occurrences of existing nonterminals (occurrences of nonterminals can be eliminated in line~\ref{remove-b} and~\ref{remove-a}).
Therefore, \GRref{2}--\GRref{4} are preserved.
Moreover, also the form of the productions guaranteed by \GRref{1} and \GRref{5} cannot be spoiled,
so \GRref{1} and \GRref{5} are preserved as well.
\qedhere
\end{proof}

Hence, to simulate $(\lettersup, \lettersdown)$-compression on \grammar{} 
it is enough to first uncross all 2-chains from $\lettersup\lettersdown$
and then compress them all using $\algtreepaircomp(\lettersup,\lettersdown,\grammar)$.

\begin{lemma}
\label{lem:pc crossing}
Let \grammar{} be a handle grammar and let
$$
\grammar' =   \algtreepaircomp(\lettersup,\lettersdown, \algpop(\lettersup,\lettersdown,\grammar)).
$$
Then the following hold:
\begin{itemize}
\item $\eval(\grammar') =  \algtreepaircomp(\lettersup,\lettersdown, \eval(\grammar))$
\item If \grammar{} satisfies \GRrefall{},  then so does  $\grammar'$.
\item The grammar $\grammar'$ has at most the same number of occurrences of nonterminals of each rank as $\grammar$.
\item At most two new occurrences of letters are introduced for each 
occurrence of a nonterminal of rank $1$, and at most one new occurrence of a letter is introduced for 
each  occurrence of a nonterminal of rank $0$.
In particular, if \grammar{} satisfies \GRrefall{} then at most $4 \grammarsizeone + 2 \grammarsizezero + 2 \grammarsizetilde$ units of credit are issued during the computation of $\grammar'$.
\item The issued credit and the credit released by $ \algtreepaircomp$ cover the representation cost of fresh letters as well as 
their credit in $\grammar'$.
\end{itemize}
\end{lemma}

\begin{proof}
By Lemma~\ref{lem:uncrossing pairs}, every chain pattern from $\lettersup\lettersdown$ is non-crossing
in $\algpop(\lettersup,\lettersdown,\grammar)$. We get 
\begin{eqnarray*}
\eval(\grammar') &=& \eval(\algtreepaircomp(\lettersup,\lettersdown, \algpop(\lettersup,\lettersdown,\grammar))) \\
& \stackrel{\text{Lemma~\ref{lem:noncrossing compression}}}{=} & \algtreepaircomp(\lettersup,\lettersdown, \eval(\algpop(\lettersup,\lettersdown,\grammar))) \\
& \stackrel{\text{Lemma~\ref{lem:uncrossing pairs}}}{=} & \algtreepaircomp(\lettersup,\lettersdown,\eval(\grammar)) .
\end{eqnarray*}
Moreover, at most $4 \grammarsizeone + 2 \grammarsizezero + 2 \grammarsizetilde$ units of credit are issued
and this is twice the number of occurrences of new
letters in the grammar.
By Lemma~\ref{lem:noncrossing compression} and~\ref{lem:uncrossing pairs} no new occurrences of nonterminals are introduced.
By Lemma~\ref{lem:noncrossing compression} the cost of representing new letters introduced by \algtreepaircomp{}
is covered by the released credit.
Finally, both \algtreepaircomp{} and \algpop{} preserve the invariants \GRrefall.
\qedhere
\end{proof}

Since by Lemma~\ref{lem:number of phases} we apply at most $\Ocomp(\log n)$ many $(\lettersup,\lettersdown)$-compressions
(for different sets $\lettersup$ and $\lettersdown$) and by \GRref{3}--\GRref{4} we have $\grammarsizezero \leq n_0$,
$\grammarsizetilde \leq n_1(r-1)$ and $\grammarsizeone \leq n_1$, we obtain.
\begin{corollary}
\label{cor: credit pair compression}
$(\lettersup,\lettersdown)$-compression issues in total $\Ocomp((n_0 + n_1r) \log n)$ units of credit
during all modifications of \grammar.
\end{corollary}

\subsection{Chain compression}
\label{subsec:block compression}
Our notations and analysis for chain compression is similar to those for 
$(\lettersup,\lettersdown)$-compression.
In order to simulate chain compression on \grammar{} we want to apply
\algtreechaincomp{} (Algorithm~\ref{alg:treechaincomp}) to the right-hand sides of \grammar.
This works as long as there are no crossing chains:
A unary letter $a$ \emph{has a crossing chain} in a rule $A_i \to \alpha_i$ if $aa$
has a crossing occurrence in $\alpha_i$, otherwise $a$ has no crossing chain.
As for $(\lettersup,\lettersdown)$-compression,
when there are no crossing chains, we apply \algtreechaincomp{} to the right-hand sides of \grammar.
We denote with $\algtreechaincomp(\letters_1,\grammar)$ the grammar obtained
by applying $\algtreechaincomp$ to all right-hand sides of \grammar.

\begin{lemma}
	\label{lem: blockscomp}
Let \grammar{} be a handle grammar and $\grammar' = \algtreechaincomp(\letters_1,\grammar)$. Then the following hold:
\begin{itemize}	
\item If no unary letter from $\letters_1$ has a crossing chain in a rule of \grammar,
then $$\eval(\grammar') = \algtreechaincomp(\letters_1,\eval(\grammar)).$$
\item The grammar $\grammar'$ has the same number of occurrences of nonterminals of each rank as $\grammar$.
\item If \grammar{} satisfies \GRrefall{}, then so does  $\grammar'$.
\end{itemize}
\end{lemma}

\begin{figure}
	\centering
		\includegraphics{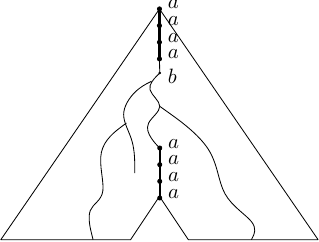}
	\caption{A context with its $a$-prefix and $a$-suffix depicted.}
	\label{fig:prefixsuffix}
\end{figure}

The proof is similar to the proof of Lemma~\ref{lem:noncrossing compression}
and so it is omitted.
Note that so far we have neither given a bound on the amount of issued credit
nor on the representation cost for the new letters $a_\ell$ introduced by $\algtreechaincomp$.
Let us postpone these points and first show how to ensure that no letter has a crossing chain.
The solution is similar to \algpop: Suppose for instance that $a$ has a crossing chain
due to \CPref{1}, i.e., some $aA_i$ is a subpattern in a right-hand side and $\eval(A_i)$ begins with $a$.
Popping up $a$ does not solve the problem, since after popping, $\eval(A_i)$ might still begin with $a$.
Thus, we keep on popping up until the first letter of $\eval(A_i)$ is not $a$,
see Figure~\ref{fig:uncrchain}.
In order to do this in one step we need some notation:
We say that $a^\ell$ is the \emph{$a$-prefix} of a tree (or context)  $t$ if $t = a^\ell t'$ and the first letter of $t'$ is not $a$
(here $t'$ might  be the trivial context $y$), see Figure~\ref{fig:prefixsuffix}.
In this terminology, we remove the $a$-prefix of $\eval(A_i)$.
Similarly, we say that $a^\ell$ is the \emph{$a$-suffix} of a context  $t(y)$ if $t = t'(a^\ell(y))$ for a context $t'(y)$ 
and the last letter of $t'$ is not $a$
(again, $t'$ might  be the trivial context $y$ and then $a^\ell$ is also the $a$-prefix of $t$).
The following algorithm \algremblocks{} eliminates crossing chains from \grammar{}.

\begin{figure}
	\centering
		\includegraphics{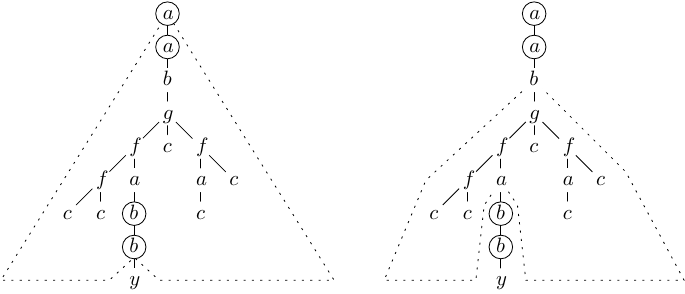}
	\caption{A tree before and after \algremblocks. Affected nodes are encircled, the dotted lines enclose the patterns generated by a nonterminal.}
	\label{fig:uncrchain}
\end{figure}

\begin{algorithm}[H]
  \caption{$\algremblocks(\grammar)$: removing crossing chains.
  \label{removing outer letters}}
  \begin{algorithmic}[1]
  \For{$i \gets 1 \twodots m-1$}
		\If{the first letter $a$ of $\eval(A_i)$ is unary}
			\State let $p$ be the length of the $a$-prefix of $\alpha_i$ \label{length-l}
			\If{$\alpha_i = a^p$}
				\State replace $A_i$ in all right-hand sides by $a^p$
			\Else
				\State remove $a^p$ from the beginning of $\alpha_i$
				\State replace $A_i$ by $a^p A_i$ in all right-hand sides \label{popping prefix}
			\EndIf
		\EndIf	
		\If{$A_i \in N_1$ and the last letter $b$ of $\eval(A_i)$ is unary}
			\State let $s$ be the length of the $b$-suffix of $\alpha_i$  \label{length-r}
			\If{$\alpha_i = b^s$}
				\State replace $A_i$ in all right-hand sides by $b^s$
			\Else
				\State remove $b^s$ from the end of $\alpha_i$
				\State replace $A_i$ by $A_ib^s$ in all right-hand sides\label{popping suffix}
			\EndIf
		\EndIf
	\EndFor
	\end{algorithmic}
\end{algorithm}

\begin{lemma} \label{lem: uncrossing chains}
Let \grammar{} be a handle grammar and $\grammar'$ =  \algremblocks(\grammar). Then the following hold:
\begin{itemize}
\item $\eval(A_m) = \eval(A_m')$ and hence $\eval(\grammar) = \eval(\grammar')$.
\item No unary letter has a crossing chain in $\grammar'$.
\item If \grammar{} satisfies \GRrefall{}, then so does $\grammar'$.
\item The grammar $\grammar'$ has at most the same number of occurrences of nonterminals of each rank as $\grammar$.
\end{itemize}
\end{lemma}

\begin{proof}
First observe that whenever we remove the $a$-prefix $a^{p_i}$ from the rule for $A_i$
we replace each occurrence of  $A_i$  by $a^{p_i} A_i$ and similarly for 
$b$-suffixes. Hence, as long as $A_j$ is not yet considered, it defines the same tree as in the input tree.
In particular, after \algremblocks{} we have $\eval(A_m') = \eval(A_m)$,
as we do not pop prefixes and suffixes from $A_m$.

Next, we show that when $\algremblocks$ considers $A_i$, then $p$ from line~\ref{length-l}
is the length of the $a$-prefix of $\eval(A_i)$ (similarly, $s$ from line~\ref{length-r} is the length of the 
$a$-suffix of $\eval(A_i)$).
Suppose that this is not the case and consider $A_i$ with smallest $i$ which violates the statement.
Clearly $i > 1$ since there are no nonterminals in the right-hand side for $A_1$.
Let $a^k$ be the $a$-prefix of $\eval(A_i)$. We have $p < k$.
The symbol below  $a^p$ in $\alpha_i$ (which must exist because otherwise $\eval(A_i) = a^p$) cannot be a letter (as the $a$-prefix of $\eval(A_i)$ is not $a^p$),
so it is a nonterminal $A_j$ with $j < i$.
The first letter of $\eval(A_j)$ must be $a$. Let $a^{k'}$ be the $a$-prefix of $\eval(A_j)$.
By induction, $A_j$ popped up $a^{k'}$, and at the time when $A_i$ is considered, the first letter 
of $\eval(A_j)$ is different from $a$.
Hence, the $a$-prefix of $\eval(A_i)$ is exactly $a^p$, a contradiction.

As a consequence of the above statement, if $aA'_i$ occurs in a right-hand side of the output grammar $\grammar'$, then $a$ is not the first letter of $\eval(A'_i)$.
This shows that \CPref{1} cannot hold for a chain pattern $aa$. 
The conditions \CPref{2} and \CPref{3} are handled similarly.
So there are no crossing chains after \algremblocks.

The arguments for the last two points of the lemma are the same as in the proof of Lemma~\ref{lem:uncrossing pairs}
and therefore omitted.
\end{proof}

So chain compression  is done by first running \algremblocks{}
and then \algtreechaincomp{} on the right-hand sides of \grammar.

\begin{lemma}
\label{lem:blocksc}
Let \grammar{} be a handle grammar and $\grammar'= \algtreechaincomp(\letters_1,\algremblocks(\grammar))$.
Then, the following hold:
\begin{itemize}
	\item If \grammar{} satisfies \GRrefall{}, then so does $\grammar'$.
\item $\eval(\grammar') = \algtreechaincomp(\letters_1, \eval(\grammar))$
\item During the computation of $\grammar'$ from \grammar{} at 
most  two new occurrences of letters are introduced for each 
occurrence of a nonterminal of rank $1$, and at most one
occurrence of a letter is introduced 
 for each occurrence of a nonterminal of rank $0$.
In particular, if \grammar{} satisfies \GRref{3}--\GRref{4} then at most  $4 \grammarsizeone + 2 \grammarsizezero + 2 \grammarsizetilde$ units of credit are issued in the computation of $\grammar'$ and
this credit is used to pay the credit for the fresh letters $a_\ell$ in the grammar introduced by \algtreechaincomp{}
(but not their representation cost).
\end{itemize}
\end{lemma}
\begin{proof}
	We shall comment only on the amount of new letters
	and the issued credit,
	as the rest follows from Lemma~\ref{lem: blockscomp} and~\ref{lem: uncrossing chains}.
	Note that the arbitrarily long chains popped by \algremblocks{}  
	are compressed into single letters by  \algtreechaincomp{}. Hence,
	as for $(\lettersup,\lettersdown)$-compression, 
	at most two letters are introduced for each occurrence of a nonterminal of rank 1 (i.e., from $N_1$) and
	at most one letter is introduced  for each occurrence of a nonterminal of rank 0 (i.e.,  from $N_0 \cup \widetilde{N}_0$).
	The \GRref{3}--\GRref{4} additionally bound the number of occurrences of nonterminals,
	so assuming \GRref{3}--\GRref{4} yields the bound on the amount of issued credit.	
\end{proof}
Since by Lemma~\ref{lem:number of phases} we apply at most $\Ocomp(\log n)$ many chain compressions to \grammar{}
and by \GRref{3}--\GRref{4} we have $\grammarsizezero \leq n_0$,
$\grammarsizetilde \leq n_1(r-1)$ and $\grammarsizeone \leq n_1$, we get:

\begin{corollary}
\label{cor: credit chain compression}
Chain compression issues in total $\Ocomp((n_0 + n_1r) \log n)$ units of credit during all modifications of \grammar.
\end{corollary}
The total representation cost for the new letters $a_\ell$ introduced by chain compression is estimated separately in Section~\ref{subsec: chain cost}.

\subsection{Leaf compression}
\label{subsec:child compression}

In order to simulate leaf compression on \grammar{} we perform similar operations as in the case of $(\lettersup,\lettersdown)$-compression:
Ideally we would like to apply \algtreechildcomp{} to each rule of \grammar. 
However, in some cases this does not return the appropriate result.
We say that the pair $(f,a)$ is a \emph{crossing \parentchild} in \grammar,
if $f \in F_{\geq 1}$, $a \in F_0$, and one of the following cases holds:
\begin{enumerate}[({FC}1)]
	\item \label{fc 1} $f(t_1,\ldots,t_\ell)$ is a subtree of some right-hand side of \grammar{}, where for some $j$ we have $t_j = A_k$ and $\eval(A_k) = a$.
	\item \label{fc 2} For some $A_i \in N_1$, $A_i(a)$ is a subtree of some right-hand side of \grammar{} and the last letter of $\eval(A_i)$ is $f$.
	\item \label{fc 3} For some $A_i \in N_1$ and $A _k \in N_0$, $A_j(A_k)$ is a subtree of some right-hand side of \grammar{},
	the last letter of $\eval(A_i)$ is $f$, and $\eval(A_k) = a$.
\end{enumerate}
When there is no crossing \parentchild{},
we proceed as in the case of any of the two previous compressions:
We apply \algtreechildcomp{} to each right-hand side of a rule. We denote the resulting grammar with 
$\algtreechildcomp(\letters_{\geq 1},\letters_0,\grammar)$.

\begin{lemma} \label{lem: no crossing parent node child compression}
Let \grammar{} be a handle grammar and  $\grammar ' = \algtreechildcomp(\letters_{\geq 1},\letters_0,\grammar)$. Then, the following hold:
\begin{itemize}
\item If there is no crossing \parentchild{} in \grammar,
then $\eval(\grammar') = \algtreechildcomp(\letters_{\geq 1},\letters_0,\eval(\grammar))$.
\item The cost of representing  new letters and the credits for those letters are covered by the released credit.
\item If \grammar{} satisfies \GRrefall{}, then so does $\grammar'$.
\item The grammar $\grammar'$ has the same number of occurrences of nonterminals of each rank as $\grammar$.
\end{itemize}
\end{lemma}

\begin{proof}
Most of the proof follows similar lines as the proof of Lemma~\ref{lem:noncrossing compression},
but there are some small differences.

Let us first prove that $\eval(\grammar') = \algtreechildcomp(\letters_{\geq 1},\letters_0,\eval(\grammar))$
under the assumption that there is no crossing \parentchild{} in \grammar.
As in the proof of Lemma~\ref{lem:noncrossing compression} it suffices to show that
$\eval(\grammar')$ does not contain a subtree of the form $f(t_1,\ldots,t_k)$ with $f \in F_{\geq 1}$
such that there exist positions $1 \leq i_1 < i_2 <\cdots < i_\ell \leq k$ ($\ell \geq 1$)
and constants  $a_1,\ldots,a_\ell \in F_0$  with $t_{i_j} =a_j$ for $1 \leq j \leq \ell$ and $t_i \notin \letters_0$ for 
$i \not\in \{i_1, \ldots, i_\ell\}$ (note that the new letters introduced by $\algtreechildcomp$ do not belong to 
the alphabet $F$). Assume that such a subtree exists in $\eval(A'_i)$. Using induction, we deduce a 
contradiction. 
If the root $f$ together with its children at positions $i_1,\ldots,i_\ell$ are generated by some other nonterminal $A'_j$ occurring in 
the right-hand side of $A'_i$,
then these nodes are compressed by the induction assumption.
If they all occur explicitly in the right-hand side, then they are compressed by $\algtreechildcomp(\letters_{\geq 1},\letters_0,\grammar)$.
The only remaining case is that $\grammar'$ contains a crossing \parentchild.
But then, since $f, a_1, \ldots, a_\ell$ are old letters, this crossing \parentchild{} must be already
present in \grammar, which contradicts the assumption from the lemma.

Concerning the representation cost for the new letters,
observe that when $f$ and $\ell$ of its children are compressed, the representation cost for the new letter is $\ell + 1$.
There is at least one occurrence of $f$ with those children in a right-hand side of \grammar. Before the compression these nodes held $2(\ell + 1)$ units of credit.
After the compression, only two units are needed for the new node. 
The other $2\ell \geq \ell + 1$ units are enough to pay for the representation cost.

Concerning the preservation of the invariants,
observe that no new nonterminals were introduced, so \GRref{2}--\GRref{4} are preserved.
Also the form of the rules for $A_i \in \widetilde{N}_0$ cannot be altered
(the only possible change affecting those rules is a replacement of $ac$,
where $a \in F_1$ and $c \in F_0$, by a new letter $c' \in F_0$).

So it is left to show that the resulting grammar is handle.
It is easy to show that after a leaf compression a handle is still a handle,
with only one exception: assume we have a handle $h=f(w_1\gamma_1, \ldots, w_{j-1}\gamma_{j-1},y,w_{j+1}\gamma_{j+1}, \ldots, w_\ell\gamma_\ell)$
followed by a constant $c$ in the right-hand side $\alpha_i$ of $A_i$. Such a situation can only occur, if 
$A_i \in  N_0$ and $\alpha_i$ is of the form $vc$ or $uBvc$ for sequences of handles $u$ and $v$, where 
$v= v'h$ for a possibly empty sequence of handles $v'$ (see (HG\ref{lg 3})).
Then leaf compression merges the constant $c$ into the $f$ from the handle $h$. 
There are two cases: If all $w_k$ (which are chains) are empty and all $\gamma_k$ are constants from $\letters_0$,
then the resulting tree after leaf compression is a constant and no problem arises. Otherwise, we obtain
a tree of the form $f'(w'_1\gamma'_1, \ldots, w'_{\ell'}\gamma'_{\ell'})$, where every $w'_k$ is a chain, and every
$\gamma'_k$ is either a constant or a nonterminal of rank $0$. We must have $\ell' > 0$ (otherwise, this is in fact the first case).
Therefore, $f'(w'_1\gamma'_1, \ldots, w'_{\ell'}\gamma'_{\ell'})$ can be written (in several ways) as a handle, followed by 
(a possibly empty) chain, followed by a constant or a nonterminal of rank $0$. For instance, we can
write the rule for $A_i$ as 
$A_i \to v' f'(y, w'_2 \gamma'_2, \ldots, w'_{\ell'}\gamma'_{\ell'}) w'_1\gamma'_1$ or 
$A_i \to u B v' f'(y, w'_2 \gamma'_2, \ldots, w'_{\ell'}\gamma'_{\ell'}) w'_1\gamma'_1$
(depending on the form of the original rule for $A_i$).
This rule has one of the forms from (HG\ref{lg 3}), which concludes the proof.
Note that we possibly add a second nonterminal to the right-hand side of $A_i \in N_0$ 
in the second case (as $\gamma_1'$ can be a non-terminal), which is allowed in (HG\ref{lg 3}).
\qedhere
\end{proof}

If there are crossing \parentchild s, then we uncross them all by a generalisation of the \algpop{} procedure.
Observe that in some sense we already have a partition: We want to pop up letters from $\letters_0$
and pop down letters from $\letters_{\geq 1}$.
The latter requires some generalisation, becasue when we pop down a letter, it may have rank greater than $1$
and so we need to in fact pop a whole handle.
This adds new nonterminals to \grammar{} as well as a large number of new letters and hence a large amount of credit, so we need to be careful.
There are two crucial details:
\begin{itemize}
	\item When we pop down a whole handle $h=f(t_1,\ldots,t_k,y,t_{k+1},\ldots,t_\ell)$, we add to the set
	$\widetilde{N}_0$ fresh nonterminals
	 for all trees $t_i$ that are non-constants, replace these $t_i$ in $h$ by their corresponding nonterminals
	 and then pop down the resulting handle.
	In this way on one hand we keep the issued credit small and on the other no new occurrence of nonterminals from $N_0 \cup N_1$ are created.
\item We do not pop down a handle from every nonterminal, but do it only when it is needed, i.e., 
	if for $A_i \in N_1$ one of the cases \FCref{2} or \FCref{3} holds. 
	This allows preserving \GRref{5}.
	Note that when the last symbol in the rule for $A_i$ is not a handle but another nonterminal,
	this might cause a need for recursive popping.
	So we perform the whole popping down in a depth-first-search style.
\end{itemize}
Our generalised popping procedure is called \algremchild{} (Algorithm~\ref{code-algremchild}) and is shown
in Figure~\ref{fig:uncrleaf}.

\begin{figure}
	\centering
		\includegraphics{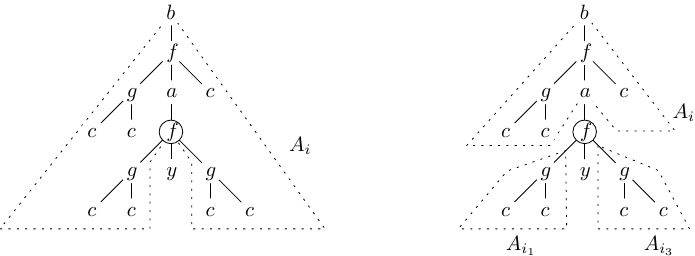}
	\caption{Popping down during uncrossing a crossing parent-leaf pair.  The popped node is encircled, and dotted lines enclose patterns generated by nonterminals.}
	\label{fig:uncrleaf}
\end{figure}

\begin{algorithm}[t] \normalsize 
  \caption{\label{code-algremchild}$\algremchild(\letters_{\geq 1},\letters_0, \grammar)$: uncrossing \parentchild s}
  \begin{algorithmic}[1]
	\For{$i \gets 1 \twodots m-1$} \Comment{popping up letters from $\letters_0$} \label{no more constants}
		\If{$\alpha_i = a \in \letters_0$}
			\State replace each $A_i$ in the right-hand sides by $a$ \label{replaced by a constant}
		\EndIf
	\EndFor
	\For{$i \gets m-1 \twodots 1$} 
		\If{$A_i(a)$ with $a \in \letters_0$ occurs in some rule}  
			\State mark $A_i$ \Comment{we need to pop down a handle from $A_i$}
		\EndIf
		\If{$A_i$ is marked and $\alpha_i$ ends with a nonterminal $A_j$}
				\State mark $A_j$	\Comment{we need to pop down a handle from $A_j$ as well}
			
		\EndIf
	\EndFor
	\For{$i \gets 1 \twodots m-1$} 
		\If{$A_i$ is marked} \Comment{we want to pop down a handle from $A_i$}
			\State let $\alpha_i$ end with handle
			$f(t_1,\ldots,t_k,y,t_{k+1},\ldots,t_\ell)$  \label{popping from marked} \Comment{$\alpha_i$ must end with a handle}
			\State remove this handle from $\alpha_i$ \label{removing handle}
			\For{$j \gets 1 \twodots \ell$}
			  \If{$t_j \notin \letters_0$} \Comment{i.e., it is not a constant}
				\State create a rule $A_{i_j} \to t_j$ for a fresh nonterminal $A_{i_j}$ \label{new rule}
				\State add $A_{i_j}$ to $\widetilde{N}_0$  \label{add-to-N_0-tilde}
				\State $\gamma_j := A_{i_j}$
		         \Else \State $\gamma_j := t_j$
		         \EndIf   
			\EndFor
			\State replace each $A_i(t)$ in a right-hand side by $A_i(f(\gamma_1,\ldots,\gamma_k,t,\gamma_{k+1},\ldots,\gamma_\ell))$ \label{new nonterminals}
		\EndIf
	\EndFor
  \end{algorithmic}
\end{algorithm}

\begin{lemma}
\label{lem: uncrossing children}
Let \grammar{} be a handle grammar and $\grammar' = \algremchild(\letters_{\geq 1},\letters_0,\grammar)$.
Then, the following hold:
\begin{itemize}
\item $\eval(A_m) = \eval(A'_m)$ and hence $\eval(\grammar) = \eval(\grammar')$.
\item If \grammar{} satisfies \GRrefall{}, then so does $\grammar'$.
\item The grammar $\grammar'$ has at most the same number of occurrences of nonterminals of rank $1$ as $\grammar$.
\item The grammar $\grammar'$ has no crossing  \parentchild.
\item During the computation of $\grammar'$ from \grammar{} at 
most one new occurrence of a letter is introduced for each occurrence of a nonterminal of rank $0$ and
at most $r$ new occurrences of letters are introduced for each 
occurrence of a nonterminal of rank $1$. In particular, if \grammar{} satisfies \GRref{3}--\GRref{4}
then at most $2\grammarsizeone r + 2 \grammarsizezero + 2 \grammarsizetilde$ units of credit are issued.
\end{itemize}
\end{lemma}

\begin{proof}
The identity $\eval(A_m) = \eval(A'_m)$ follows as for $\algpop$.
Next, we show that  \GRrefall{} are preserved: so, assume that \grammar{} satisfies \GRrefall.
Replacing nonterminals by constants and popping down handles cannot turn a handle grammar into one that is not a handle grammar,
so \GRref{1} is preserved.
The number of nonterminals in $N_0$ and $N_1$ does not increase, so \GRref{2} also holds.
Concerning \GRref{3}, observe that no new occurrences of nonterminals from $N_1$ are produced and
that new occurrences of nonterminals from $N_0$ can be created only in line~\ref{new rule},
when a rule $A_{i_j} \to t_j$ is added to \grammar{} ($t_j$ may end with a nonterminal from $N_0$).
However, immediately before, in line~\ref{removing handle}, we removed one occurrence of  $t_j$ from \grammar,
so the total count is the same. Hence \GRref{3} holds.

The rules for the new nonterminals
 $A_{i_j} \in \widetilde{N}_0$ that are added in line \ref{add-to-N_0-tilde} are of the form $A_{i_j} \to t_j$, where $f(t_1,\ldots,t_k,y,t_{k+1},\ldots,t_\ell)$ was a handle.
 So, by the definition of a handle, every $t_j$ is either of the form $wc$ or $wA_k$,
where $w$ is a string of unary letters, $c$ a constant, and $A_k \in N_0 \cup \widetilde{N}_0$.
Hence, the rule for  $A_{i_j}$ is of the form required in \GRref{5} and thus \GRref{5} is preserved.

It remains to show \GRref{4}, i.e., the bound on the number of occurrences of nonterminals from $\widetilde{N}_0$,
which is the only non-trivial task.
When we remove the handle $f(t_1,\ldots,t_k,y,t_{k+1},\ldots,t_\ell)$ from the rule for $A_i$ and
introduce new nonterminals $A_{i_1}$, \ldots, $A_{i_\ell}$ then we say that $A_i$ \emph{owns}
those new nonterminals (note that $A_i \in N_1$).\footnote{Some $t_j$ might be constants and
are not replaced by new nonterminals $A_{i_j}$. For notational simplicity we assume here that no $t_j$ is a constant, which in some sense is the worst case.}
Furthermore, when we replace an occurrence of $A_i(t)$ in a right-hand side by $A_i(f(A_{i_1},\ldots,A_{i_k},t,A_{i_{k+1}},\ldots,A_{i_\ell}))$ in line~\ref{new nonterminals},
those new occurrences of $A_{i_1}$, \ldots, $A_{i_\ell}$ are \emph{owned} by this particular occurrence of $A_i$.
If the owning nonterminal (or its occurrence) is later removed from the grammar,
the owned (occurrences of) nonterminals  get \emph{disowned} and they remain so till they get removed.

The crucial technical claim is that one occurrence of a nonterminal owns at most $r-1$ occurrences of nonterminals,
here stated in a slightly stronger form:
\begin{clm}
	\label{clm: at most r-1}
When an occurrence of $A_i \in N_1$ creates new occurrences of nonterminals (in line~\ref{new nonterminals}) from $\widetilde{N_0}$, then right before
it does not own any occurrences of other nonterminals from $\widetilde{N_0}$.
\end{clm}
This is shown in a series of simpler claims.

\begin{clm}
	\label{clm: the same owned appearances}
For a fixed nonterminal $A_i \in N_1$, every occurrence of $A_i$ owns occurrences of the same nonterminals $A_{i_1}, \ldots, A_{i_\ell}$.
\end{clm}
This is obvious: We assign occurrences of the same nonterminals $A_{i_1}, \ldots, A_{i_\ell}$ to each occurrence of $A_i$
in line~\ref{new nonterminals} and the only way that such an occurrence ceases to exist is when $A_{i_j}$ is replaced with a constant.
But this happens for all occurrences of $A_{i_j}$ at the same time.

\medskip

\begin{figure}
	\centering
		\includegraphics{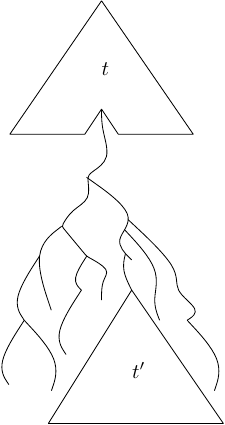}
	\caption{The context $t$ dominates the tree $t'$.}
	\label{fig:domination}
\end{figure}

In order to formulate the next claim, we need some notation:  we say that an occurrence of a subcontext $t(y)$ of \mytree{} \emph{dominates} an occurrence of the
subtree $t'$ of \mytree,
if \mytree{} can be written as $\mytree = C_1(t(C_2(t')))$, where $t$ and $t'$ refer here to the specific occurrences of $t$ and $t'$, respectively,
see Figure~\ref{fig:domination} for an illustration.
\begin{clm}
	\label{clm: domination prevails}
When $A_i$ owns $A_{i_j}$ then each subcontext generated by $A_i$ in \mytree{}
dominates a subtree generated by $A_{i_j}$.
\end{clm}
This is true right after the introduction of an owned nonterminal $A_{i_j}$:
Each occurrence of $A_i$ is replaced by $A_i(f(A_{i_1},\ldots,A_{i_k},t,A_{i_{k+1}},\ldots,A_{i_\ell}))$
and this occurrence of $A_i$ owns the occurrence of $A_{i_j}$ in this particular
$f(A_{i_1},\ldots,A_{i_k},t,A_{i_{k+1}},\ldots,A_{i_\ell})$.
What can change? Compression of letters does not affect  dominance, as we always compress subtrees that are either completely within
$\eval(A_i)$ or completely outside $\eval(A_i)$ and the same applies to each $A_{i_j}$.
When popping up from  $A_{i_j}$ then the new tree generated by this occurrence of $A_{i_j}$ is a subtree of the previous one,
so the dominance is not affected.
When popping up or popping down from  $A_i$, then the new context is a subcontext of the previous one,
so  dominance is also not affected (assuming that $A_i$ exists afterwards).
Hence the claim holds.

\begin{clm}
	\label{clm: dominates constant}
When $A_i$ is marked by \algremchild, then $T$ contains a subtree of the form $t(a)$ for a constant symbol $a$,
where the subcontext $t$ is generated by an occurrence of $A_i$.
\end{clm}
If $A_i$ was marked because $A_i(a)$ occurs in some rule then this is obvious,
otherwise it was marked because it is the last nonterminal in the right-hand side of some $A_j$ which is also marked (and $j > i$).
By induction we conclude that \mytree{} contains a subtree of the form $t(a)$ for a constant symbol $a$,
where the subcontext $t$ is generated by an occurrence of $A_j$.
But as $A_i$ is the last symbol in the right-hand side for $A_j$, the same is true for $A_i$.

\medskip
Getting back to the proof of Claim~\ref{clm: at most r-1}, suppose that 
we create new occurrences of the nonterminals $A_{i_1}, \ldots, A_{i_\ell}$ in 
line~\ref{new nonterminals} 
and right before line~\ref{new nonterminals}, $A_i$ already owns a nonterminal $A_q$.
Then $A_i$ must be marked and so by Claim~\ref{clm: dominates constant} we know that
\mytree{} contains a subtree of the form $t(a)$ for a constant $a$,
where  the subcontext $t$ is generated by an occurrence of $A_i$.
By Claim~\ref{clm: the same owned appearances} this occurrence of $A_i$
owns an occurrence of $A_q$.
Then by Claim~\ref{clm: domination prevails} $A_q$ must produce the constant $a$.
But this is not possible, since in line~\ref{no more constants} we eliminate all nonterminals
that generate constants, and 
there is no way to introduce a  nonterminal that produces a constant.
So, we derived a contradiction and thus Claim~\ref{clm: at most r-1} holds.
\medskip

Using Claim~\ref{clm: at most r-1} we can show the bound from \GRref{4} on the number of occurrences of nonterminals from $\widetilde{N_0}$.
The bound clearly holds as long as there are no disowned nonterminals:
Each occurrence of a nonterminal from $\widetilde{N_0}$ is owned by an occurrence of a nonterminal from $N_1$.
By Claim~\ref{clm: at most r-1} at most $r-1$ of them are owned by an occurrence of a nonterminal from $N_1$,
so there are at most $\grammarsizeone (r-1)$ such occurrence, where $\grammarsizeone \leq n_1$ by \GRref{3}.
As the second subclaim we show that there are at most $(n_1 - \grammarsizeone)(r-1)$ disowned occurrences of nonterminals from $\widetilde{N_0}$, which 
finally shows \GRref{4}.
This is shown by induction and clearly holds when there are no disowned nonterminals.
By \GRref{3} the number $\grammarsizeone$ of occurrences of nonterminals from $N_1$ never increases.
So it is left to consider what happens,
when a nonterminal gets disowned. Assume that it was owned by $A_i \in N_1$ and now this $A_i$ is removed from \grammar.
Thus $\grammarsizeone$ decreases by $1$ and we know, from Claim~\ref{clm: at most r-1} that $A_i$ owns at most $r-1$ nonterminals,
which yields the claim.

Concerning the occurrences of new letters and their credit: We introduce at most one letter 
for each occurrence of a nonterminal of rank $0$ (i.e., from $N_0 \cup \widetilde{N_0}$) during popping up
and at most $r$ letters for each occurrence of a nonterminal of rank $1$ (i.e., from $N_1$) during popping down.
As \GRref{3}--\GRref{4} give a bound on number of occurrences of nonterminals,
assuming them yields that at most 
$2r \grammarsizeone + 2 \grammarsizezero + 2 \grammarsizetilde$ units of credit are issued.

Finally, we show that $\grammar' = \algremchild(\letters_{\geq 1},\letters_0,\grammar)$ does not contain crossing \parentchild s.
Observe that after the loop in line~\ref{no more constants} there are no nonterminals $A_i$ such that $\eval(A_i) \in \letters_0$.
Afterwards, we cannot create a nonterminal that evaluates to a constant in $F_0$.
Hence there can be no crossing \parentchild{} that satisfies \FCref{1} or \FCref{3}.

In order to rule out \FCref{2},
we proceed with a series of claims. We first claim that if $A_i$ is marked then in line~\ref{popping from marked}
indeed the last symbol in the rule $A_i \to \alpha_i$ is a handle (so it can be removed in line \ref{removing handle}).
Suppose this is wrong and let $A_i$ be the nonterminal with the smallest $i$ for which this does not hold.
As a first technical step observe that if some $A_j$ is marked then $A_j \in N_1$:
Indeed, if $A_j(a)$ occurs in a rule of \grammar{} then clearly $A_j \in N_1$
and if $A_k$ is the last nonterminal in the rule for $A_j \in N_1$ then $A_k \in N_1$ as well.
Hence $A_i \in N_1$. So the last symbol in the rule for $A_i$ is either a nonterminal $A_j \in N_1$ with $j < i$ or a handle.
In the latter case we are done as there is no way to remove this handle from the rule for $A_i$
before $A_i$ is considered in line~\ref{popping from marked}.
In the former case observe that $A_j$ is also marked. By the minimality of $i$,
when $A_j$ is considered in line~\ref{popping from marked}, it ends with a handle 
$f(t_1,\ldots,t_k,y,t_{k+1},\ldots,t_\ell)$.
Hence the terminating $A_j(y)$ in the right-hand side for $A_i$ is replaced by
$A_j(f(\gamma_1,\ldots,\gamma_k,y,\gamma_{k+1},\ldots,\gamma_\ell))$ 
and there is no way to remove the handle
$f(\gamma_1,\ldots,\gamma_k,y,\gamma_{k+1},\ldots,\gamma_\ell)$
from the end until $A_i$ is considered in line \ref{removing handle}.

Finally, suppose that there is a crossing \parentchild{}  because of the situation \FCref{2}
after $\algremchild$, i.e., $A_i(a)$ occurs in some right-hand side and the last letter of $\eval(A_i)$ is $f$.
Then in particular we did not pop down a letter from $A_i$, so by the earlier claim $A_i$ was not marked.
But $A_i(a)$ occurs in the right-hand already after the loop in line~\ref{no more constants}, because
$a$ cannot be introduced after the loop.
So we should have marked $A_i$, which is a contradiction.
\qedhere
\end{proof}

So in case of leaf compression we can proceed as in the case of $(\lettersup,\lettersdown)$-compression
and chain compression: we first uncross all \parentchild{}s and then compress each right-hand side independently.

\begin{lemma}
\label{lem: representation cost child}
Let \grammar{} be a handle grammar and
$\grammar' = \algtreechildcomp(\letters_{\geq 1},\letters_0,\algremchild(\letters_{\geq 1},\letters_0,\grammar))$.
Then, the following hold:
\begin{itemize}
\item $\eval(\grammar') =  \algtreechildcomp(\letters_{\geq_1},\letters_0,\eval(\grammar))$
\item The grammar $\grammar'$ has at most the same number of occurrences of nonterminals of rank $1$ as $\grammar$.
\item  If \grammar{} satisfy \GRrefall{}, then so does $\grammar'$.
\item  During the computation of $\grammar'$ from \grammar{} at 
most one new occurrence of a letter is introduced for each occurrence of a nonterminal of rank $0$ and
at most $r$ new occurrences of letters are introduced for each 
occurrence of a nonterminal of rank $1$.
In particular, if \grammar{} satisfies \GRref{3}--\GRref{4}
then at most $2\grammarsizeone r + 2 \grammarsizezero + 2 \grammarsizetilde$ units of credit are issued.
\item The issued credit and the credit released by $\algtreechildcomp$ cover the representation cost of fresh letters as well as 
their credit in $\grammar'$.
\end{itemize}
\end{lemma}

\begin{proof}
This is a combination of Lemma~\ref{lem: no crossing parent node child compression} and~\ref{lem: uncrossing children}:
By Lemma~\ref{lem: uncrossing children},
$\algremchild$ eliminates all  crossing \parentchild{}s 
and introduces at most one letter for each occurrence of nonterminal of rank $0$
and at most $r$ letters for each occurrence of a nonterminal of rank $1$.
The \GRref{3}--\GRref{4} give a bound on number of occurrences of nonterminals,
which leads to bound $2\grammarsizeone r + 2 \grammarsizezero + 2 \grammarsizetilde$ of credit.
Then by Lemma~\ref{lem: no crossing parent node child compression}, $\algtreechildcomp$
ensures that $\eval(\grammar') = \algtreechildcomp(\letters_{\geq_1},\letters_0,\eval(\grammar))$.
Furthermore the credit of the new letters and the representation cost is covered by the credit released by $\algtreechildcomp$.
Finally, both subprocedures preserve \GRrefall{}
and do not introduce occurrences of rank $1$ nonterminals.
\end{proof}

By Lemma~\ref{lem:number of phases} we apply at most $\Ocomp(\log n)$ many leaf compressions to $\grammar$.
By \GRref{3}--\GRref{4} we have $\grammarsizezero \leq n_0$,
$\grammarsizetilde \leq n_1(r-1)$ and $\grammarsizeone \leq n_1$. Hence, we get:

\begin{corollary}
\label{cor: credit child compression}
Leaf compression issues in total at most  $\Ocomp((n_0 + n_1r) \log n)$ units of credit during all
 modifications of \grammar.
\end{corollary}

From Corollaries~\ref{cor: credit pair compression}, \ref{cor: credit chain compression}, and \ref{cor: credit child compression}
and the observation that the initial credit is $\Ocomp(|\grammar|) \leq \Ocomp(\grammarsize r)$ we get:

\begin{corollary}
\label{cor: credit all compressions}
The whole credit issued during all modifications of \grammar{} is in $\Ocomp(\grammarsize r + (n_0 + n_1r) \log n)$.
\end{corollary}

\subsection{Calculating the cost of representing letters in chain compression}
\label{subsec: chain cost}

The issued credit is enough to pay the two units of credit  for every letter introduced during  popping,
whereas the released credit covers the cost of representing the new letters introduced by $(\lettersup,\lettersdown)$-compression
and leaf compression.
However, the released credit \emph{does not} cover the cost of representation for letters created during  chain compression.
The appropriate analysis is presented in this section. The overall plan is as follows:
Firstly, we define a scheme of representing letters introduced by chain compression based on the grammar \grammar{}
and the way \grammar{} is changed by \algtreechaincomp{} (the \emph{\grammar{}-based representation}).
Then, we show that for this scheme the representation cost is bounded by $\Ocomp((n_0 + n_1r) \log n)$.
Lastly, it is proved that the actual representation cost for the letters introduced by chain compression during the 
run of  \algmain{}  (as defined in Lemma~\ref{lem: cost of powers} in Section~\ref{sec chain comp}, called the \emph{\algmain-based representation} later on) 
is smaller than the \grammar-based one.
Hence, it is bounded by $\Ocomp((n_0 + n_1r) \log n)$ as well.

\subsubsection{\grammar{}-based representation}
We now define the \grammar-based representation,
which is different from the representation
actually used by \algmain.
As noted, \grammar{} is a tree grammar obtained by modifying the optimal grammar for the input tree.
At each stage, it produces the tree currently stored by \algmain.
The intuition is as follows:
While \grammar{} can produce patterns of the form $a^\ell$, which have exponential length in $|\grammar|$,
most patterns of this form are obtained by concatenating explicit $a$-chains to a shorter pattern.
In such a case, the credit that is released from the explicit occurrences of $a$
can be used to pay for the representation cost.
This does not apply when the new pattern is obtained by concatenating two patterns 
(popped from nonterminals) inside a rule.
In such a case we represent the pattern using the binary expansion at the cost of $\Ocomp(\log \ell)$.
However, this cannot happen too often: When patterns of length $p_1, p_2, \ldots, p_\ell$
are compressed and the obtained letters are represented
(at the cost of $\Ocomp(\log \prod _{i=1}^\ell  p_i )$),
then it can be shown that the size of the derived context in the input tree is at least $\prod_{i=1}^\ell p_i$, which is at most $n$.
Thus $\sum_{i=1}^\ell \log p_i = \Ocomp(\log \prod_{i=1}^\ell p_i) = \Ocomp(\log n)$; this is formally shown later on.

Let us fix a unary letter $a$ whose chains are compressed and represented.
The \grammar-based representation
creates a new letter for each chain pattern from $a^+$  that is either popped from a right-hand side during \algremblocks{}
or is in a rule at the end of \algremblocks{} (i.e., after popping but before the actual replacement in \algtreechaincomp).
Some of those chain patterns are designated as powers:
fix a rule that is considered by \algremblocks.
If the $a$-suffix popped from the first nonterminal
and the $a$-prefix popped from the second nonterminal
are part of one $a$-pattern (obtained after those poppings),
then this $a$-pattern is a \emph{power}.
Note that this power may either stay in this rule or be popped (if one of the nonterminals is removed from the rule).
For each chain pattern $a^{\ell}$ that is not a power we can identify another represented pattern $a^k$
(where we allow $k = 0$ here) such that $a^\ell$ is obtained by concatenating explicit occurrences of $a$ from some right-hand side to $a^k$.

Note that for a fixed length $\ell$
there may be many different occurrences of the pattern $a^\ell$ that are represented.
In particular, some of them may be powers and some not.
We arbitrarily choose one of those occurrences
and the way it is created and represent $a^\ell$ (once) according to this choice.

\begin{example}
	\label{ex:powers}
	Consider the following grammar, in which only letters of arity $1$ are used and they are written in a string notation:
	$A_0 \to a$, $A_1 \to bc aA_0$, $A_2 \to A_0abc$,
	$A_3 \to A_0aA_0$, $A_4 \to A_1 aA_0$, $A_5 \to A_0aA_2$, $A_6 \to A_1aA_2$, $A_7 \to baaA_3$.
	Let us consider $a$-chains. 
	The $a$-chain $a$ in the right-hand side of $A_0$ is not a power;
	it is obtained by concatenating an explicit occurrence of $a$ to $\epsilon$.
	After replacing $A_0$ by $a$ in all right-hand sides, we obtain the rules
	$A_1 \to bc a \underline{a}$, $A_2 \to \underline{a}abc$,
	$A_3 \to \underline{a}a\underline{a}$, $A_4 \to A_1 a\underline{a}$, $A_5 \to \underline{a}aA_2$, $A_6 \to A_1aA_2$, $A_7 \to baaA_3$
	(new occurrences of $a$ in right-hand sides are underlined). The occurrences of the $a$-chain $a^2$ 
	in the right-hand sides for $A_1$ and $A_2$ are not powers, since they are obtained by concatenating an 
	explicit $a$ to the $a$ popped from    $A_0$. On the other hand, the occurrence of the $a$-chain $a^3$ in the rule for 
	$A_3$ is a power, since it is obtained by concatenating a popped $a$, and explicit $a$, and a popped $a$. 
	After popping $a$'s from $A_1$, $A_2$, and $A_3$ we obtain the following rules for $A_4$, $A_5$, $A_6$, $A_7$:
	$A_4 \to A_1 \underline{aa} a\underline{a}$, $A_5 \to \underline{a}a \underline{aa}  A_2$, $A_6 \to A_1 \underline{aa}  a \underline{aa} A_2$,   
	$A_7 \to baa \underline{aaa}$. The occurrences of the $a$-chain $a^4$ (resp., $a^5$) in the rules for $A_4$ and $A_5$ (resp., $A_6$)
	are powers, whereas the occurrence of $a^5$ in the rule for $A_7$ is not a power. Note that the $a$-chain $a^5$ can be either represented
	as a power or as a non-power.
	\hfill $\diamond$
\end{example}

We represent chain patterns as follows:
\begin{enumerate}[(a)]
	\item For a chain pattern $a^\ell$ that is a power we represent $a_\ell$ using the binary expansion,
	which costs $\Ocomp(1 + \log \ell)$.
	\item A chain pattern $a^\ell$ that is not a power is obtained by concatenating $\ell - k \geq 1$ explicit occurrences of $a$ from a right-hand side to $a^k$ (recall that
	we fixed some choice in this case), in particular $a_k$ is represented.
	In this case we 
	represent $a_\ell$ as $a_k a^{\ell-k}$. The representation cost is $\ell - k + 1$, which
	is covered by the $2(\ell-k) \geq \ell - k + 1$ units of credit released from the $\ell-k \geq 1$ many explicit occurrences of $a$.
	Recall that the credit for occurrences of a fresh letter $a_\ell$ is covered by the issued credit,
	see Lemma~\ref{lem:blocksc}. Hence the released credit is still available.
\end{enumerate}
We refer to the cost in (a) as the \emph{cost of representing a power}.
As remarked above, the cost in (b) is covered by the released credit.
The cost in (a) is redirected towards the rule in which this power was created.
Note that this needs to be a rule for a nonterminal from $N_0 \cup N_1$,
as the right-hand side of the rule needs to have two nonterminals to generate a power
and by \GRref{5} the right-hand sides for nonterminals from $\widetilde{N_0}$ have at most one nonterminal.
In Section~\ref{subsubsec: grammar based} we show that the total cost redirected towards a rule 
during all modifications of \grammar{} is at most $\Ocomp(\log n)$.
Hence, the total cost in (b) is $\Ocomp((n_0 + n_1)\log n)$.

\begin{example}
In Example~\ref{ex:powers} we can represent $a^5$ as a power, which yields the rules $a_5 \to a_4a$, $a_4 \to a_2 a_2$ and $a_2 \to aa$
corresponding to the binary notation of $5$, or as a non-power. For the latter choice we get the rule $a_4 \to a_3aa$, where $a_3$ is represented elsewhere. \hfill $\diamond$
\end{example}

\subsubsection{Cost of \grammar{}-based representation}
\label{subsubsec: grammar based}
We now estimate the cost of representing the letters introduced during chain
compression described in the previous section.
The idea is that if we redirect towards $A_i$ the cost of representing powers of length
$p_1$, $p_2$, \ldots, $p_\ell$ (which have total representation cost $\Ocomp(\sum_{i=1}^\ell (1 + \log p_i)) = \Ocomp(\log (\prod_{i=1}^\ell p_i))$)
during all chain compression steps,
then in the initial grammar, $A_i$ generates a subpattern of the input tree of  size at least $p_1 \cdot p_2 \cdot \cdots \cdot p_\ell \leq n$
and so the total cost of representing powers is at most $\log n$ per nonterminal from $N_0 \cup N_1$.
This is formalised in the lemma below.

\begin{lemma}
\label{clm:cost from power to rule}
The total cost of representing powers  charged towards a single rule for a nonterminal from $N_0 \cup N_1$ is $\Ocomp(\log n)$.
\end{lemma}

\begin{proof}
We first bound the cost  redirected towards a rule for $A_i \in N_1$.
There are two cases:
First, after the creation of a power in the rule $A_i \to u A_j v A_k w$ one of the nonterminals
$A_j$ or $A_k$ is removed from the grammar. But this happens at most once for the rule
(there is no way to reintroduce a nonterminal from $N_1$ to a rule) and the cost of $\Ocomp(\log n)$
of representing the power can be charged to the rule. Note that here the assumption that we consider $A_i \in N_1$
is important: otherwise it could be that the second nonterminal in a right-hand side
is removed and added several times, see the last sentence in the proof of Lemma~\ref{lem: no crossing parent node child compression}.

The second and crucial case is when after the creation of a power
both nonterminals remain in the rule. Fix such a rule $A_i \to u A_j v A_k w$, where $u$, $v$, and $w$ are sequences of handles.
Since we create a power, there is a unary letter $a$ such that $v \in a^*$ and 
$\eval(A_j)$ (respectively,  $\eval(A_k)$) has a suffix (respectively, prefix) from $a^+$,
see Figure~\ref{fig:powers}.

\begin{figure}
	\centering
		\includegraphics{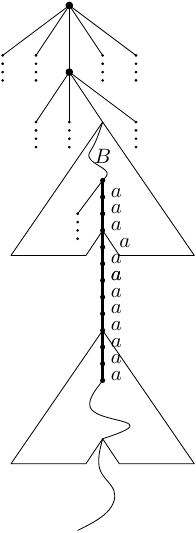}
	\caption{During the creation of a power of $a$ the first nonterminal has an $a$-suffix, while the second
	 one has an $a$-prefix, and the letters between them are all $a$'s.}
	\label{fig:powers}
\end{figure}

Fix this rule and consider all such creations of powers performed in this rule during all modifications of \grammar.
Let the consecutive letters, whose chain patterns are compressed, be
$a^{(1)}$, $a^{(2)}$, \ldots, $a^{(\ell)}$ and their lengths $p_1$, $p_2$, \ldots, $p_\ell$.
Let also $b^{(s+1)}$ be the letter that replaces the chain $(a^{(s)})^{p_s}$;
note that $b^{(s+1)}$ does not need to be $a^{(s+1)}$, as there might have been some other compressions performed on the letter $b^{(s+1)}$.
Then the cost of the representation charged towards this rule is bounded by
\begin{equation}
	\label{eq:representation cost chains}
	\Ocomp(\sum_{s=1}^\ell (1 + \log p_s)) = \Ocomp(\sum_{s=1}^\ell \log p_s) \enspace ,
\end{equation}
as $p_s \geq 2$ for each $1 \leq s \leq \ell$.

Define the \emph{weight} $\weight(a)$ of a letter $a$ as follows: in the input tree each letter has weight $1$.
When we replace $ab$ by $c$, set $\weight(c) = \weight(a) + \weight(b)$.
Similarly, when $a_\ell$ represents $a^\ell$ then set $\weight(a_\ell) = \ell \cdot \weight(a)$,
and when $f'$ represents $f$ with constant-labelled children $a_1,\ldots,a_\ell$, then set $\weight(f') = \weight(f) + \sum_{i = 1}^\ell \weight(a_i)$.
The weight of a tree is defined as the sum of the weights of all node labels.
It is easy to see that in this way $\weight(\eval(A_m)) = n$ is preserved during all modifications of the grammar \grammar.

For a rule $A_i \to u A_j v A_k w$ we say that the letters in handles from $v$ are \emph{between} $A_j$ and $A_k$.
Observe that as long as both $A_j$ and $A_k$ are in the rule, the maximal weight of letters between $A_j$ and $A_k$ cannot decrease:
popping letters and handles from $A_j$ and $A_k$ cannot decrease this maximal weight, and the same is true for a compression step.
Moreover, there is no way to remove a letter that is between $A_j$ and $A_k$ or to change it into a nonterminal.

Now, directly after the $s$-th chain compression the only letter between $A_j$ and $A_k$
is $b^{(s+1)}$ which has weight $p_s \cdot \weight(a^{(s)})$ since it replaces
$(a^{(s)})^{p_s}$.
On the other hand,
right before the $(s+1)$-th chain compression
the sequence between $A_j$ and $A_k$ is $(a^{(s+1)})^{p_{s+1}}$.
Since the maximal weight of a letter between $A_j$ and $A_k$ cannot decrease,
we have
$$
\weight(a^{(s+1)}) \geq \weight(b^{(s+1)}) = p_s \cdot  \weight(a^{(s)}) \enspace .
$$
Since $\weight(a^{(1)}) \geq 1$ it follows that $\weight(b^{(\ell+1)}) \geq \prod_{s=1}^\ell p_s$.
As $\weight(b^{(\ell+1)}) \leq n$ we have
\begin{align*}
n &\geq \prod_{s=1}^\ell p_s \enspace ,
\intertext{and so it can be concluded that}
\log n &\geq \log \left( \prod_{s=1}^\ell p_s \right)\\
		&= \sum_{s=1}^\ell \log p_s \enspace .
\end{align*}
Therefore, the total cost $\Ocomp(\sum_{s=1}^\ell \log p_s)$, as estimated in~\eqref{eq:representation cost chains},
is $\Ocomp(\log n)$.

It is left to describe the differences, when considering  nonterminals from $N_0$.
There are two of them:
\begin{itemize}
	\item 
	 When a power is created in a rule for a nonterminal $A_i \in N_0$, then the rule must contain two nonterminals, i.e.,
	 it must be of the form $A_i \to  u A_j  a^k A_k$ for a unary symbol $a$, and afterwards it is of the same form. In particular we do not have to consider the case when the second 
	 nonterminal $A_k$ is removed from the rule (as $A_k$ is of rank $0$, it cannot be replaced with a chain).
	\item Instead of considering the letters between $A_j$ and $A_k$, we consider letters that are \emph{below} $A_j$:
	In a rule $A_i \to u A_j v A_k$\ or  $A_i \to u A_j vc$, these are the letters that are in handles in $v$ as well as the ending $c$.
\end{itemize}
As before, as long as $A_j$ is in the rule, the maximal weight of letters that are below $A_j$ can only increase
(note that the rule for $A_i$ can switch between the forms $A_i \to u A_j v A_k$ and $A_i \to u A_j vc$ many times,
but this does not affect the claim).

Considering the cost of creating powers:
The representation of the power that is created in the phase when $A_j$ is removed costs at most
$\Ocomp(\log n)$ and there is no way to bring a nonterminal from $N_1$ back to this rule. Hence,  this cost is paid once.
So, it is enough to consider the cost of powers that were created when $A_j$
was still present in the rule.
Let as in the previous case the consecutive letters, whose chain patterns are compressed, be
$a^{(1)}, a^{(2)}, \ldots, a^{(\ell)}$ and let their lengths be $p_1, p_2, \ldots, p_\ell$.
Let also $b^{(s+1)}$ be the letter that replaces the chain $(a^{(s)})^{p_s}$.
It is enough to show that $\weight(a^{(s+1)}) \geq p_s \cdot \weight(a_s)$ as then the rest of the proof
follows as in the case of a nonterminal from $N_1$.

After the $s$-th compression, the right-hand side of $A_i$ has the form $u A_j b^{(s+1)} A_k$.
Before the $(s+1)$-th compression, the right-hand side of $A_i$ has the form
$u' A_j (a^{(s+1)})^{p_{s+1}} A_{k'}$.
By the earlier observation, the maximal weight of letters below by $A_j$ can only increase, hence
$\weight(a^{(s+1)}) \geq \weight(b^{(s+1)}) =
\weight((a^{(s)})^{p_{s}})=
p_s \cdot \weight(a_s)$, as claimed.
\qedhere
\end{proof}

Now, the whole cost of the \grammar-based representation can be calculated:
\begin{corollary} \label{coro:grammar-based cost}
The cost of the \grammar-based representation is $\Ocomp(\grammarsize r + (n_0 + n_1r) \log n)$.
\end{corollary}

\begin{proof}
Concerning powers, we assign to each nonterminal from $N_0 \cup N_1$ a cost of $\Ocomp(\log n)$
by Lemma~\ref{clm:cost from power to rule}.
There are at most $n_0 + n_1$ such nonterminals, as we do not introduce new ones. So, the total representation cost 
for powers is  $\Ocomp((n_0 + n_1) \log n)$. For non-powers, the representation cost is paid from the released
credit. But the released credit is bounded by the credit assigned to the initial grammar \grammar,
which is at most $\Ocomp(r \grammarsize)$ by Lemma~\ref{lem: monadic grammar},
plus the total issued credit during all modifications of \grammar, which is  $\Ocomp((n_0 + n_1r) \log n)$ by 
Corollary~\ref{cor: credit all compressions}. We get the statement by summing all contributions.
\qedhere
\end{proof}

\subsubsection{Comparing the \grammar{}-based representation cost  and the  \algmain-based representation cost}
\label{sec-comparing}

Recall the \algmain-based representation from Lemma~\ref{lem: cost of powers} in  Section~\ref{sec chain comp}.
We now show that the \algmain-based representation cost is bounded
by the \grammar-based representation cost
(note that both costs include the credit released by explicit letters).
We first bound the costs of both representations by edge-weighted graphs: 
the total cost of a representation is bounded (up to a constant factor) by
the  sum of all edge weights of the corresponding graph. 
Then we show that we can transform the \grammar-based graph into the $\algmain$-based graph without increasing the sum of the 
edge weights. For an edge-weighted graph $\mathcal G$ let $w(\mathcal G)$ be the sum of all edge weights.

Let us start with the \grammar-based representation.
We define the graph $\mathcal{G}_{\grammar}$ as follows:
Each chain pattern that is represented in the  \grammar-based representation
is a node of $\mathcal{G}_{\grammar}$, and edges are defined as follows:
\begin{itemize}
	\item A power $a^\ell$ has an edge with weight $1 + \log \ell$ to $\varepsilon$.
	Recall that the cost of representing this power is $\Ocomp(1 + \log \ell)$.
	\item When $a_\ell$ is represented as $a_k a^{\ell - k}$ ($\ell > k$), then
	node $a^\ell$ has an edge to $a^k$ of weight $\ell - k$.
	The cost of representing $a_\ell$ is $\ell - k + 1 \leq 2(\ell - k)$.
\end{itemize}
From the definition of this graph the following statement  is obvious:

\begin{lemma} \label{lem:representation cost}
 The \grammar-based representation cost is in $\Theta(w(\mathcal{G}_{\grammar}))$.
\end{lemma}

Next let us define the graph $\mathcal{G}_{\algmain}$ of the $\algmain$-based representation:
The nodes of this graph are all chain patterns that are represented in the \algmain-based representation,
and there is an edge of weight $1 + \log(\ell - k)$
from $a^\ell$ to $a^k$
if and only if $\ell > k$ and there is no node $a^q$ with $k < q < \ell$ (note that we may have $k=0$).

\begin{lemma} \label{lem: algmain representation cost}
The  \algmain-based cost of representing the letters introduced during chain compression is in 
$\Ocomp(w(\mathcal{G}_{\algmain}))$.
\end{lemma}

\begin{proof}
Observe that this is a straightforward consequence of the way chain patterns are represented in Section~\ref{sec chain comp}:
Lemma~\ref{lem: cost of powers} guarantees that if $a^{\ell_1}, a^{\ell_2}, \ldots, a^{\ell_k}$ ($\ell_1 < \ell_2 < \cdots < \ell_k$)
are all chain patterns of the form $a^+$  (for a fixed unary letter $a$) that are represented by \algmain{},
then the \algmain-based representation cost for these patterns is $ \Ocomp(\sum_{i=1}^k(1 + \log(\ell_{i} - \ell_{i-1})))$, where $\ell_0=0$.
\qedhere
\end{proof}

We now show that  $\mathcal{G}_{\grammar}$ can be transformed into $\mathcal{G}_{\algmain}$ without increasing the sum of edge weights:

\begin{lemma}
\label{lem: transforming representations}
We have $w(\mathcal{G}_{\grammar}) \geq w(\mathcal{G}_{\algmain}) $.
\end{lemma}

\begin{proof}
We transform the graph $\mathcal{G}_{\grammar}$ into the graph $\mathcal{G}_{\algmain}$ without increasing the sum of edge weights.
Thereby we can fix a letter $a$ and consider only nodes of the form $a^k$ in $\mathcal{G}_{\grammar}$ and $\mathcal{G}_{\algmain}$.
We start with  $\mathcal{G}_{\grammar}$.
Firstly, let us sort the nodes from $a^*$ according to the increasing length.
For each node $a^\ell$ with $\ell > 0$, we redirect its unique outgoing edge to its unique predecessor $a^k$ (i.e., $k < \ell$ and
there is no node $a^q$ with $k < q < \ell$), and assign the weight 
$1 + \log (\ell-k)$ to this new edge. This cannot increase the sum of edge weights:
\begin{itemize}
	\item If $a^\ell$ has an edge of weight $1 + \log \ell$ to $\epsilon$ in $\mathcal{G}_{\grammar}$,
	then $1 + \log \ell \geq 1 + \log (\ell - k)$.
	\item Otherwise it has an edge to some $a^{k'}$ ($k' \leq k$) with  weight $\ell - k'$. Then
	$\ell - k' \geq \ell - k \geq 
	1 + \log (\ell - k)$, as claimed (note that $1 + \log x \leq x$ for $x \geq 1$). 
\end{itemize}
Let $\mathcal G'$ be the graph obtained from $\mathcal{G}_{\grammar}$ by this redirecting.
Note that $\mathcal G'$  is not necessarily  $\mathcal{G}_{\algmain}$,
because $\mathcal{G}_{\grammar}$  may contain nodes that are not present in $\mathcal{G}_{\algmain}$.
In other words: there might exist a chain $a^\ell$ which occurs in the $\grammar$-based representation
but which does not occur in the \algmain-based representation. On the other hand, every node $a^\ell$
that occurs in $\mathcal{G}_{\algmain}$ also occurs in $\mathcal{G}_{\grammar}$: if $a^\ell$ 
is represented by the \algmain-based representation, then it occurs as an $a$-maximal chain in \mytree. 
But right before chain compression, there are no crossing chains in \grammar, see Lemma~\ref{lem: uncrossing chains}.
Hence, $a^\ell$ occurs in some right-hand side of \grammar{} and is therefore represented by the \grammar-based representation
as well.

So, assume that $(a^{\ell_0}, a^{\ell_k})$ is an edge in  $\mathcal{G}_{\algmain}$ but in $\mathcal G'$ we have
edges $(a^{\ell_0}, a^{\ell_1}), (a^{\ell_1}, a^{\ell_2}), \ldots$, $(a^{\ell_{k-1}}, a^{\ell_k})$, where $k > 1$.
But the sum of the weights of these edges in $\mathcal G'$
(which is $\sum_{i=1}^k 1+\log(\ell_{i-1}-\ell_i$))
is larger or equal than the weight of $(a^{\ell_0}, a^{\ell_k})$ in $\mathcal{G}_{\algmain}$ (which is $1+\log(\ell_0-\ell_k)$).
This follows from $1 + \log (x) + 1 + \log(y)  \geq 1 + \log(x+y)$ when $x, y \geq 1$.
\qedhere
\end{proof}
Using (in this order) Lemmata~\ref{lem: algmain representation cost}, \ref{lem: transforming representations}, and \ref{lem:representation cost}, 
followed by Corollary~\ref{coro:grammar-based cost}, we get:

\begin{corollary}
	\label{cor: algmain representation}
The total cost of the \algmain-representation is $\Ocomp(\grammarsize r + (n_0 + n_1r) \log n)$.
\end{corollary}

\subsection{Total cost of representation}
\begin{corollary}
The total representation cost of the letters introduced by \algmain{} (and hence the size of the grammar produced by \algmain)
is $\Ocomp(\grammarsize r + (n_0 + n_1r) \log n) \leq \Ocomp(\grammarsize r + \grammarsize r \log n)$.
\end{corollary}

\begin{proof}
By Corollary~\ref{cor: algmain representation} the representation cost of letters introduced by chain compression is
$\Ocomp(\grammarsize r + (n_0 + n_1 r) \log n)$,
while by Lemmata~\ref{lem:pc crossing} and~\ref{lem: representation cost child} the representation cost of letters introduced by
\paircompression{} and leaf compression is covered by the initial credit (which is $\Ocomp(\grammarsize r)$ by Lemma~\ref{lem: monadic grammar})
plus the total amount of issued credit.
By Corollary~\ref{cor: credit all compressions} the latter is
$\Ocomp(\grammarsize r + (n_0 + n_1r) \log n)$. Recalling that $n_0 = \Ocomp(\grammarsize r)$ and $n_1 = \Ocomp(\grammarsize)$ by Lemma~\ref{lem: monadic grammar}
ends the proof.
\end{proof}

\section{Improved analysis}
\label{sec: improved}
The naive algorithm, which simply represents the input tree $\mytree$ as $A_1 \to \mytree$ results in a grammar of size $n$.
In some extreme cases this might be better than $\Ocomp(\grammarsize r  + \grammarsize r  \log n)$ as guaranteed by \algmain.
In fact, even a stronger fact holds: \emph{any} `reasonable' grammar for a tree $t$ has size at most $2|t|-1$,
where a grammar (for $t$) is reasonable if
\begin{itemize}
	\item it has no production of the form $A \to \alpha$, where $|\alpha|=1$ and
	\item all its nonterminals are used in the derivation of $t$
\end{itemize}
(recall that the size of $\alpha$ does not include the parameters in it).

\begin{lemma}
\label{lem: trivial estimation}
Let \grammar{} contain no production $A \to \alpha$ with $|\alpha| = 1$ and assume that every production is used in the derivation of the tree $t$ defined by \grammar.
Then $|\grammar| \leq 2 |t|-1$.
In particular, if at any point \algmain{} already paid $k$ units of credit for the representation of the letters
and the remaining tree is $\mytree$
then the final grammar for the input tree has size at most $k + 2|\mytree|-1$
\end{lemma}

\begin{proof}
Assume that \grammar{} has the properties from the lemma.
An application of a rule $A_i \to \alpha_i$ to the current tree increases its size by $|\alpha_i| - 1 \geq 1$ for each occurrence of $A_i$ in the tree derived so far.
As we assume that each production is used in the derivation,
each of $|\alpha_i| - 1 \geq 1$ is added at least once and so we get $\sum_{i=1}^m (|\alpha_i| - 1) \leq |t|$.
Thus $\sum_{i=1}^m |\alpha_i| \leq |t| + m$ and so it is left to estimate $m$.
As there are $m$ productions and each application increases the size by at least $1$,
and we start the derivation with a tree of size one, we get  $m \leq |t|-1$.
Thus $\sum_{i=1}^m |\alpha_i| \leq |t| + m \leq 2|t| - 1$.
\end{proof}

We show that when $|\mytree| \approx \propersize$ at a certain point of \algmain, 
then up to this point $\Ocomp(\grammarsize r + \propersize \log(n /\propersize))$ units of credit are issued so far,
where $\grammarsize$ is the size of an optimal SLCF grammar for the input tree.
It follows that the size of the SLCF grammar returned by \algmain{} is
$\Ocomp(\grammarsize r + \grammarsize r \log(n/\grammarsize r))$,
as claimed in Theorem~\ref{thm: main}.

Let $T_i$ be the tree at the beginning of phase $i$ and choose phase $i$ such that $|T_{i}| \geq \propersize > |T_{i+1}|$
(for an input tree with at least $\propersize$ symbols such an $i$ exists, as $|T_1| = n \geq \propersize$
and for the `last' $i$ we have $|T_i| = 1$;
the easy special case in which $n < \propersize$ follows directly from Lemma~\ref{lem: trivial estimation}).
We estimate  the representation cost 
(i.e., the issued credit and the cost of the $\algmain$-based representation)
up to phase $i$ (inclusively).
We show that this cost is bounded by
$\Ocomp(\grammarsize r + \grammarsize r \log(n/\grammarsize r))$,
which shows the full claim of Theorem~\ref{thm: main}.

\begin{lemma}
\label{lem:cost of compression}
If $|\mytree| \geq \propersize$ at the beginning of a phase,
then till the end of this phase, the representation cost of the fresh letters introduced by \algmain{}
as well as the credit of the letters in the current SLCF grammar \grammar{} is $\Ocomp(\grammarsize r+ \propersize \log(n/\propersize))$.
\end{lemma}
\begin{proof}

We estimate separately the amount of issued credit and the representation cost for letters replacing chains.
This covers the whole representation cost for fresh letters
(see Lemmata~\ref{lem:pc crossing},~\ref{lem:blocksc} and~\ref{lem: representation cost child})
as well as the credit on the letters in the current SLCF grammar.

\subsubsection*{Credit}
Observe first that the initial grammar \grammar{} has at most $\grammarsize r$ credit, see Lemma~\ref{lem: monadic grammar}.
The input tree has size $n$ and the one at the beginning of the phase is of size $s = |\mytree|$.
Hence, there were $\Ocomp(\log(n/s))$ phases before, as in each phase the size of $\mytree$
drops by a constant factor, see Lemma~\ref{lem:number of phases}.
Adding one phase for the current phase still yields $\Ocomp(\log(n/s))$ phases. 
As $s \geq \propersize$, we obtain the upper bound $\Ocomp(\log(n/\propersize))$ on the number of phases.
Due to Lemmata~\ref{lem:pc crossing},~\ref{lem:blocksc} and~\ref{lem: representation cost child},
during \paircompression, chain compression and leaf compression
at most $\Ocomp(n_0 + n_1 r)$  units of credit per phase are issued, and by Lemma~\ref{lem: monadic grammar} this is at most $\Ocomp(\propersize)$.
So in total $\Ocomp(\grammarsize r + \propersize \log(n/\propersize))$ units of credit are issued.
From Lemmata~\ref{lem:pc crossing},~\ref{lem:blocksc} and~\ref{lem: representation cost child}
we conclude that this credit is enough to cover
the credit of all letters in \grammar's right-hand sides as well as the representation cost of letters introduced during \paircompression{} and leaf compression.
So it is left to calculate the cost of representing chains.

\subsubsection*{Representing chains}
Observe that the analysis in Section~\ref{subsec: chain cost} did not assume anywhere that \algmain{}
was carried out completely, i.e., the final grammar was returned.
So we can consider the cost of the \grammar-based representation, the \algmain-based representation, and the corresponding graphs.
Lemma~\ref{lem:representation cost} still applies and the cost of the $\grammar$-based representation is $\Theta(w(\mathcal{G}_{\grammar}))$.
By Lemma~\ref{lem: algmain representation cost} the cost of the \algmain-based representation is $\Ocomp(w(\mathcal{G}_\algmain))$.
Lemma~\ref{lem: transforming representations} shows that we can transform $\mathcal{G}_{\grammar}$ to $\mathcal{G}_\algmain$
without increasing the sum of weights.
Hence it is enough to show that the \grammar-based representation cost is at most $\Ocomp(\grammarsize r + \propersize \log(n/\propersize))$.

The \grammar-based representation cost consists of some released credit and the cost of representing powers, see its definition.
The former was already addressed
(the whole issued credit is $\Ocomp(\grammarsize r + \propersize \log(n/\propersize))$)
and so it is enough to estimate the latter, i.e., the cost of representing powers.

The outline of the analysis is as follows:
When a new power $a^\ell$ is represented,
we mark some letters of the input tree (and perhaps modify some other markings).
Those markings are associated with nonterminals. 
Formally, for a nonterminal  $A_i \in N_0 \cup N_1$ we introduce the notions
of an  $A_i$-pre-power marking and $A_i$-in marking. Such a marking is  a subset of 
the node set of the initial tree (note that we do not define such markings for a nonterminal $A_i \in \widetilde{N_0}$). 
These marking satisfy the following conditions:
\begin{enumerate}[(M1)]
	\item \label{M1} Each marking contains at least two nodes and two different markings are disjoint.
         \item \label{M2} For every nonterminal $A_i$ and every $X \in \{ \text{pre-power}, \text{in}\}$ there is at most on
           $A_i$-X marking.
        \item \label{M3} If $p_1, p_2, \ldots, p_k \geq 2$ are the sizes of the markings (i.e., the cardinalities of the node sets), then the cost of representing powers (created up to the current phase) 
	by the \grammar-based representation is  $c \sum_{i = 1}^{k} \log p_i$ (for some fixed constant c).
\end{enumerate}
Note that in (M\ref{M3}) we must have $k \leq d r \grammarsize$ for some constant $d$, because $k \leq 2(|N_0| + |N_1|) \leq \Ocomp(\grammarsize r) $ by
Lemma~\ref{lem: monadic grammar}.

Using \Mrefall{} the total cost of representing powers (in the \grammar-based representation) can be upper-bounded by
(a constant times)
\begin{subequations}
\label{eq:estimations}
\begin{equation}
\label{eq:estimations 1}
\sum_{i = 1}^{k} \log p_i \quad \text{under the constraints } k \leq d r \grammarsize \text{ and } \sum_{i = 1}^{k} p_i \leq n \enspace ,
\end{equation}
where $d$ is some constant.
Let us bound the sum $\sum_{i = 1}^{k} \log p_i$ under the above constraints:
Clearly, the sum is maximised for $\sum_{i = 1}^{k} p_i = n$.
For a fixed $k$ and $\sum_{i = 1}^{k} p_i = n$
we have $\sum_{i = 1}^{k} \log p_i = \log ( \prod_{i = 1}^{k} p_i)$.
By the  inequality of arithmetic and geometric means we conclude that $\log ( \prod_{i = 1}^{k} p_i) \leq \log ((\sum_{i=1}^k p_i / k)^k) = k \log(n/k)$,
where the maximum $k \log(n/k)$ is achieved if  each $p_i$ is equal to $n / k$.
Now, the term $k \log(n/k)$ is maximised for $k = n/e$ (independently of the base of the logarithm).
Moreover, in the range $[0,n/e)$ the function $f(k) = k \log(n/k)$ is monotonically increasing. 
Hence, if $d r \grammarsize \leq n/e$, then, indeed, the maximal value of $\sum_{i = 1}^{k} \log p_i$  under the constraints in~\eqref{eq:estimations 1}
is in
\begin{equation}
\label{eq:estimations final}
\Ocomp\left(\propersize + \propersize \log \left(\frac{n}{\propersize} \right)\right) \enspace .
\end{equation}
\end{subequations}
On the other hand, if $d r \grammarsize > n/e$, then $n \leq \Ocomp(r \grammarsize)$
and the bound in the statement of the Lemma trivially holds.

The idea of preserving~\Mrefall{} is as follows: If a new power of length $\ell$ is represented,
this yields a cost of $\Ocomp(1 + \log \ell)$.
Since $\ell \geq 2$, we can treat this cost as $\Ocomp(\log \ell)$ and
choose $c$ in \Mref{3} so that this is at most $c \log \ell$.
Then either we add a new marking of size $\ell$ or we remove some marking of size $\ell'$ and add a new marking of size $\ell \cdot \ell'$.
It is easy to see that in this way~\Mrefall{} are preserved (still, those details are repeated later in the proof).

Whenever we have to represent powers $a^{\ell_1}, a^{\ell_2}, \ldots$, for each power $a^\ell$,
where $\ell > 1$, we find the last (according to preorder) maximal chain pattern $a^\ell$ in the current tree \mytree.
It is possible that this particular $a^\ell$ was obtained as a concatenation of $\ell - k$ explicit letters to $a^k$
(so, not as a power). In such a case we are lucky, as the representation of this $a_\ell$ is paid by the credit
and we do not need to separately consider the cost of representing the power $a^\ell$.
Otherwise  $a^\ell$ is a~power and we add a new marking which is contained in the subcontext of the input 
tree that is derived from the last occurrence of $a^\ell$ in the current tree.
Let $A_i$ be the smallest nonterminal that derives (before \algremblocks)
this last occurrence of the maximal chain pattern $a^\ell$ (clearly there is such non-terminal, as $A_m$ derives it).
Note that $A_i \in N_0 \cup N_1$ as otherwise this $a^\ell$ is not a~power, since powers cannot be created inside nonterminals from $\widetilde{N_0}$.
The new marking is either an $A_i$-pre-power marking or an $A_i$-in marking:
If one of the nonterminals in $A_i$'s right-hand side was removed during \algremblocks,
then we add an $A_i$-pre-power marking (note that such a removed nonterminal is necessarily from $N_1$,
as no nonterminal from $N_0 \cup \widetilde{N_0}$ is removed during \algremblocks).
Otherwise,  we add an $A_i$-in marking.

\begin{clm}
\label{clm: markings}
At any time, there is at most one $A_i$-pre-power marking in the input tree.

When an $A_i$-in marking is added because of a power $a^\ell$, then after chain compression $A_i$ has a rule of the form
\begin{itemize}
	\item $A_i \to w A_j a_\ell A_k v$, where $w$ and $v$ are (perhaps empty) sequences of handles and $A_j,A_k \in N_1$,  if $A_i \in N_1$, or
	\item $A_i \to w A_j a_\ell A_k$ where $w$ is a (perhaps empty) sequence of handles, $A_j \in N_1$, and $A_k \in N_0$, if $A_i \in N_0$.
\end{itemize}
\end{clm}
\begin{proof}
Concerning $A_i$-pre-power markings, let $a^\ell$ be the first power that causes the creation of an 
$A_i$-pre-power marking.
So one nonterminal from $N_1$ was removed from the right-hand side for $A_i$ and there is no way to reintroduce such a nonterminal.
Hence, $A_i$'s rule has at most one nonterminal from $N_1$ (when $A_i \in N_1$) or none at all (when $A_i \in N_0$).
Thus, no more powers can be created in $A_i$'s right-hand side.
In particular, neither $A_i$-pre-power markings nor $A_i$-in markings will be added in the future.

Next, suppose that an $A_i$-in marking is added to the input tree because a new power $a^\ell$ is created. Thus,
the last occurrence of the maximal chain pattern $a^\ell$ is generated by $A_i$ but not by the
nonterminals in the rule for $A_i$ (as then, a different marking would be introduced).
Since $a^\ell$ is a~power it is obtained in the rule as the concatenation of an $a$-prefix and an $a$-suffix popped from nonterminals in the rule for $A_i$.
The suffix needs to come from a nonterminal of rank $1$.
In particular this means that those two nonterminals in the rule for $A_i$ generate parts of this last occurrence of $a^\ell$
and in between them only the letter $a$ occurs.
If any of those nonterminals would be removed during the chain compression for $a^\ell$, then an $A_i$-pre-power marking
would be introduced, which is not the case.
So both nonterminals remain in the rule for $A_i$.
Hence after popping prefixes and suffixes, between those two nonterminals there is exactly a chain pattern $a^\ell$, which is then
replaced by $a_\ell$.
This yields the desired form of the rule, both in case $A_i \in N_0$ or $A_i \in N_1$.
\qedhere
\end{proof}

Consider the occurrence of $a^\ell$ and the `derived' subcontext $w^\ell$ of the \emph{input tree}.
We show that if there exists a marking inside $w^\ell$, then this marking is contained in the last occurrence of $w$ inside the occurrence
of $w^\ell$.

\begin{clm}
\label{clm: existing marking}
Let $a^\ell$ be an occurrence of a maximal chain pattern, which is replaced by $a_\ell$. Assume that $a_\ell$ derives the subcontext $w^\ell$ of the input tree,
where $w$ is a context. If this occurrence of $w^\ell$ contains a marking, then this marking is contained in the last occurrence of $w$ inside 
the occurrence of $w^\ell$.
\end{clm}

\begin{proof}
Consider a marking $M$ within $w^\ell$. Assume it was created, when some $b^k$ was replaced by $b_k$.
As $b_k$ is a single letter and $a^\ell$ derives it, each $a$ derives at least one $b_k$. Then, the marking $M$ must be
contained in the subcontext derived from the last $b_k$
(as we always create markings within the last occurrence of the chain pattern to be replaced).
Clearly the last $b_k$ can be  only derived from the last $a$ within $a^\ell$.
So in particular, the marking $M$ is contained in the last $w$ inside $w^\ell$.
So all markings within $w^\ell$ are in fact within the last $w$.
\qedhere
\end{proof}

We now demonstrate how to add markings to the input tree.
Suppose that we replace a power $a^\ell$. Note that we must have $\ell \geq 2$.
Let us consider the last occurrence of this $a^\ell$ in the current tree $\mytree$
and the smallest $A_i$ that generates this occurrence.
This $a^\ell$ generates some occurrence of $w^\ell$ (for some context $w$) in the input tree.
If this occurrence of $w^\ell$ contains no marking, then we simply add a marking
(either an $A_i$-pre-power or an $A_i$-in marking according to the above rule)
consisting of $\ell \geq 2$ arbitrarily chosen nodes within $w^\ell$.
In the other case, by Claim~\ref{clm: existing marking}, we know that all markings within
the occurrence of $w^\ell$ are contained in the last $w$.
If one of them is the (unique, by \Mref{2}) $A_i$-in marking, let us choose it. Otherwise choose any other marking in the last $w$.
Let $M$ be the chosen marking and let $\ell' = |M| \geq 2$.
We proceed,  depending on whether $M$ is the only marking in the last $w$:
\begin{itemize}
\item $M$ is the unique marking in the last $w$:
 Then we remove it and mark arbitrarily chosen $\ell \cdot \ell'$ nodes in $w^\ell$.
This is possible, as $|w| \geq \ell'$ and so $|w^\ell| \geq \ell \cdot \ell'$.
Since $\log(\ell \cdot \ell') = \log \ell + \log \ell'$, \Mref{3} is preserved,
as it is enough to account for the $1 + \log \ell \leq c \log  \ell$ representation cost for $a^\ell$
as well as the $c \log \ell'$ cost associated with the previous marking of size $\ell'$.
\item $M$ is not the unique marking in the last $w$:
 Then $|w| \geq \ell' + 2$ (the `$+2$' comes from the other markings, which are of size at least $2$, see~\Mref{1}).
We first remove the chosen marking of size $\ell'$. Let us calculate how many unmarked nodes are in $w^\ell$ afterwards:
In $w^{\ell-1}$ there are at least $(\ell-1) \cdot (\ell'+2)$ nodes and by  Claim~\ref{clm: existing marking} none of them belongs to a marking.
In the last $w$ there are at least $\ell'$ unmarked nodes (from the marking that we removed). Hence, in total we have 
\begin{align*}
(\ell-1) \cdot (\ell'+2) + \ell'
	&=
\ell \ell' + 2 \ell - \ell' - 2 + \ell'\\
	&=
\ell \ell' + 2 \ell - 2\\
	&>
\ell \ell'
\end{align*}
many unmarked nodes (recall that $\ell \geq 2$). We arbitrarily choose 
$\ell \cdot \ell'$ many unmarked nodes and add them as a new marking.
By the same argument as in the previous case, \Mref{3} is preserved.
\end{itemize}
By the above construction, \Mref{1} is preserved.
There is one remaining  issue concerning~\Mref{2}: It might be that we create an $A_i$-in marking while there already was one, violating~\Mref{2}.
However, we show that if there already exists an $A_i$-in marking, then it is within $w^\ell$ (and so within the last $w$, by Claim~\ref{clm: existing marking}).
Hence, we could choose this $A_i$-in marking as the one that is removed when the new one is created.
Consider the previous $A_i$-in marking.
It was introduced when some power $b^k$ was replaced by $b_{k}$, which, by Claim~\ref{clm: markings}, became the unique letter between the
first and second nonterminal in the right-hand side for $A_i$.
Consider the last (as usual, with respect to preorder) subpattern 
of the input tree that is either generated by the explicit letters
between nonterminals of rank $1$ in the rule for $A_i$ (when $A_i \in N_1$)
or is generated by the explicit letters below the nonterminals of rank $1$ (when $A_i \in N_0$)
(recall from the proof of Lemma~\ref{clm:cost from power to rule} that in a rule $A_i \to w A_j v A_k$ for a nonterminal of rank $0$, 
where $w$ and $v$ are sequences of handles, all letters occurring in handles in $v$
are classified as being {\em below} $A_j$).
The operations performed on \grammar{} cannot make this subpattern smaller, in fact popping letters expands it.
When $b_k$ is created, then this subpattern is generated by $b_k$,
as by Claim~\ref{clm: markings} this is the unique letter between the nonterminals (resp., below the nonterminal).
When $a_\ell$ is created, it is generated by $a_\ell$, again by Claim~\ref{clm: markings}, i.e., it is exactly $w^\ell$.
So in particular $w^\ell$ includes the $A_i$-in marking that was added when the  power $b^k$ was replaced by $b_{k}$.

This shows that \Mrefall{} hold and so also the calculations in~\eqref{eq:estimations} hold,
in particular, the representation cost of powers is $\Ocomp(\grammarsize r + \propersize \log (n/\propersize))$.  
\end{proof}

Now the estimations from Lemma~\ref{lem:cost of compression}
allow to prove Theorem~\ref{thm: main}.

\begin{proof}[Proof of the full version of Theorem~\ref{thm: main}]
Suppose first that the input tree \mytree{} has size smaller than $\propersize$.
Then by Lemma~\ref{lem: trivial estimation}, \algmain{} returns a tree of size at most $2 \propersize - 3 = \Ocomp(\grammarsize r)$.
Otherwise, consider the phase, such that before it  \mytree{} has size $s_1$ and right after it has size is $s_2$,
where $s_1 \geq \propersize > s_2$. There is such a phase as in the end \mytree{} has size $1$ and initially it has size at least $\propersize$.
Then by Lemma~\ref{lem:cost of compression} the cost of representing letters introduced till the end of this phase is
$\Ocomp\left(\grammarsize r + \propersize \log\left(\frac{n}{\propersize} \right)\right)$.
By Lemma~\ref{lem: trivial estimation} the cost of representing the remaining tree is
at most $2 \propersize - 3 = \Ocomp(\grammarsize r)$.
Hence, the size of the grammar that is returned by \algmain{}
is at most $\Ocomp\left(\grammarsize r + \propersize \log\left(\frac{n}{\propersize} \right)\right)$.
\end{proof}

\section{Conclusions}
We presented a linear-time grammar-based tree compressor 
with an approximation ratio of $\Ocomp(r + r \log (n/r \grammarsize))$,
where $n$ is the size of the input tree \mytree{}, $\grammarsize$ is the size of 
a minimal linear context-free tree grammar for \mytree, and $r$ is the maximal rank
of symbols in \mytree{}.

Possible future work is listed and discussed below.
\subsubsection*{Non-linear grammars}
One possible direction for future research are non-linear context-free tree grammars. They are defined
in the same way as linear context-free tree grammars with the only exception that parameters may occur
more than once in right-hand sides. With non-linear context-free tree grammars one can achieve double
exponential compression. For instance, the non-linear grammar with the productions 
$S \to A_1(a)$, $A_i(x) \to A_{i+1}(A_{i+1}(x))$
for $1 \leq i \leq n-1$ and $A_n(x) \to f(x,x)$ produces a perfect binary tree of height $2^n$.
The authors are not aware of any grammar-based tree compressor that exploits this additional
succinctness of non-linear context-free tree grammars.

\subsubsection*{Graph grammars}
A non-linear context-free tree grammar can be viewed as a context-free graph grammar that produces
a directed acyclic graph. This graph grammar is obtained by merging all occurrences of the same parameter
in a right-hand side.
Recently, grammar-based graph compression via context-free graph grammars was considered
in \cite{DBLP:conf/icde/ManethP16}. But no quantitative results, e.g., concerning the approximation ratio, have been shown so far. Perhaps techniques used here can
help in developing such results.

\subsubsection*{XML trees}
In contrast to trees as considered in this paper, XML trees are usually modelled using unranked (but ordered) trees,
i.e.\ the rank of a node is not determined by its label.
SLCF grammars can be used to generate such trees: we drop the assumption of the ranked alphabet,
but keep the ranks for nonterminals.
In this way, letters in SLCF grammars are \emph{de facto} ranked,
as each occurrence in the SLCF grammar has a fixed arity.
Thus, when computing such an SLCF grammar for an unranked tree,
we can artificially rank all letters and proceed as in the case of a ranked alphabet.
The approximation guarantee is the same.

There are also more powerful constructs that are intended to capture XML trees,
for instance forest algebras~\cite{forestsalgebras},
which work on forests instead of trees and allow also horizontal ``concatenation'' of trees,
and this operation yields a forest.
A corresponding grammar model can, for instance,
represent the tree $f(\underbrace{c,\ldots,c}_{n})$
by a grammar of size
$\Ocomp(\log n)$, whereas this tree is incompressible by SLCF tree grammars.
Approximation algorithms for such grammars have not been investigated so far.
Whether the methods proposed here apply in this model remains an open question.

\subsubsection*{Unordered trees}
Our method depends very little on the fact that the considered trees are ordered and it should work
also in the unordered (but ranked) case:
it is enough that leaf compression does not take the positions of leaves into the account.

\subsubsection*{Practical applications}
While \algmain{} achieves the best known approximation bound,
the authors doubt that it beats in practice the known heuristics,
especially \algofont{TreeRePair}~\cite{LohreyMM13}.
These doubts are based on the fact that in the string case, \RePair{} by far beats the
compression algorithms with the best known approximation bound.

\subsubsection*{Explicit grammar}
In many problems the SLCF grammar is given explicitly and we are interested in processing it.
The presented ``recompression'' approach can be naturally applied in this setting,
but there is no bound on the size of the SLCF obtained in this way.
However, we can use a similar trick as in the case of strings~\cite{FCPM}:
we have two alternating compression phases and in one of them we proceed
as described in Section~\ref{sec:algorithm}
while in the other we try to make the SLCF small.
The only difference is that during pair compression we choose the partition of letters
so that many occurrences of pairs in the SLCF are covered;
see~\cite{FCPM} for details.

\section*{Acknowledgements}
The first author would like to thank Pawe\l{} Gawrychowski for introducing him to the topic of compressed data,
for pointing out the relevant literature~\cite{MehlhornSU97}
and discussions, as well as
Sebastian Maneth and Stefan B\"ottcher for the question of applicability of the recompression-based approach to the tree case.

\section*{Appendix}

\begin{figure}[H]
	\centering
	\includegraphics[scale = 1.3]{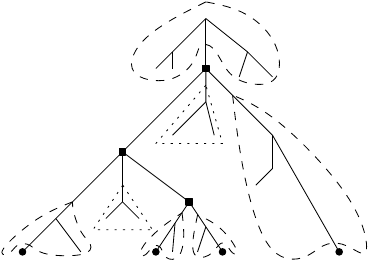}
	\caption{The partitioning of a tree depending on its parameter nodes.
		The parameters are the round nodes, the square nodes are the branching nodes of the spanning tree.
		Contexts and trees obtained after the removal of those nodes are enclosed by dashed and dotted lines, respectively.}
	\label{fig:preskeleton}
\end{figure}

\subsection*{Proof of Lemma~\ref{lem: monadic grammar} (transforming an SLCF grammar into a handle grammar).}
The proof is a slight modification of the original proof of~\cite[Theorem~10]{grammar}
and it follows exactly the same lines. Let \grammar{} be an SLCF grammar of size $g$.

The idea is as follows, see Figure~\ref{fig:preskeleton}.
Consider a nonterminal $A(y_1,\ldots,y_k)$ of \grammar{} and the tree $\eval(A)$ that it generates.
Within $\eval(A)$ take the nodes representing the parameters $y_1,\ldots,y_k$ and the spanning tree (within $\eval(A)$) for those nodes.
Consider the nodes of degree at least $2$ within this spanning tree, delete those nodes and delete the parameters.
What is left is a collection of subtrees and subcontexts.
We want to construct a grammar that has for each such subtree and subcontext
a nonterminal generating it. This is done inductively on the structure of \grammar.
As the starting nonterminal of \grammar{} has rank $0$, such a decomposition for $\eval(\grammar)$ is in fact trivial.
So, in particular the constructed grammar generates $\eval(\grammar)$.
Lastly, the construction guarantees that the introduced nonterminals, which are of rank $0$ and $1$,
are expressed through each other (plus some rules introduced on the way).
So the new grammar generates the same tree and it is monadic.
Moreover, the rules for those nonterminals are in the form required for a handle grammar, see~\HGref{2} and \HGref{3}.

Formalising this approach, we say that a \emph{skeleton tree}\footnote{Note that our definition of a skeleton slightly differs
from the one of~\cite[Lemma~1]{grammar}, but the differences are inessential.}
is a pattern from $\mathcal T(N_0 \cup N_1 \cup \letters, \mathbb Y)$, satisfying the following conditions:
\begin{enumerate}[{(SK1)}]
	\item \label{sk 1} The child of a node of degree $1$ can be labelled only with a letter of arity at least $2$ or with a parameter.
	\item \label{sk 2} If $f$ of arity at least $2$ labels a node with children $v_1, \ldots v_k$, then there are $i \neq j$ such that
	the subtrees rooted in $v_i$ and $v_j$ both contain parameters.
\end{enumerate}
Intuitively, the skeleton tree is what one obtains after replacing each context and tree in the tree constructed above with a nonterminal:
\SKref{1} says that whole context is replaced with a nonterminal, while \SKref{2} says that only branching nodes
of the spanning tree of parameters are be labelled with letters.

Our first goal is to construct for each nonterminal $A$  of the input  grammar \grammar{}  a skeleton tree $\sk_A$
together with rules for the nonterminals appearing in $\sk_A$.
These rules allow to rewrite $\sk_A$ into $\eval(A)$;
note that we do not assume that these introduced rules satisfy \SKref{1}--\SKref{2}
and on the other hand, the skeletons do not satisfy
\HGref{2}--\HGref{3} and they are not part of the constructed grammar,
they are just means of construction.
Instead, we show that in the introduced rules the nonterminals of arity $1$ occur at most $\Ocomp(\grammarsize)$ times,
while nonterminals of arity $0$ and letters occur at most $\Ocomp(r \grammarsize)$ times.

As a first step, we transform the grammar into \emph{Chomsky normal form} (CNF),
which is obtained by a straightforward decomposition of rules.
The rules in a CNF grammar are of two possible forms, where $A,B,C \in \mathbb{N}$ and $f \in \mathbb{F}$:
\begin{itemize}
	\item $A(y_1,\ldots,y_k) \to f(y_1,\ldots,y_k)$ or
	\item $A(y_1,\ldots,y_k) \to B(y_1,\ldots,y_\ell, C(y_{\ell +1} ,\ldots,y_{\ell'}), y_{\ell'+1}, \ldots, y_k)$
\end{itemize}
Note that the number of parameters can be $0$.
It is routine to check that any SLCF grammar \grammar{} of size $g$ can be transformed into an equivalent CNF grammar of size $\Ocomp(\grammarsize)$ and with 
$\Ocomp(\grammarsize)$ nonterminals~\cite[Theorem 5]{grammar}.

Given an SLCF grammar in CNF, we build bottom-up skeleton trees for its nonterminals,
During this we introduce $\Ocomp(1)$ nonterminals per considered nonterminal,
their rules have $\Ocomp(1)$ occurrences of nonterminals of arity $1$ and at most $\Ocomp(r)$ occurrences of  nonterminals of arity $0$ and constants.
Moreover, the rules for those nonterminals are in the form required by the definition of a handle grammar, see
\HGref{2} and~\HGref{3}.
All nonterminals occurring in the constructed skeletons use only those introduced nonterminals.

\begin{figure}
	\centering
	\includegraphics{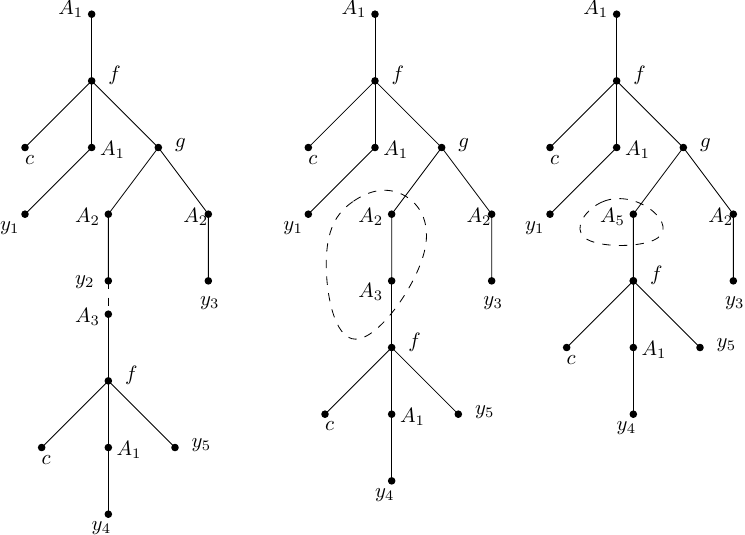}
	\caption{Calculating the skeleton when the lower skeleton is of rank at least $1$.
	On the left-most figure we see the two skeletons. After substitution, there are two neighbouring rank-$1$ nonterminals: $A_2$ and $A_3$.
	They are replaced with the nonterminal $A_5$.\label{fig:rank1}}
\end{figure}

Consider some nonterminal $A$ of the CNF grammar.
If its rule has the form $A(y_1,\ldots,y_k) \to f(y_1,\ldots,y_k)$, then $\sk_A = f(y_1,\ldots,y_k)$
and if the arity of $A$ is at most $1$ then we add $A$ and its rule to the set of constructed rules as well (if the rank of $A$ is at least $2$,
then we do not add $A$ and its rule).
This rule has the desired form~\HGref{2} or \HGref{3},
there is no nonterminal on the right-hand side and at most $1$ letter on the right-hand side.
If the rule for $A$ has the form $A(y_1,\ldots,y_k) \to B(y_1,\ldots,y_\ell, C(y_{\ell +1} ,\ldots,y_{\ell'}), y_{\ell'+1}, \ldots, y_k)$,
then we take $\sk_B$ and $\sk_C$ and replace in $\sk_B$ the parameter  $y_{\ell+1}$  by $\sk_C$,
see Figures~\ref{fig:rank1} and~\ref{fig:rank0} for two different cases. Let us denote the resulting tree with $\sk'_A$; it 
is transformed into a proper skeleton tree $\sk_A$ in the following. Let $y = y_{\ell+1}$.

Let us inspect what changes are needed, so that $\sk'_A$ satisfies \SKref{1}--\SKref{2}.
Suppose first that $C$ is of arity at least $1$, see Figure~\ref{fig:rank1}.
It might be that the root node of $\sk_C$ and the node above the leaf $y$ in $\sk_B$ are both of arity $1$,
without loss of generality assume that their labels are nonterminals $C'$ and $B'$ (the case of letters follows in the same way).
We then introduce a new nonterminal $A'$ of rank $1$ and replace the subpattern $B'(C')$ in $\sk'_A$ with $A'$
and add a rule $A' \to B' C'$. Note that it is in a form required by~\HGref{2}.
We claim that the resulting tree satisfies  \SKref{1} and \SKref{2} and hence
can be taken for $\sk_A$.

Since the node above $B'$ and the node below $C'$ are not of degree $1$
(by induction assumption on \SKref{1}), \SKref{1} is satisfied.
Concerning~\SKref{2}, take any node of arity at most $2$ in $\sk'_A$.
Any node labelled with a letter in $\sk_C$ has the same subtrees in $\sk_C$ and in $\sk'_A$,
so \SKref{2} holds for them.
For the nodes in $\sk_B$ the only problem can arise for nodes that had the replaced $y$ in some of their subtrees.
However, as $y$ is replaced with $\sk_C$, which has a parameter, the condition is preserved for them.

\begin{figure}
	\centering
		\includegraphics{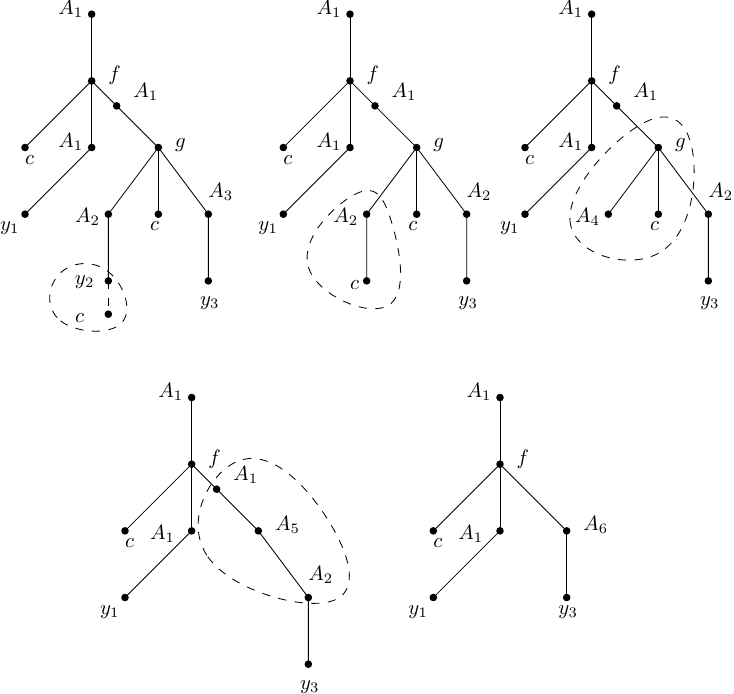}
	\caption{Calculating the skeleton, when the lower skeleton is of rank $0$,
	going from the left-top to bottom-right.
	In the first picture there are two skeletons, in the second they are substituted into each other.
	In the third we replace $A_2 c$ by $A_4$ with the rule $A_4 \to A_2 c$.
	In the fourth we replace $g(A_4, c, \cdot)$ by $A_5$ with the rule $A_5(y) \to g(A_4, c, y)$
	and finally we replace $A_1 A_5 A_2$ by $A_6$ with the rule $A_6 \to A_1 A_5 A_2$, which can be split into two rules.}
	\label{fig:rank0}
\end{figure}

Suppose now that $C$ is of arity $0$, see Figure~\ref{fig:rank0}.
Then the skeleton $\sk_C$ has no parameters, which implies that it
is either a constant or a nonterminal of arity $0$ ($\sk_C$ cannot use letters of arity larger than $1$
by~\SKref{2}, and cannot use nonterminals and letters of arity $1$, as their children need to be labelled with letters of arity at least $2$).
We only consider the former case (the same argument hold for the latter case). Let $\sk_C$ be the constant $c$.
Firstly, the node above $y$ (the parameter which is replaced by $\sk_C = c$) in $\sk_B$ can be a node of arity $1$. Without loss of generality suppose that it is a nonterminal
$B'$ (the case of unary letter follows in the same way).
We introduce a fresh nonterminal $A'$ of arity $0$, replace the subtree $B'(c)$ by $A'$ and introduce the rule $A' \to B' c$.
The rule is of the form required by~\HGref{3}.
For uniformity, if the node above $y$ is not of arity $1$, introduce $A'$ with the rule $A' \to c$ and replace $c$ by $A'$.
Condition \SKref{1} now holds.

Concerning~\SKref{2},
consider the parent node $v$ of $A'$. Either it does not exist, in which case we are done (as $\sk_A = A'$)
or it is labelled with a letter $f$ of arity at least $2$.
All other nodes in $\sk'_A$ labelled with letters of arity at least $2$ satisfy~\SKref{2}, as the subtree rooted at $v$ still contains at least one parameter.
Focusing on the $f$-labelled node $v$, if it still has at least two children with parameters in their subtrees, then we are done,
as \SKref{2} is satisfied for $v$.
If not, then exactly one of $v$'s children is a subtree with a parameter. Without loss of generality let it be $v$'s first child,
all other children are constants or nonterminals of arity $0$.
So let the children of $v$ (except  for the first one) be labelled with $\gamma_2, \gamma_3, \ldots, \gamma_\ell$,
where each $\gamma_i$ is either a constant or a nonterminal of rank $0$.
Introduce a new nonterminal $A'$ of rank $1$ with the rule $A' \to f(y_1, \gamma_2, \gamma_3,\ldots,\gamma_\ell)$
and replace the subpattern $f(t_1,  \gamma_2, \gamma_3,\ldots,\gamma_\ell)$ with $A'(t_1)$ (where, $t_1$ is the subtree rooted
at the first child of $v$).
Observe that the rule for $A'$ is of size  $\ell \leq r$ and is of the form~\HGref{2}.
Lastly, now again~\SKref{1} can be violated, because the parent node or the child (or both) of the $A'$-labelled node can be of degree $1$.
This can be fixed by replacing those $2$ or $3$ nodes of degree $1$ by one nonterminal of rank $1$.
This requires adding at most $2$ rules for nonterminals of arity $1$ of the required form~\HGref{2}.

We constructed $\Ocomp(1)$ rules of the form~\HGref{2} or~\HGref{3} per nonterminal of the CNF grammar
(there are $\Ocomp(g)$ of them), each of them has at most $r$ occurrences of letters and nonterminals of arity $0$
and at most $2$ of nonterminals of arity $1$.
By a routine calculation it can be shown that $\eval(\sk_A) = \eval(A)$. If $S$ is the start nonterminal of the CNF grammar,
then $\sk_S$ has no parameters and hence is either a constant (this case is of course trivial) or a nonterminal of rank $0$,
which is the start nonterminal of our output grammar.

Concerning the efficiency of the construction, the proof follows in the same way as in~\cite[Theorem~10]{grammar}:
It is enough to observe that $\sk_A$, where $A$ has rank $k$, has at most $2r(k-1) + 2$ nodes:
By \SKref{1} nodes of arity $1$ constitute at most half of all nodes.
Secondly,  as it has $k$ parameters, it has at most $k-1$ nodes of arity larger than $1$, so at most $(k-1)(r-1)+1$ leaves.
Summing up yields the claim.
\qed

\end{document}